\documentclass[onecolumn, 12pt, draftclsnofoot, comsoc]{IEEEtran}

\ifCLASSINFOpdf
\else
   \usepackage[dvips]{graphicx}
\fi
\usepackage{url}
\usepackage{hhline}
\usepackage{graphicx}
\usepackage{multicol}
\usepackage{graphicx}
\usepackage{amssymb}
\usepackage{amsmath}

\usepackage{calc}
\usepackage{array}
\usepackage{optidef}
\usepackage{epstopdf}
\usepackage{mathtools}
\usepackage{mathrsfs}
\usepackage{bm}
\usepackage{cite}
\usepackage{boldline}
\usepackage{amsthm}
\usepackage{algorithm}
\usepackage{algorithmicx, algpseudocode}
\usepackage{amsthm}

\newtheorem{remark}{Remark}

\newtheorem{lemma}{Lemma}

\usepackage{amsmath}

\DeclareMathOperator*{\minimize}{minimize}

\usepackage{booktabs,makecell,tabularx}

\newcolumntype{C}[1]{>{\centering\arraybackslash}p{#1}}
\newcolumntype{L}{>{\raggedright\arraybackslash}X}

\usepackage{siunitx}
\usepackage{etoolbox}
\newrobustcmd{\B}{\bfseries}

\usepackage{comment}
\usepackage{algorithm}% http://ctan.org/pkg/algorithms
\usepackage{algpseudocode}% http://ctan.org/pkg/algorithmicx
\usepackage{enumitem} % for '\newlist' and '\setlist' macros

\usepackage{colortbl}     % preamble
\usepackage{subfigure}
\usepackage{color, soul}
\usepackage{stfloats}
\usepackage{float}
\usepackage{multirow}
\usepackage{wrapfig}
\usepackage{booktabs}
\usepackage{slashbox}
\usepackage{optidef}
\definecolor{LightBlue}{rgb}{0.75,0.936,1.00}
\sethlcolor{LightBlue}
\definecolor{LightCyan}{rgb}{0.88,1,1}

\usepackage{xcolor}

\begin{document}
\bstctlcite{IEEEexample:BSTcontrol}
\title{Rate-Splitting for Joint Unicast and Multicast Transmission in LEO Satellite Networks \\ with Non-Uniform Traffic Demand}
\author{Jaehyup Seong, Juha Park, Dong-Hyun Jung, Jeonghun Park, and Wonjae Shin \vspace{-10mm}
    \thanks{J. Seong, J. Park, and W. Shin are with the School of Electrical Engineering, Korea University, Seoul 02841, South Korea 
    (email: {\texttt{\{jaehyup, juha, wjshin\}@korea.ac.kr}});
    D. Jung is with the Radio and Satellite Research Division, Electronics and Telecommunications Research Institute, Daejeon 34129, South Korea (email: {\texttt{dhjung@etri.re.kr}});
    J. Park is with the School of Electrical and Electronic Engineering, Yonsei University, Seoul 03722, South Korea (email: {\texttt{jhpark@yonsei.ac.kr}}).
    (\textit{Corresponding author: Wonjae Shin.})}} 
%Robust Rate-Matching Framework for Multibeam Satellite Communications with Phase Perturbations
%RSMA-Based Rate Matching for \\ Non-Orthogonal Unicast and Multicast \\ Transmission in LEO Satellite Networks
%\markboth{Submitted to IEEE Journal on Selected Areas in Communications}
\maketitle
\vspace{-2mm}
\begin{abstract}\vspace{-1.5mm}
Low Earth orbit (LEO) satellite communications (SATCOM) with ubiquitous {global} connectivity is deemed a pivotal catalyst in advancing wireless communication systems for {5G and beyond}.
LEO SATCOM excels in delivering versatile information services across expansive areas, facilitating both unicast and multicast transmissions via high-speed broadband capability.
Nonetheless, given the broadband coverage of LEO SATCOM, traffic demand distribution within the service area {is non-uniform}, and the time/frequency/power resources available at LEO satellites remain significantly limited. 
%
%Thus, it is required at LEO satellites to satisfy various requirements of data traffic using limited available resources.
%
Motivated by {these challenges}, we propose a rate-matching framework for non-orthogonal unicast and multicast (NOUM) transmission. {Our approach aims to minimize} the difference between offered rates and traffic demands for both unicast and multicast messages.
{By} multiplexing unicast and multicast transmissions over the same radio resource, rate-splitting multiple access (RSMA) is employed to manage interference between unicast and multicast streams{, as well as} inter-user interference under imperfect channel state information at the LEO satellite.
%
%To solve the non-convex formulated problem, we represent the rate expressions as Rayleigh quotients form and the common portions as a fractional form, converting it to the unconstrained problem. 
%To solve the non-convex formulated problem, we convert it to an unconstrained problem by expressing the rate expression in the Rayleigh quotient form and the common portion in fraction form.
%
%To solve the non-convex formulated problem, the common rate is approximated using the LogSumExp technique; thereafter, we represent the common portion as the ratio of the approximated function, converting the problem into an unconstrained form. 
To address the formulated problem's non-smoothness and non-convexity, the common rate is approximated using the LogSumExp technique{. Thereafter,} we represent the common rate portion as the ratio of the approximated function, converting the problem into an unconstrained form. 
A generalized power iteration (GPI)-based algorithm, coined \emph{GPI-RS-NOUM}, is proposed upon this reformulation. 
Through comprehensive numerical analysis across diverse simulation setups, we {demonstrate}  that the proposed framework outperforms various benchmarks for LEO SATCOM with uneven traffic demands.

%we show that the common message of RSMA plays a triple role as 1) interference management among unicast messages, 2) multicast message transmission, and 3) robustness over imperfect CSIT. This leads us to conclude that an RSMA is a powerful multiple-access technique for NOUM transmission in LEO SATCOM.

%we show that the proposed framework outperforms various benchmarks for LEO SATCOM with uneven traffic demands. 

\end{abstract}

\begin{IEEEkeywords}\vspace{-1.5mm}
NOUM transmission, LEO SATCOM, rate-matching, RSMA, heterogeneous traffic demands.
\end{IEEEkeywords}

\IEEEpeerreviewmaketitle
%\vspace{-2mm}
\section{Introduction}
With the explosive {growth} in wireless applications, the demand for high throughput, content-oriented service, and heterogeneous service types is continuously increasing \cite{zhong2018traffic}. 
%In recent years, low earth orbit (LEO) satellite communications (SATCOM) has been much of interest thanks to its capability to ensure broadband services over a wide coverage area with massive connectivity and low latency \cite{perez2019signal}. 
%LEO SATCOM can not only provide truly ubiquitous connectivity in remote areas but also decongest dense urban areas, providing resilient communication services with high throughput.
%Moreover, it facilitates providing rich multicast services, such as video streaming, live broadcasting, corporate data distribution, and disaster alert messages, across extensive coverage areas \cite{kawamoto2014prospects}. 
In recent years, low Earth orbit (LEO) satellite communications (SATCOM) {have garnered} significant attention for {their} capability to deliver high-speed broadband services with low latency across expansive coverage areas \cite{perez2019signal}. 
%LEO SATCOM can not only provide truly ubiquitous connectivity in remote areas but also decongest dense urban areas, providing resilient communication services. It also facilitates providing rich multicast services, such as video streaming and live broadcasting across extensive coverage areas 
LEO SATCOM {not only provides} ubiquitous {global} connectivity in remote regions but also alleviates congestion in dense urban areas, ensuring resilient communication services. Additionally, it facilitates the delivery of diverse multicast services, including video streaming, live broadcasting, and disaster alert messages, across extensive service areas \cite{kawamoto2014prospects}. 
Therefore, LEO SATCOM excels in {reliably} providing various information services, such as unicast and multicast services. In LEO SATCOM, however, available time/frequency resources are highly limited, and the number of users within the target service area ($\num{100}-\num{1000}$ km) is much larger compared to terrestrial networks (up to $\num 10$ km).
{Consequently}, providing unicast and multicast services with {distinct} time/frequency blocks cannot {adequately address} the dramatically growing demands of
fifth-generation (5G) and {beyond} wireless communications.

% 비단 위성에만 적용되던게 아니다.
%  with multi-user linear precoding (MU-LP) 
%In NOUM transmission, the multicast and unicast streams are superimposed by layered division multiplexing (LDM) at the transmitter. At the receiver, the multicast stream is first decoded considering the unicast streams as noise based on the principle of LDM \cite{kim2008superposition}.
%The decoded multicast stream is removed from the received signal through successive interference cancellation (SIC), and then the intended unicast stream is decoded by treating the other unicast streams as noise. 
%\vspace{-2mm}
\subsection{Related Works}

Non-orthogonal unicast and multicast (NOUM) transmission, {which} simultaneously offers both unicast and multicast services with the same resource blocks, is {expected} to play a vital role in LEO SATCOM.  
In NOUM transmission, the multicast and unicast streams are superimposed by layered division multiplexing (LDM) {at the transmitter}. At the receiver, the multicast stream is first decoded {by treating} the unicast streams as noise based on the principle of LDM \cite{kim2008superposition}.
The decoded multicast stream is {then} removed from the received signal {using} the successive interference cancellation (SIC) technique, and then the intended unicast stream is decoded by treating the other unicast streams as noise. 
%In integrated satellite-terrestrial networks (ISTN), it has been verified that NOUM transmission outperforms orthogonal unicast and multicast (OUM) transmission, such as time division multiplexing (TDM) in minimizing transmit power
In integrated satellite-terrestrial networks (ISTN), NOUM transmission has been shown to outperform orthogonal unicast and multicast (OUM) transmission, such as time division multiplexing (TDM), in minimizing transmit power \cite{peng2022non}.
Nevertheless, fully reusing available resources for both services can cause severe interference between the multicast and unicast messages, along with interference among the unicast messages. Accurate channel state information (CSI) at the transmitter (CSIT) is essential to mitigate {interference issues} with conventional multiple access techniques, such as spatial division multiple access (SDMA) and non-orthogonal multiple access (NOMA).
In LEO SATCOM, however, obtaining precise CSIT is usually infeasible due to long propagation delay and rapid movement of satellites \cite{vazquez2016precoding, you2020massive}.
To manage interference issues under imperfect CSIT conditions {in} LEO SATCOM, an advanced multiple access technique for NOUM transmission is required.

%Accurate channel state information (CSI) at the transmitter (CSIT) is essential to mitigate the interference issues with conventional multiple access, such as spatial division multiple access (SDMA) and non-orthogonal multiple access (NOMA).
%In LEO SATCOM, however, obtaining precise CSIT is usually infeasible due to long propagation delay and rapid movement of satellites \cite{you2020massive}. 

%Moreover, the effectiveness of SDMA severely depends on the quality of CSIT, which is difficult to obtain perfectly at the satellite.
%In order to manage interference issues under imperfect CSIT conditions, an advanced multiple access technique for NOUM transmission is required.   

%To tackle such limitation of SDMA-based NOUM transmission, the authors in \cite{mao2019rate} have investigated the application of rate-splitting multiple access (RSMA) to NOUM transmission in terrestrial networks. 
%Rate-splitting multiple access (RSMA) has been recognized as one of the promising solutions to overcome interference issues under imperfect CSIT.

{Rate-splitting multiple access (RSMA) stands out as a promising solution to address interference issues with imperfect CSIT.} {RSMA has been shown to} ensure robustness over imperfect CSI conditions{, as well as,} high spectral utilization and energy consumption efficiency in {various} network scenarios and propagation conditions \cite{mao2018rate, clerckx2016rate, mao2018energy, park2023rate}.
RSMA can {encompass} conventional multiple access {techniques,} such as SDMA, NOMA, and multicasting, as special cases by adjusting the portion of common and private parts. Thanks to its flexibility and generality, RSMA has great potential for {enhancing} rate and quality of service (QoS) \cite{clerckx2019rate}.
{Combined with multibeam SATCOM,
{various} RSMA-based unicast or multicast beamforming strategies have been {developed to enhance} sum-rate maximization or max-min fairness (MMF) \cite{yin2020rate, yin2020rate_J, yin2021ratephy, si2022rate, khan2023rate, xu2023distributed}}. 

{Inspired by the advantages of RSMA in unicast and multicast transmission, the applications of RSMA in NOUM transmission have been studied \cite{mao2019rate, li2023non, li2023cooperative, han2024joint}.} RSMA-based NOUM transmission {has been verified to outperform} SDMA- and NOMA-based NOUM transmission {in terms of} sum-rate and energy efficiency \cite{mao2019rate}.  
{An additional advantage of RSMA-based NOUM transmission is that it can be realized by incorporating a multicast message into the common stream \cite{mao2019rate}.}
Therefore, the SIC layer of RSMA can be used not only {to manage} the interference among the unicast messages but also {to separate}  the multicast and unicast messages. 
{It is important to note} that the SIC structure of RSMA does not increase complexity for the receivers compared to conventional NOUM transmission{, as the SIC layer is required} for separating the multicast and unicast messages \cite{mao2019rate}. 
{Leveraging the numerous advantages of RSMA for NOUM transmission, RSMA-based NOUM transmission has found widespread applications in LEO SATCOM systems \cite{li2023non, li2023cooperative, han2024joint}.} 
%Inspired by this, RSMA has begun to be applied to NOUM transmission in LEO SATCOM \cite{li2023non, li2023cooperative, han2024joint}. 
The authors of \cite{li2023non} have focused on maximizing the sum of the unicast rates {while meeting} the QoS constraint of multicast rate in LEO SATCOM with perfect CSIT. 
The problem of maximizing the minimum unicast rate while satisfying multicast traffic demand under the QoS constraint has been investigated in ISTN with perfect CSIT at both the base station and LEO satellite \cite{li2023cooperative, han2024joint}.

%The problem of maximizing the minimum unicast rate while satisfying multicast traffic demand under the QoS constraint has been investigated in ISTN \cite{li2023cooperative, han2024joint}. The study in \cite{li2023cooperative, han2024joint} has been operated under the assumption of perfect CSIT at both the LEO satellite and the base station.

%From \cite{li2023non, li2023cooperative, han2024joint}, a superiority of RSMA-based NOUM transmission compared to OUM transmission and SDMA- and NOMA-based NOUM transmission in terms of SR maximization or MMF has been verified. 
%The additional advantage of RSMA in NOUM transmission is it can be realized by incorporating a multicast message into the common message \cite{mao2019rate}. Thus, the SIC layer of RSMA can be used for managing the interference among the unicast streams as well as separating the multicast and unicast streams. 
%\vspace{-2mm}
\subsection{Motivations and Contributions}

While {significant} research efforts have been invested in designing NOUM transmissions for LEO SATCOM under fairly homogeneous traffic conditions \cite{li2023non, li2023cooperative, han2024joint}, less attention has been dedicated to the heterogeneity of traffic demands over time and geographical locations.
%It is noted that the heterogeneity of unicast traffic demands has not been taken into account in \cite{li2023non, li2023cooperative, han2024joint}. 
%The traffic demand distribution from users within the service area tends to be potentially asymmetric owing to the capability of satellites to provide broadband coverage 
{The traffic distribution within the service area of LEO SATCOM tends to be highly asymmetric due to its broadband coverage capability \cite{lizarraga2014flexibility}.}
{Considering such heterogeneity of traffic demands, the precoder designed to maximize the minimum or sum unicast rate can {lead to} \emph{unmet rate} (i.e., the amount of unsatisfied rate to traffic demand) and \emph{unused rate} (i.e., the amount of exceeded rate to traffic demand).}
{This results in a significant degradation of the unicast service reliability.}
Therefore, a new performance metric for precoder design is required to effectively reduce such unmet and unused rates, {especially given a power-hungry payload at the LEO satellite.}
Moreover, obtaining precise CSIT at the LEO satellite within a coherence time is {quite} challenging due to the significant end-to-end propagation delay and the {rapid movement of LEO satellites}  \cite{vazquez2016precoding, you2020massive}. 
%As such, the imperfect CSIT at the LEO satellite that has not been regarded in \cite{li2023non, li2023cooperative, han2024joint} also needs to be carefully considered. 
As such, it is necessary to carefully consider the impact of imperfect CSIT conditions at the LEO satellite, a facet yet to be {addressed} in previous studies \cite{li2023non, li2023cooperative, han2024joint}, on precoder design.

%\tcr{MOTIVATION SHOULD BE FURTHER ADDED}
%It does not bring any complexity increase for the receivers compared to conventional NOUM transmission since the SIC layer has been already required for separating the multicast and unicast streams \cite{mao2019rate}. 

%Motivated by these, we put forth an RSMA-based rate-matching (RM) framework for NOUM transmission that flexibly designs precoder according to unicast and multicast traffic demands under imperfect CSIT at the LEO satellite. 
{Motivated by these challenges}, we put forth an RSMA-based rate-matching (RM) framework, which minimizes the difference between {traffic demands and offered rates} for both unicast and multicast messages, under imperfect CSIT.
%The proposed framework designs the RSMA precoder to minimize the difference between the traffic demands and offered rate for both unicast and multicast messages.
{This framework enables the LEO satellites to stably fulfill} uneven unicast traffic demands and effectively provides {the} intended multicast message with limited available power.
%By doing so, unmet/unused rates for both unicast and multicast messages are minimized effectively with the limited available power from the LEO satellites.
{Our} key contributions are summarized as follows: 
\begin{itemize}
\item  We propose an RSMA-based RM framework that minimizes the difference between traffic demands and actual offered rates for both unicast and multicast messages. 
By flexibly allocating the usable power into the common and private streams according to the traffic requirements, the unused/unmet rates are {effectively minimized.}
{To cope with the challenge of obtaining instantaneous CSIT at LEO satellites, we leverage the statistical and geometrical information of satellite-to-user channels, which vary comparably slower, in the RSMA precoder design.}
%Furthermore, to consider more realistic scenarios in LEO SATCOM, we assume that only statistical CSIT (sCSIT), which varies comparably slower than instantaneous CSIT (iCSIT) is enabled to be obtained at the LEO satellite.
%Furthermore, given the challenge of acquiring instantaneous channel state information (iCSI) at the LEO satellite, we exploit the statistical information of channel and geometry in the RSMA precoder design between the satellite and the users that varies comparably slower than iCSIT.

\item %Third, to solve the formulated non-smooth and non-convex problem efficiently, 
%we propose a generalized power iteration-based rate-splitting for NOUM transmission (GPI-RS-NOUM) algorithm. 
To jointly find an optimal precoding vector and common rate portions of each unicast/multicast message, we propose a generalized power iteration (GPI)-based rate-splitting for NOUM transmission (GPI-RS-NOUM) algorithm. Specifically, to tackle the non-smoothness from the minimum function raised by the common rate, we approximate the minimum function with the LogSumExp technique, making it differentiable. To {address} the non-convexity caused by multiple constraints upon the common rate, we represent the common rate portions as a ratio of the approximated function, thereby converting the formulated problem into an unconstrained form. 
{We then express the first-order Karush-Kuhn-Tucker (KKT) condition of the reformulated problem as an eigenvector-dependent nonlinear eigenvalue problem (NEPv).
By applying the principle of conventional power iteration, we propose the GPI-RS-NOUM algorithm that efficiently computes a principal eigenvector of the expressed NEPv.}
%By interpreting the problem as a class of functional generalized eigenvalue problems, we propose an efficient algorithm, which is referred to GPI-RS-NOUM algorithm, 
%\tcr{Leveraging this reformulation, we derive the first-order Karush-Kuhn-Tucker (KKT) condition to find stationary points. The GPI-RS-NOUM algorithm is proposed to efficiently identify an opi point among the stationary points.}
%2) 제안하는 공식화된 문제의 해를 찾기 위해 -> 프리코더 하고 공통 비율이 있음을 명시 (이걸 효율적으로 찾기 위해) -> 어떻게 할건지 -> main contribution 만 (다 말할수는 없으니까) (common rate로 인해 multiple constraint 가 생기는구나 도 간접적으로 명시)
\item 
%We numerically show the superiority of the proposed GPI-RS-NOUM algorithm in terms of traffic demand satisfaction for unicast and multicast messages compared to benchmarks. 
%In other words, the numerical results demonstrate the effectiveness of the proposed GPI-RS-NOUM algorithm in minimizing the gap between traffic demand and actual offered rates for both unicast and multicast messages, outperforming benchmark methods.
We numerically {demonstrate the effectiveness of the proposed  GPI-RS-NOUM algorithm in} minimizing the difference between actual offered rates and traffic demands for both unicast and multicast messages.  
{Through} comprehensive numerical analysis spanning diverse LEO SATCOM scenarios, including variations in traffic distributions, user locations, scattering conditions, and the number of transmit antennas, the superiority of the proposed framework over several benchmarks is {shown}.
{ We demonstrate that the common stream plays crucial triple functions: i) managing interference between unicast and multicast streams and inter-user interference, ii) enabling multicast stream transmission, and iii) ensuring robustness against imperfect CSIT.} This leads us to conclude that RSMA is a formidable multiple access technique for NOUM transmission in LEO SATCOM.
\end{itemize}

%\vspace{-3mm}
\subsection{Notations}    
The notations employ standard letters for scalars, lower-case boldface letters for vectors, and upper-case boldface letters for matrices.
% 복수 - 복수
The matrix $\mathbf{A}={\sf{blkdiag}}(\mathbf{A}_{1}, \cdots, \mathbf{A}_{K})\in\mathbb{C}^{NK \times NK}$ represents the block-diagonal matrix composed of $\mathbf{A}_{1}, \cdots, \mathbf{A}_{K}\in\mathbb{C}^{N \times N}$.
Notations $(\cdot)^{\sf{T}}$,  $(\cdot)^{\sf{H}}$, $(\cdot)^{-1}$, $\mathbb{E}[\cdot]$, and $\exp\{\cdot\}$ identify the transpose, conjugate transpose, matrix inversion, expectation, and exponential operators, respectively.
%Furthermore, $\mathbf{I}$ and $\mathbf{E}_{k}$ denote the identity matrix and the diagonal matrix in which the $(k, k)$-th diagonal element is set to be $1$ and otherwise $0$.
{Additionally, $\mathbf{I}$ denotes the identity matrix, with its size determined by the dimension of the matrix it operates on.}

\section{System Model}

We consider an LEO SATCOM system in which the LEO satellite is equipped with uniform planar arrays (UPAs)
{consisting of}  $N_{\sf{t}} \triangleq N_{\sf{t}}^x \times N_{\sf{t}}^y$ antennas. Herein, $N_{\sf{t}}^x$ and $N_{\sf{t}}^y$ {represent} the number of antennas on the $x$-axis and $y$-axis, respectively.
%with $N_t^x$ array elements along the $x$-axis and $N_t^y$ array elements along the $y$-axis.
%In other words, the LEO satellite is equipped with a total of $N_t$ array antennas, as $N_t \triangleq N_t^x \times N_t^y$.
Within the coverage area, there are $K$ users{, each} equipped with a single antenna (indexed by $\mathcal{K} \triangleq \{ 1,\cdots, K \}$). {Each user desires}   to receive not only a multicast message (intended for all users) but also a dedicated unicast message. 
{The LEO satellite provides both types of messages {using the same} time/frequency resources. The traffic demands for each unicast message are heterogeneous, as {illustrated} in Fig. \ref{Fig1}. It is assumed that the traffic demands for both unicast and multicast messages are perfectly known at the LEO satellite.}

%Given the significant propagation delay and the high mobility, we presume that only sCSIT is available for the LEO satellite.
%Given the significant propagation delay and high mobility of LEO satellites, we presume that geometry between the satellite with users and statistical information of complex channel gain are available at the LEO satellite \cite{you2020massive, li2021downlink, you2022hybrid, you2022beam}.

%We consider an LEO SATCOM system in which the LEO satellite is equipped with uniform planar arrays (UPAs) with $N_t^x$ and $N_t^y$ array elements in the $x$-axis and $y$-axis, respectively.
%In other words, the LEO satellite is equipped with $N_t$ antennas as $N_t \triangleq N_t^x \times N_t^y$. 
%Within the coverage area, $K$ users equipped with a single antenna, indexed by $\mathcal{K} \triangleq \{ 1,\cdots, K \}$, require multicast message and dedicated unicast message.
%The demand for unicast service types is heterogeneous among users as illustrated in Fig. \ref{Fig1}. 
%We assume that the LEO satellite is enabled to know sCSIT alone which varies comparably slower than the iCSIT due to the long propagation delay and the high mobility of LEO satellites. The traffic demands of unicast and multicast messages are assumed to be perfectly known at the LEO satellite.

\begin{figure}[!t]
\centering
 \includegraphics[width=0.9\linewidth]{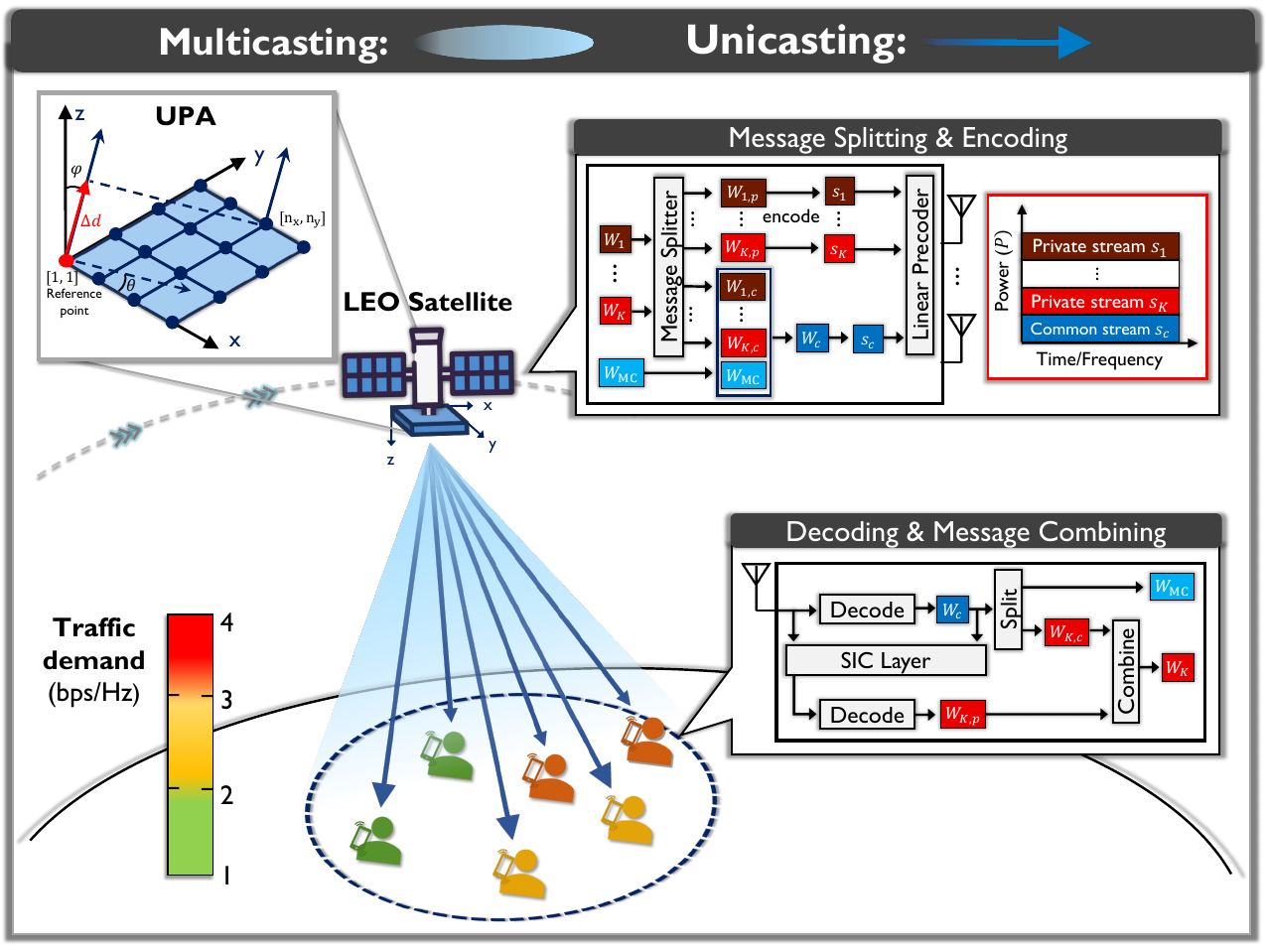}
 		\caption{System model of the proposed RSMA-based NOUM transmission.}
    	\label{Fig1}\vspace{-2mm}
\end{figure}

%\vspace{-2mm}
\subsection{Satellite Channel Model}

%The LEO SATCOM system, where the scattering condition is stationary within the interest of time interval, is considered. 
% Path attenuation과 delay가 time에 따라 변하지 않도록 = scattering condition is stationary within the interest of time interval (원래는 시간마다 주변 환경이 바뀌는데)
To characterize a satellite channel, a widely adopted ray-tracing-based model is employed. 
%By assuming stationary scattering condition within the interest of time interval, (어짜피 이 한순간 보는거니까 -> 물리적인거니까 이 한 순간에는 고정이지)
The received baseband signal at the $k$-th user from the $(n_{x}, n_{y})$-th antenna, {in the absence of noise, can be expressed as}
\begin{align}
    \label{ch_1}
     y_{k}^{[n_{x}, n_{y}]}(t) = 
     & \sum_{l=0}^{L_{k}-1} 
     \alpha_{k, l} \cdot x(t-\tau_{k,l}^{[n_{x}, n_{y}]})
     \nonumber \\
     & \times \exp(-j2\pi f_{\sf{c}}\tau_{k,l}^{[n_{x}, n_{y}]}) \cdot \exp(j2\pi f^{\sf{D}}_{k, l}t),
\end{align}
where $L_{k}$, $\alpha_{k, l}$, and $f_{\sf{c}}$ {denote} the number of propagation paths, path attenuation over the $l$-th path, and a carrier frequency, respectively. 
$\tau_{k,l}^{[n_{x}, n_{y}]}$ is the propagation delay from the $(n_{x}, n_{y})$-th antenna to the $k$-th user over the $l$-th path. $x(t-\tau_{k,l}^{[n_{x}, n_{y}]})$ is the transmitted signal from the $(n_{x}, n_{y})$-th antenna delayed by $\tau_{k,l}^{[n_{x}, n_{y}]}$. $f^{\sf{D}}_{k, l}$ indicates the Doppler shift over the $l$-th propagation path. Without loss of generality, the propagation delays are sorted in the ascending order, {such that} $\tau_{k, 0}^{[n_{x}, n_{y}]} \leq \cdots \leq \tau_{k, L_{k}-1}^{[n_{x}, n_{y}]}$.

{Due to the high altitude of the LEO satellite}, LEO SATCOM {typically operates under channel conditions  dominated by} line-of-sight (LOS) components. This feature {results in a much smaller delay spread} compared to conventional terrestrial networks, as measured in \cite{3gpp_dealy, vojcic1994performance}. 
%Moreover, since the delay spread becomes much smaller than the symbol time duration in the narrowband mm-wave communications, multipath 으로 인한 delay spread 영향이 매우 작은거지
Therefore, $x(t-\tau_{k,l}^{[n_{x}, n_{y}]})$ can be {approximated as} $x(t-\tau_{k,l}^{[n_{x}, n_{y}]}) \approx x(t-\tau_{k,0}^{[n_{x}, n_{y}]})$, $\forall l$, reformulating (\ref{ch_1}) into 
%\begin{align}
%    \label{ch_2}
%    & y_{k}^{[n_{x}, n_{y}]}(t) = 
%     \nonumber \\
%     & x(t) \times \sum_{l=0}^{L_{k}-1} 
%     \alpha_{k, l} \cdot
%     \exp(-j2\pi f_{\sf{c}}\tau_{k,l}^{[n_{x}, n_{y}]}) \cdot
%     \exp(j2\pi f^{\sf{D}}_{k, l}t),
%\end{align}
\begin{align}
    \label{ch_2}
    y_{k}^{[n_{x}, n_{y}]}(t) = h_{k}^{[n_{x}, n_{y}]}(t) \cdot x(t-\tau_{k,0}^{[n_{x}, n_{y}]}),
\end{align}
where $h_{k}^{[n_{x}, n_{y}]}(t)$ is the channel between the $k$-th user and the $(n_{x}, n_{y})$-th antenna, given by
\begin{align}
    \label{ch_3}
     h_{k}^{[n_{x}, n_{y}]}(t) \! = \!\!
     &\sum_{l=0}^{L_{k}-1} \! 
     \alpha_{k, l} \cdot
     \exp\left(-j2\pi \{ f_{\sf{c}}\tau_{k,l}^{[n_{x}, n_{y}]} \! - \! f^{\sf{D}}_{k, l}t \} \right).
\end{align}

%\textbf{1) Delay:} 

In (\ref{ch_3}), the propagation delay from the $(n_{x}, n_{y})$-th antenna to the $k$-th user over the $l$-th path can be rewritten as
\begin{align}
    \label{ch_3.5}
    \tau_{k,l}^{[n_{x}, n_{y}]}
    = \tau_{k,l}^{[1, 1]} + \Delta \tau^{[n_{x}, n_{y}]}(\theta_{k, l}, \varphi_{k, l}),
\end{align}
% steering vector를 표현하기 위함
% 제일 먼저 신호가 도달한 안테나를 기준으로 하는게 아니라, 임의로 설정하는거임: 그래서 delta tau 가 음수가 나올수도 있음, 그리고 tau^[ref] 에 propagation delay 뿐만이 아닌 추가적인 delay가 들어갈수 있음 
% Note: Multipath로 인한 delay의 영향은 살아있음
%where $\tau_{k,l}^{[1, 1]}$ is the  propagation delay between the $k$-th user and the reference point over the $l$-th path as shown in Fig. \ref{Fig1}.
where $\tau_{k,l}^{[1, 1]}$ is the  propagation delay from the reference point to the $k$-th user over the $l$-th path as shown in Fig. \ref{Fig1}.
{$\Delta \tau^{[n_{x}, n_{y}]}(\theta_{k, l}, \varphi_{k, l})$ is the time difference of arrival (TDoA) between the reference point and the $(n_x, n_y)$-th antenna. $(\theta_{k, l}, \varphi_{k, l})$ is angle-of-departure (AoD) pair of the $k$-th user over the $l$-th path, where $\theta_{k, l}$ and $\varphi_{k, l}$ are the azimuth and off-nadir angles.}
%  azimuth and off-nadir angle
%between the reference point and $[n_x, n_y]$-th antenna over the $l$-th path at the $k$-th user. 
Owing to the high altitude of the LEO satellite, the angle of multiple propagation paths for a certain user can be assumed to be identical, i.e., $\Delta \tau^{[n_{x}, n_{y}]}(\theta_{k, l}, \varphi_{k, l}) \approx \Delta \tau^{[n_{x}, n_{y}]}(\theta_{k}, \varphi_{k})$, $\forall l$ \cite{you2020massive}. 
% elevation angle = 90 - AoD / compared to the surrounding scatters at the user
% 위성에서 보기에 vertical / horizontal AOD가 거의 비슷하기 때문에
% 이게 지상망하고 다른 특징이다. 
% 1: (지상망에서는 이게 같지 않고, 같은 안테나여도 multipath 별로 delta tau가 다를것이다: 서로 다른 steering vector가 L개 만큼 더해진 꼴 -> complex channel gain과 특정 steering vector로 채널을 분리할수 없다 -> AoD 로 인한 delay를 빼고 complex channel gain을 디자인 할수 없다 -> 각 element의 correlation이 낮아진다)
% 2: 굳이 steering vector를 유지하고 싶으면, 안테나 별로 서로 다른 tau^[ref]를 곱해줘야 한다 -> Multipath 별로 안테나 별 서로 다른 complex channel gain을 가진다 -> 근데 이로 인해 steering vector의 의미는 퇴색된다: 서로 다른 걸 곱해주는 순간 안테나 spacing 으로 인한 delay의 의미는 퇴색. -> AoD에 대한 함수가 아님
Accordingly, (\ref{ch_3}) is reformulated as
%\begin{align}
%     \label{ch_4}
%     & h_{k}^{[n_{x}, n_{y}]}(t) 
%     \nonumber \\
%     & = \sum_{l=0}^{L_{k}-1} 
%     \alpha_{k, l} \cdot
%     \exp\left(-j2\pi f_{\sf{c}} \left\{\tau_{k,l}^{[1, 1]} + \Delta \tau^{[n_{x}, n_{y}]}(\theta_{k}, \varphi_{k})\right\} \right)
%     \nonumber \\
%     & \,\,\,\,\,\,\, \times \exp(j2\pi f^{\sf{D}}_{k, l}t) 
%     \nonumber \\
%     & = \sum_{l=0}^{L_{k}-1} 
%     \alpha_{k, l} \cdot
%     \exp\left(-j2\pi \left\{ f_{\sf{c}} \tau_{k,l}^{[1, 1]}
%     - f^{\sf{D}}_{k, l}t \right\} \right)
%     \nonumber \\
%     & \,\,\,\,\,\,\, \times \exp\left(-j2\pi f_{\sf{c}} \cdot \Delta \tau^{[n_{x}, n_{y}]}(\theta_{k}, \varphi_{k})\right).
%\end{align}
\begin{align}
     \label{ch_4}
     h_{k}^{[n_{x}, n_{y}]}(t) & = \sum_{l=0}^{L_{k}-1} 
     \alpha_{k, l} \cdot
     \exp\left(-j2\pi \left\{ f_{\sf{c}} \tau_{k,l}^{[1, 1]}
     - f^{\sf{D}}_{k, l}t \right\} \right)
     \nonumber \\
     & \,\,\,\,\,\,\, \times \exp\left(-j2\pi f_{\sf{c}} \cdot \Delta \tau^{[n_{x}, n_{y}]}(\theta_{k}, \varphi_{k})\right).
\end{align}
Then, with $\tau_{k, 0}^{[1,1]} = \min_l{\tau_{k, l}^{[1,1]}}$, (\ref{ch_4}) can be rewritten as
\begin{align}
     \label{ch_4.5}
     h_{k}^{[n_{x}, n_{y}]}(t) = & \sum_{l=0}^{L_{k}-1} 
     \alpha_{k, l} \cdot
     \exp\left(-j2\pi \left\{ f_{\sf{c}} \tau_{k, 0}^{[1,1] }
     - f^{\sf{D}}_{k, l}t \right\} \right)
     \nonumber \\
     &  \times \exp\left(-j2\pi \left\{ f_{\sf{c}} \left( \tau^{[1, 1]}_{k, l} - \tau^{[1, 1]}_{k, 0} \right) \right\} \right) 
     \nonumber \\
     &  \times \exp\left(-j2\pi f_{\sf{c}} \cdot \Delta \tau^{[n_{x}, n_{y}]}(\theta_{k}, \varphi_{k})\right).
\end{align}
% 이렇게 하면 delay 성질을 사용자 근방에서의 delay로 나타낼 수 있음 (위성 통신에서 특히 발생하는 굉장히 높은 지연은 없어짐: 지상망하고 비슷하게 만들수 있음)

%\textbf{2) Doppler:} 
Given the movement of the LEO satellite and user, $f^{\sf{D}}_{k, l}$ is composed of two independent Doppler shifts arising from the mobility of the LEO satellite $(f^{\sf{D}}_{k, l})^{\sf {sat}}$ and the user $(f^{\sf{D}}_{k, l})^{\sf {user}}$, {such that} $f^{\sf{D}}_{k, l} = (f^{\sf{D}}_{k, l})^{\sf {sat}}+ (f^{\sf{D}}_{k, l})^{\sf {user}} $\cite{you2020massive}.
%as $f^{\sf{D}}_{k, l}=(f^{\sf{D}}_{k, l})^{\sf {sat}}+(f^{\sf{D}}_{k, l})^{\sf {user}}$ \cite{you2020massive}. 
It is rational to assume identical $(f^{\sf{D}}_{k, l})^{\sf {sat}}$ across multiple propagation paths, i.e., $(f^{\sf{D}}_{k, l})^{\sf {sat}}=(f^{\sf{D}}_{k})^{\sf {sat}}, \forall l$ due to the high altitude of the LEO satellite \cite{you2020massive}.
%This is because the angle-of-departure (AoD) of each propagation path is similar owing to the greatly higher altitude of the satellite compared to the surrounding scatters at the user \cite{you2020massive}. 
Thus, the equation (\ref{ch_4.5}) can be rewritten as 
\begin{align}
    \label{ch_5}
     h_{k}^{[n_{x}, n_{y}]}(t) 
     = & 
     \exp\left(-j2\pi \left\{ f_{\sf{c}} \tau_{k, 0}^{[1,1]} - (f^{\sf{D}}_{k})^{\sf {sat}}t \right\} \right)
     \nonumber \\
     & \times 
     g_{k}(t) \cdot
     a^{[n_{x}, n_{y}]}(\theta_{k}, \varphi_{k}),
\end{align}
where $g_{k}(t)$ and $a^{[n_{x}, n_{y}]}(\theta_{k}, \varphi_{k})$ are respectively defined as
\begin{align}
    \label{ch_6}
    & \!\!\! g_{k}(t) \! \triangleq \!\! 
    \nonumber \\ 
    & \!\!\! \sum_{l=0}^{L_{k}-1} \! \alpha_{k, l} \! \cdot  \exp \! \left(-j2 \pi \! \left\{ f_{\sf{c}} \left( \tau^{[1, 1]}_{k, l} \! - \! \tau^{[1, 1]}_{k, 0} \right)  \! - \! (f^{\sf{D}}_{k, l})^{\sf{user}}t \right\} \right),
    \\ 
    \label{ch_7}
    & \!\!\! a^{[n_{x}, n_{y}]}(\theta_{k}, \varphi_{k}) \triangleq \exp\left(-j2\pi f_{\sf{c}} \cdot \Delta \tau^{[n_{x}, n_{y}]}(\theta_{k}, \varphi_{k})\right).
\end{align}
In the equation (\ref{ch_7}), the TDoA between the reference point and the $(n_{x}, n_{y})$-th antenna can be derived by calculating the propagation distance difference between the reference point and the $(n_{x}, n_{y})$-th antenna, $\Delta d^{[n_{x}, n_{y}]}(\theta_{k}, \varphi_{k})$, as follows:
\begin{align}
    \label{ToA1}
    & \Delta \tau^{[n_{x}, n_{y}]}(\theta_{k}, \varphi_{k}) 
    = 
    \frac{\Delta d^{[n_{x}, n_{y}]}(\theta_{k}, \varphi_{k})}{c}  
    \nonumber \\
    & = \frac{\sin\varphi_{k}[(n_{x}-1)d_{x} \sin\theta_{k} + (n_{y}-1)d_{y} \cos\theta_{k}]}{c},
\end{align}
% 거리가 더 가까운거면 알아서 distacne가 - 값 나오겠지
where $c$ denotes the speed of light \cite{jiang2022initial}.
Therefore, the equation (\ref{ch_7}) can be rewritten as 
\begin{align}
    \label{ToA2}
     &a^{[n_{x}, n_{y}]}(\theta_{k}, \varphi_{k}) 
     \nonumber \\
     &= 
     \exp(-j \frac{2\pi}{\lambda} \sin\varphi_{k}[(n_{x} \! - \! 1)d_{x} \sin\theta_{k} + (n_{y} \! - \! 1)d_{y} \cos\theta_{k}])
     \nonumber \\
     &\overset{(a)}{=} 
     \exp(-j \pi \sin\varphi_{k}[(n_{x} \! - \! 1) \sin\theta_{k} + (n_{y} \! - \! 1) \cos\theta_{k}]),
\end{align}
where $(a)$ follows from the half-wavelength antenna spacing, i.e., $d_{x} = d_{y} = \frac{\lambda}{2}$.
{The array response vector of the $k$-th user
$\mathbf{a}_{k}$, which is a stacked version of $a^{[n_x,n_y]}(\theta_k,\varphi_k)$,
can be represented using the Kronecker product as follows:} 
%\begin{equation}
%\label{steering vector}
%    \mathbf{a}\left(\theta_k,\varphi_k\right) = \mathbf{a}_x\left(\theta_k,\varphi_k\right)\otimes\mathbf{a}_y\left(\theta_k,\varphi_k\right),
%\end{equation}
\begin{equation}
\label{steering vector}
    \mathbf{a}_{k} = \mathbf{a}_{k}^{x} \otimes \mathbf{a}_{k}^{y} \in \mathbb{C}^{N_{\sf{t}}\times1}.
\end{equation}
%\begin{equation}
%\label{steering vector}
%    \mathbf{a}_{k} = \mathbf{a}^{x}{(\theta_k, \varphi_k)} \otimes \mathbf{a}^{y}{(\theta_k, \varphi_k)} \in \mathbb{C}^{N_t\times1},
%\end{equation}
In (\ref{steering vector}), $\mathbf{a}_{k}^{x}\in\mathbb{C}^{N_{\sf{t}}^{x}\times1}$ and $\mathbf{a}_{k}^{y}\in\mathbb{C}^{N_{\sf{t}}^{y}\times1}$ respectively indicate
\begin{align}
&\mathbf{a}_{k}^{x}
\triangleq 
\left[1,e^{-j\pi\sin\varphi_k\cos\theta_k},
\cdots ,  e^{-j\pi\left(N_{\sf{t}}^x-1\right)\sin\varphi_k\cos\theta_k}\right]^{\sf T},
\nonumber \\
&\mathbf{a}_{k}^{y}
\triangleq  
\left[1,e^{-j\pi\sin\varphi_k\sin\theta_k},
 \cdots ,  e^{-j\pi\left(N_{\sf{t}}^y-1\right)\sin\varphi_k\sin\theta_k}\right]^{\sf T}.
 \nonumber
\end{align}
% TN에 비해 phi가 너무 작기는 작다 max(sin(phi)) = 0.196
% 같은 complex channel gain이 곱해지기에, 다른 channel model에 비해 random 성이 확실히 저기는 적다. 
Therefore, the channel vector of the $k$-th user is given by 
\begin{align}
    \label{ch_8}
    \mathbf{h}_{k}(t) \! =   \exp  \left(-j2\pi f_{\sf{c}} \left\{ \tau_{k, 0}^{[1,1]} - (f^{\sf{D}}_{k})^{\sf {sat}} \right\} \right) \cdot
     g_{k}(t) \cdot
     \mathbf{a}_{k}.
\end{align}
% owing to the favorable LoS condition in LEO SATCOM and the constant velocity of the LEO satellite over time, 
% 사용자 doppler 만 남게하고, delay를 지상망하고 비슷하게 해서 terrestrial에서의 Rician 하고 동일하게 만들려는 목적
% LoS가 풍부하고 약간 약하게 송신한 (Path attenuation) 지상망 하고 동일: -> 지상망에서 쓰던 Rician 그대로 사용 가능
% 그리고 이렇게 하니까 모든 element가 동일한 statistical property를 띔
In (\ref{ch_8}), $\tau_{k, 0}^{[1,1]}$ and $(f^{\sf{D}}_{k})^{\sf {sat}}$ can be compensated with proper time and frequency synchronization as depicted in \cite{you2020massive, li2021downlink, you2022hybrid, you2022beam}. Hence, the channel expression (\ref{ch_8}) can be rewritten as 
\begin{align}
    \label{ch_9}
    \mathbf{h}_{k}(t) = g_{k}(t) \cdot
    \mathbf{a}_{k} \in \mathbb{C}^{{N_{\sf{t}}}\times{1}}.
\end{align}

In (\ref{ch_9}), it can be observed that the statistical properties of the channel are primarily determined by the propagation environment vicinity users.
%\textbf{4) Complex channel gain:}
Moreover, $g_k(t)$ can be postulated to conform to the Rician fading owing to the favorable LOS condition of LEO SATCOM. {From such features, the real and imaginary parts of $g_k(t)$ exhibit to follow an independent and identically distributed (i.i.d.) Gaussian distribution such that ${\sf{Re}}\{g_k(t)\}, {\sf{Im}}\{g_k(t)\} \sim \mathcal{N}\left(\sqrt{\frac{\kappa_k\gamma_k}{2(\kappa_k+1)}},\frac{\gamma_k}{2(\kappa_k+1)}\right)$ with the Rician K-factor $\kappa_k$.}
%quantifying the relative strength of the LOS path compared to the non-LOS paths. 
The average channel power is set as 
\begin{align}\label{channel_gamma}
\mathbb{E}\left[\vert g_k(t) \vert ^2\right] = \gamma_k = \frac{G_{\sf{Tx}}{G_{\sf Rx}}}
{ (4\pi {f_{\sf{c}}} {d_k} / {c})^{2} k_{\sf{B}} {T_{\sf sys}} {B}},
\end{align}
where $G_{\sf Tx}$, $G_{\sf Rx}$, $d_k$, $k_{\sf{B}}$, $T_{\sf sys}$, and $B$ denote the transmit antenna gain, user terminal antenna gain, satellite-to-$k$-th user distance, Boltzmann constant, system noise temperature, and bandwidth, respectively. 
{Given the severe propagation delay and short channel coherence time of
LEO SATCOM, obtaining instantaneous CSIT is usually infeasible. 
{Instead, we presume that the geometry relations of the satellite-to-users (i.e., $\theta_k$, $\varphi_k$, $\forall k \in \mathcal{K}$) and the statistical information of channel gain (i.e., $\mathbb{E}\left[\vert g_k(t) \vert ^2\right] = \gamma_k$, $\forall k \in \mathcal{K}$), which vary comparably slower, are available at the LEO satellite
\cite{you2020massive, roper2022beamspace, li2021downlink, you2022hybrid, you2022beam}.}
Since we focus on investigating the system within a certain fading block, we rewrite the channel vector (\ref{ch_9}) as $\mathbf{h}_{k} = g_{k}\mathbf{a}_{k}$ from the following subsection onwards.}

\subsection{Signal Model for RSMA-Based NOUM Transmission}

With a rate-splitting strategy for NOUM transmission, unicast messages of each user $W_{1}, \cdots, W_{K}$ are firstly split
into common and private parts as $W_k \rightarrow \{W_{k,{\sf{c}}}, \,\, W_{k,{\sf{p}}} \}$, $\forall k \in \mathcal{K}$. 
The common messages $W_{1,{\sf{c}}}, \cdots, W_{K,{\sf{c}}}$ and multicast message $W_{\sf{mc}}$ are combined into one common message $W_{\sf{c}}$ and then encoded into a common stream $s_{\sf{c}}$ using a codebook that is known by all users. 
On the other hand, the private messages are encoded into private streams $s_1, \cdots, s_K$ intended for each user, using the dedicated codebooks. Through the precoding vectors $\mathbf{f}_{\sf{c}}, \mathbf{f}_1,\cdots,\mathbf{f}_K \in \mathbb{C}^{{N_{\sf{t}}}\times{1}}$, the stream vector $\mathbf{s}=[s_{\sf{c}}, s_1, \cdots, s_K]^{\sf{T}}\in \mathbb{C}^{{(K+1)}\times{1}}$ is linearly combined as 
%Since we consider the rate-splitting strategy for the NOUM transmission, unicast messages $W_{1}, \cdots, W_{K}$ dedicated for each user are firstly split
%into common and private messages, i.e., $W_k \rightarrow \{W_{k,{\sf{c}}}, \,\, W_{k,p} \}$, $\forall k \in \mathcal{K}$. 
%The common messages $W_{1,{\sf{c}}}, \cdots, W_{K,{\sf{c}}}$ and multicast message $W_{\sf{mc}}$ are combined into one common message $W_{\sf{c}}$ and encoded into a common stream $s_{\sf{c}}$ using a codebook that is known by all users. 
%On the other hand, the private messages are encoded into private streams $s_1, \cdots, s_K$ intended for each user using codebooks that are only known by corresponding users. 
%By using the precoding vectors $\mathbf{f}_{\sf{c}}, \mathbf{f}_1,\cdots,\mathbf{f}_K \in \mathbb{C}^{{N_t}\times{1}}$, the stream vector $\mathbf{s}=[s_{\sf{c}}, s_1, \cdots, s_K]^{\sf{T}}\in \mathbb{C}^{{(K+1)}\times{1}}$ is linearly combined as
\begin{align}
\label{TSEq}
\mathbf{x} = \mathbf{f}_{\sf{c}} s_{\sf{c}} + \sum_{j=1}^{K}{\mathbf{f}_j s_j}\in \mathbb{C}^{{N_{\sf{t}}}\times{1}},
\end{align}
where $s_{\sf{c}}$ and $s_k$ follow an i.i.d. Gaussian distribution such that $s_{\sf{c}}, s_k  \sim  \mathcal{CN}{(0,P)}$. The transmit power constraint is presented as $\Vert \mathbf{f}_{\sf{c}} \Vert^{2} + \sum_{j=1}^{K} \Vert \mathbf{f}_{j} \Vert^{2}\leq 1$, whereby the total transmit power is constrained by $P$.
%with an assumption that $s_{\sf{c}}$ and $s_k$ follow an i.i.d. Gaussian distribution such that $s_{\sf{c}}, s_k \sim \mathcal{CN}{(0,P)}$. The transmit power constraint is presented as $\Vert \mathbf{f}_{\sf{c}} \Vert^{2} + \sum\limits_{j=1}^{K} \Vert \mathbf{f}_{j} \Vert^{2}\leq 1$, whereby the total transmit power is constrained by $P$.
%The transmit signal $\mathbf{x}$ is transmitted to the $k$-th user through $\mathbf{h}_k$. 
Then, the received signal at the $k$-th user is represented as
%\begin{align}
%    \label{RSEq} 
%    y_{k} & = \mathbf{h}_{k}^{\sf{H}}\mathbf{x} + n_{k}  = \mathbf{h}_{k}^{\sf{H}}\bigg(\mathbf{f}_{\sf{c}} s_{\sf{c}} + \sum_{j=1}^{K}{\mathbf{f}_j s_j}\bigg) + n_{k} \nonumber \\ 
%    & = 
%    \left(g_k \mathbf{a}_{k}\right)^{\sf{H}}\bigg(\mathbf{f}_{\sf{c}} s_{\sf{c}} + \sum_{j=1}^{K}{\mathbf{f}_j s_j}\bigg) + n_{k},
%\end{align}
\begin{align}
    \label{RSEq} 
    y_{k} = \mathbf{h}_{k}^{\sf{H}}\mathbf{x} + n_{k}  = 
    \left(g_k \mathbf{a}_{k}\right)^{\sf{H}}\bigg(\mathbf{f}_{\sf{c}} s_{\sf{c}} + \sum_{j=1}^{K}{\mathbf{f}_j s_j}\bigg) + n_{k},
\end{align}
where $n_k$ is the additive white Gaussian noise (AWGN) that follows i.i.d. such that $n_k\sim\mathcal{CN}{(0,\sigma_n^2)}$, $\forall k \in \mathcal{K}$. At the receiver, the common stream $s_{\sf{c}}$ is first decoded by treating private streams as noise. After the successful decoding, the data is re-encoded to subtract it from $y_k$ via SIC operation. After the SIC, the corresponding private stream is decoded while treating the other private streams as noise. 

{To design a robust RSMA precoder for NOUM transmission under imperfect CSIT, we characterize the ergodic rate function of instantaneous common and private rates at the $k$-th user as $R_{{\sf{c}}, k}=\mathbb{E} [R_{{\sf{c}},k}^{\sf{ins}}]$ and $R_{{\sf{p}}, k}=\mathbb{E} [R_{{\sf{p}},k}^{\sf{ins}}]$, respectively.} However, since the ergodic rate is usually challenging to handle due to the absence of a general closed form, we derive an upper bound of the ergodic rate as follows: \footnote{The upper bound in (\ref{RCEq1}) becomes tight as the Rician K-factor increases, which will be verified in the numerical studies.}
%The following procedure illustrates this:
%To design an averagely robust RSMA precoder for NOUM transmission when sCSIT can be captured at the LEO satellite, we characterize the common rate $R_{{\sf{c}}, k}$ and private rate $R_{{\sf{p}}, k}$ as an ergodic form. However, since the ergodic rate function is usually challenging to handle because closed-form does not exist generally, we substitute the ergodic rate function with an upper bound as the following procedure. 
\begin{align}
    \label{RCEq1}
    \!\!\! R_{{\sf{c}},k} & = \mathbb{E} [R_{{\sf{c}},k}^{\sf{ins}}] = \mathbb{E} \Bigg[\log_{2} 
    \Bigg(1 + \frac{\vert\mathbf{h}_{k}^{\sf{H}} \mathbf{f}_{\sf{c}}\vert^{2}}{\sum_{j=1}^{K} \vert \mathbf{h}_{k}^{\sf{H}} \mathbf{f}_{j}\vert^{2} + \frac{\sigma_{n}^2}{P}} \Bigg) \Bigg]
    \nonumber \\
    & =
    \mathbb{E} \Bigg[ \log_{2} 
    \Bigg(1 + \frac{\vert g_k \vert^{2} \vert\mathbf{a}_{k}^{\sf{H}} \mathbf{f}_{\sf{c}}\vert^{2} }{\sum_{j=1}^{K} \vert g_k \vert^{2} \vert \mathbf{a}_{k}^{\sf{H}} \mathbf{f}_{j}\vert^{2} + \frac{\sigma_{n}^2}{P}} \Bigg) \Bigg]
    \nonumber \\
    & \overset{(b)}{\leq}
    \log_{2} 
    \Bigg(1 + \frac{\mathbb{E}\left[ \vert g_k \vert^{2} \right] \vert\mathbf{a}_{k}^{\sf{H}} \mathbf{f}_{\sf{c}}\vert^{2} }{\sum_{j=1}^{K} \mathbb{E}\left[\vert g_k \vert^{2} \right]\vert \mathbf{a}_{k}^{\sf{H}} \mathbf{f}_{j}\vert^{2} + \frac{\sigma_{n}^2}{P}} \Bigg)
    \nonumber \\
    & \overset{(c)}{=}
    \log_{2} 
    \Bigg(1 + \frac{ \gamma_{k} \vert\mathbf{a}_{k}^{\sf{H}} \mathbf{f}_{\sf{c}}\vert^{2} }{\sum_{j=1}^{K} \gamma_{k}
    \vert \mathbf{a}_{k}^{\sf{H}} \mathbf{f}_{j}\vert^{2} + \frac{\sigma_{n}^2}{P}} \Bigg) \! = \! \bar{R}_{{\sf{c}},k},
\end{align}
where the step $(b)$ follows from the Jensen's inequality, i.e., $\log_{2}(1+\frac{a\mathbb{E}[X]}{b\mathbb{E}[X]+c}) \geq \mathbb{E}[\log_{2}(1+\frac{aX}{bX+c})]$, due to the concavity of $\log_{2}(1+\frac{aX}{bX+c})$ in terms of $X \geq 0$ for given non-negative numbers of $a$, $b$, and $c$. Step $(c)$ follows from the statistical information of the complex channel gain that is assumed to be known at the LEO satellite. 
{In order to ensure that all users are capable of decoding the common stream, we define a common rate as $\min_{k \in \mathcal{K}} \bar{R}_{{\sf c},k} \triangleq \sum_{j=1}^{K}C_{j} + C_{\sf{mc}}$ in which $C_{k}$ and $C_{\sf{mc}}$ are portions of the common rate allocated to the $k$-th user's unicast message and the multicast message.}

Similarly, the upper bound of $R_{{\sf{p}}, k}$ is derived as follows:
\begin{align}
    \label{RPEq1}
    \!\!\! R_{{\sf{p}},k} & = \mathbb{E} [R_{{\sf{p}},k}^{\sf{ins}}] = \mathbb{E} \Bigg[ \log_{2} 
    \Bigg(1 + \frac{\vert\mathbf{h}_{k}^{\sf{H}} \mathbf{f}_{k}\vert^{2}}{\sum_{j=1, j \neq k}^{K} \vert \mathbf{h}_{k}^{\sf{H}} \mathbf{f}_{j}\vert^{2} + \frac{\sigma_{n}^2}{P}} \Bigg) \Bigg]
    \nonumber \\
    & = 
    \mathbb{E} \Bigg[ \log_{2} 
    \Bigg(1 + \frac{ \vert g_k \vert^{2}  \vert\mathbf{a}_{k}^{\sf{H}} \mathbf{f}_{k}\vert^{2} }{\sum_{j=1, j \neq k}^{K} \vert g_k \vert^{2} \vert \mathbf{a}_{k}^{\sf{H}} \mathbf{f}_{j}\vert^{2} + \frac{\sigma_{n}^2}{P}} \Bigg) \Bigg]
    \nonumber \\
    & \leq
    \log_{2} 
    \Bigg(1 + \frac{\mathbb{E}\left[ \vert g_k \vert^{2} \right] \vert\mathbf{a}_{k}^{\sf{H}} \mathbf{f}_{k}\vert^{2} }{\sum_{j=1, j \neq k}^{K} \mathbb{E}\left[\vert g_k \vert^{2} \right]\vert \mathbf{a}_{k}^{\sf{H}} \mathbf{f}_{j}\vert^{2} + \frac{\sigma_{n}^2}{P}} \Bigg)
    \nonumber \\
    & =
    \log_{2} 
    \Bigg(1 + \frac{ \gamma_{k} \vert\mathbf{a}_{k}^{\sf{H}} \mathbf{f}_{k}\vert^{2} }{\sum_{j=1, j \neq k}^{K} \gamma_{k}
    \vert \mathbf{a}_{k}^{\sf{H}} \mathbf{f}_{j}\vert^{2} + \frac{\sigma_{n}^2}{P}} \Bigg) \! = \! \bar{R}_{{\sf{p}},k}.
\end{align}

\subsection{{Problem Formulation}}
% design an RSMA-based RM precoder for NOUM transmission that effectively
% tably satisfies non-uniform unicast traffic demands and provides intended multicast services effectively with limited available power
%The main objective of this work is to satisfy non-uniform unicast demands while providing an intended multicast message for all users with the limited power payload at the LEO satellite.
The main objective of this work is to match the offered rates to the non-uniform traffic demands for an efficient superimposed unicast and multicast transmission at the LEO satellite, {which typically has} a limited power budget.
{To this end, we formulate an optimization problem to jointly find optimal precoding vectors $\mathbf{f}_{\sf{c}}, \mathbf{f}_{1}, \cdots, \mathbf{f}_{K}$ and {common rate portions $C_{1}, \cdots, C_{K}, C_{\sf{mc}}$} that minimize the {disparity} between traffic demands and offered rates as follows:}
%To this end, we formulate the optimization problem as follows:
%The main objective of this work is to design an RSMA-based RM precoder for NOUM transmission that can effectively fulfill the traffic demands of both unicast and multicast messages. Thus, the optimization problem is formulated as 
% 이전 NOUM work에서는 mc의 demand를 QoS를 통해 잘 맞췄다, 하지만 unicast의 demand또한 매우 상이하고 이걸 맞추는게 중요하다. -> NOUM 에서 multicast 뿐만 아니라 unicast의 demand를 맞추는걸 목표로 하겠다 (Motivation)
\begin{align}
\nonumber
\mathscr{P}_1: \,\,\,\, 
\minimize_{\mathbf{f}_{\sf{c}}, \mathbf{f}_{1}, \cdots, \mathbf{f}_{K}, \mathbf{c}, C_{\sf{mc}}} \,\, 
&\sum_{j=1}^{K} \vert R_{{\sf target}{,j}}-(C_{j} + \bar{R}_{{\sf{p}}, j}) \vert ^{2} \\ 
& + \eta_{\sf{mc}} \vert R_{\sf{target, mc}}-C_{\sf{mc}} \vert ^{2}
\nonumber
\end{align}
\setcounter{equation}{19}%\vspace{-3mm}
\begin{subequations}\label{condition1}
\begin{align}
\mathrm{s.t.}\,\,\,\,\,\,
\label{PF1CST1}
& \min_{k \in \mathcal{K}} \bar{R}_{{\sf c},k} \triangleq \sum_{j=1}^{K}C_{j} + C_{\sf{mc}}, \\
\label{PF1CST2}
&C_k \geq 0, \,\, \forall k \in \mathcal{K}, \,\, C_{\sf{mc}} \geq 0, \\
\label{PF1CST3}
& \Vert \mathbf{f}_{\sf{c}} \Vert^{2} + \sum\limits_{j=1}^{K} \Vert \mathbf{f}_{j} \Vert^{2}\leq 1,
\end{align}
\end{subequations}
where $R_{{\sf target}{,k}}$ and $R_{\sf target, mc}$ denote the traffic demands of the unicast message for the $k$-th user and the multicast message, respectively. $\mathbf{c}=[C_1,\cdots, C_K]^{\sf{T}}$ indicates a vector consisting of common rate portions associated with unicast messages, and $\eta_{\sf{mc}}$ indicates a regularization parameter.
{The constraint (\ref{PF1CST1}) enables $s_{\sf{c}}$ to be decodable by all users, and the constraint (\ref{PF1CST2}) ensures the common rate portions are non-negative values.} 
The transmit power is constrained by (\ref{PF1CST3}).
%To elaborate further, $\eta_{\sf{mc}}$ is a hyperparameter that can be flexibly adjusted according to the traffic demands of unicast and multicast messages. 
% 복수 <-> 복수 (대응되도록)
%The constraint (\ref{PF1CST1}) enables $s_{\sf{c}}$ to be decodable by all users, 

\begin{remark} 
{\rm \textbf{(Difference from sum-rate maximization problem with QoS constraints \cite{mao2019rate, li2023non})}:} \rm{In the sum rate maximization problem with QoS constraints for unicast and multicast rates, a feasible solution is not always guaranteed. In other words, satisfying demands with QoS constraints can lead to an infeasible solution if the usable power budget is insufficient or the channel condition is unfavorable, thereby reducing communication reliability in LEO satellite networks with power-hungry payloads. 
In contrast, a feasible solution is always guaranteed in the proposed rate-matching framework regardless of the usable power budget and channel condition.}
%In the proposed rate-matching framework, a feasible solution is always guaranteed regardless of the usable power budget and channel condition.
\end{remark}

%\vspace{-2mm}

\section{Proposed RSMA-Based Rate-Matching Framework for NOUM Transmission} 
The formulated problem $\mathscr{P}_{1}$ is challenging to solve due to {its} non-convexity and non-smoothness properties. To tackle this, we convert it into a more tractable form by transforming the constrained problem into an unconstrained problem in which the objective function is differentiable. {We then derive the first-order KKT condition of the reformulated problem and show that the first-order KKT condition is cast as an NEPv. Subsequently, we introduce the GPI-RS-NOUM algorithm, based on the principle of the conventional power iteration, to jointly find an optimal rate-matching precoding vector (including power allocation) and common rate portion.}

\subsection{{Reformulate the Problem to a Tractable Form}}

We first reformulate the rate expressions (\ref{RCEq1}) and (\ref{RPEq1}) as 
\begin{align}
    \label{RCEq2}
    \bar{R}_{{\sf{c}},k} & = \log_{2} 
    \Bigg(\frac{\sum_{j \in \mathcal{L}} \gamma_{k} \vert\mathbf{a}_{k}^{\sf{H}} \mathbf{f}_{j}\vert^{2} + \frac{\sigma_{n}^2}{P}} {\sum_{j=1}^{K} \gamma_{k} \vert\mathbf{a}_{k}^{\sf{H}} \mathbf{f}_{j}\vert^{2} + \frac{\sigma_{n}^2}{P} } \Bigg) \nonumber \\ 
    & = \log_{2} 
    \Bigg(\frac{\sum_{j \in \mathcal{L}}  \mathbf{f}_{j}^{\sf{H}} (\gamma_{k} \mathbf{a}_{k} \mathbf{a}_{k}^{\sf{H}}) \mathbf{f}_{j}  + \frac{\sigma_{n}^2}{P}}{\sum_{j=1}^{K} \mathbf{f}_{j}^{\sf{H}} (\gamma_{k} \mathbf{a}_{k} \mathbf{a}_{k}^{\sf{H}}) \mathbf{f}_{j} + \frac{\sigma_{n}^2}{P} }\Bigg),
\end{align}
and
\begin{align}
    \label{RPEq2}
    \bar{R}_{{\sf{p}},k} & = \log_{2} 
    \Bigg(\frac{\sum_{j=1}^{K} \gamma_{k} \vert\mathbf{a}_{k}^{\sf{H}} \mathbf{f}_{j}\vert^{2} + \frac{\sigma_{n}^2}{P}}{\sum_{j=1}^{K} \gamma_{k} \vert\mathbf{a}_{k}^{\sf{H}} \mathbf{f}_{j}\vert^{2} + \frac{\sigma_{n}^2}{P} }\Bigg) \nonumber \\
    & = \log_{2} 
    \Bigg(\frac{\sum_{j=1}^{K}  \mathbf{f}_{j}^{\sf{H}} (\gamma_{k} \mathbf{a}_{k} \mathbf{a}_{k}^{\sf{H}}) \mathbf{f}_{j} + \frac{\sigma_{n}^2}{P}}{\sum_{j=1, j \neq k}^{K} \mathbf{f}_{j}^{\sf{H}} (\gamma_{k} \mathbf{a}_{k} \mathbf{a}_{k}^{\sf{H}}) \mathbf{f}_{j} + \frac{\sigma_{n}^2}{P}}\Bigg),
\end{align}
respectively in which $\mathcal{L} \triangleq \{ {\sf{c}}, 1, \cdots, K \}$.
Subsequently, by stacking the precoding vectors as $\bar{\mathbf{f}} = [\mathbf{f}_{\sf{c}}^{\sf{T}}, \mathbf{f}_{1}^{\sf{T}}, \cdots, \mathbf{f}_{K}^{\sf{T}}]^{\sf{T}} \in \mathbb{C}^{N_{\sf{t}}(K+1) \times 1}$, the numerator term of (\ref{RCEq2}) is reformulated as
\begin{align}
\sum\limits_{j \in \mathcal{L}} \mathbf{f}_{j}^{\sf{H}} (\gamma_{k} 
 \mathbf{a}_{k} \mathbf{a}_{k}^{\sf{H}}) \mathbf{f}_{j}  + \frac{\sigma_{n}^2}{P} = \bar{\mathbf{f}}^{\sf{H}}\mathbf{A}_{k}^{\sf{c}}\bar{\mathbf{f}},
\end{align}
where
\begin{align} \label{Ack}
   \mathbf{A}_{k}^{\sf{c}} = {\sf{blkdiag}}\left(\gamma_{k} \mathbf{a}_{k} \mathbf{a}_{k}^{\sf{H}}, \cdots, \gamma_{k} \mathbf{a}_{k} \mathbf{a}_{k}^{\sf{H}}\right)
    + \frac{\sigma_{n}^2}{P} \mathbf{I},
\end{align}
% _{N_t(K+1)}
whose size is ${N_{\sf{t}}(K+1) \times N_{\sf{t}}(K+1)}$. We here assume 
 $\Vert \bar{\mathbf{f}} \Vert^{2} = 1$.
With the similar manner, the denominator term of (\ref{RCEq2}) is rewritten as
\begin{align}
\sum\limits_{j=1}^{K} \mathbf{f}_{j}^{\sf{H}} (\gamma_{k} \mathbf{a}_{k} \mathbf{a}_{k}^{\sf{H}}) \mathbf{f}_{j}  + \frac{\sigma_{n}^2}{P} = \bar{\mathbf{f}}^{\sf{H}}\mathbf{B}_{k}^{\sf{c}}\bar{\mathbf{f}},
\end{align}
where
\begin{align} \label{Bck}
   \mathbf{B}_{k}^{\sf{c}} = \mathbf{A}_{k}^{\sf{c}} - {\sf{blkdiag}}\left(\gamma_{k} \mathbf{a}_{k} \mathbf{a}_{k}^{\sf{H}}, 0, \cdots, 0 \right),
\end{align}
whose size is ${N_{\sf{t}}(K+1) \times N_{\sf{t}}(K+1)}$.
%with a matrix $\mathbf{B}_{k}^{\sf{c}} \in \mathbb{C}^{N_t(K+1) \times N_t(K+1)}$, where
The numerator term of (\ref{RPEq2}) is reformulated as
\begin{align}
\sum\limits_{j=1}^{K} \mathbf{f}_{j}^{\sf{H}} (\gamma_{k} \mathbf{a}_{k} \mathbf{a}_{k}^{\sf{H}}) \mathbf{f}_{j}  + \frac{\sigma_{n}^2}{P} = \bar{\mathbf{f}}^{\sf{H}}\mathbf{A}_{k}^{\sf{p}}\bar{\mathbf{f}},
\end{align}
where
\begin{align} \label{Apk}
   \!\!\! \mathbf{A}_{k}^{\sf{p}} \! = \!
   {\sf{blkdiag}}\left(0, \gamma_{k} \mathbf{a}_{k} \mathbf{a}_{k}^{\sf{H}}, \cdots, \gamma_{k} \mathbf{a}_{k} \mathbf{a}_{k}^{\sf{H}} \right)
    + \frac{\sigma_{n}^2}{P} \mathbf{I},
\end{align}
% _{N_t(K+1)}
whose size is ${N_{\sf{t}}(K+1) \times N_{\sf{t}}(K+1)}$.
Also, the denominator term of (\ref{RPEq2}) can be given by
\begin{align}
\sum\limits_{j=1, j \neq k}^{K} \mathbf{f}_{j}^{\sf{H}} (\gamma_{k} \mathbf{a}_{k} \mathbf{a}_{k}^{\sf{H}}) \mathbf{f}_{j}  + \frac{\sigma_{n}^2}{P} = \bar{\mathbf{f}}^{\sf{H}}\mathbf{B}_{k}^{\sf{p}}\bar{\mathbf{f}},
\end{align}
where
\begin{align} \label{Bpk}
   \!\!\!\mathbf{B}_{k}^{\sf{p}} = \mathbf{A}_{k}^{\sf{p}} - {\sf{blkdiag}}\bigg(0, \cdots, 0, \! \underbrace{\gamma_{k} \mathbf{a}_{k} \mathbf{a}_{k}^{\sf{H}}}_{(k+1)\textrm{th block}}, 0, \cdots, 0 \bigg),
\end{align}
whose size is ${N_{\sf{t}}(K+1) \times N_{\sf{t}}(K+1)}$.
By doing so, the equations (\ref{RCEq2}) and (\ref{RPEq2}) can be reformulated as
%By doing so, the rate expression of the common rate and private rate can be reformulated as
\begin{align}
    \label{REq}
    \bar{R}_{{\sf{c}},k} = \log_{2} 
\left(\frac{\bar{\mathbf{f}}^{\sf{H}}\mathbf{A}_{k}^{\sf{c}}\bar{\mathbf{f}}}{\bar{\mathbf{f}}^{\sf{H}}\mathbf{B}_{k}^{\sf{c}}\bar{\mathbf{f}}}\right), \,\, \bar{R}_{{\sf{p}},k} = \log_{2} 
\left(\frac{\bar{\mathbf{f}}^{\sf{H}}\mathbf{A}_{k}^{\sf{p}}\bar{\mathbf{f}}}{\bar{\mathbf{f}}^{\sf{H}}\mathbf{B}_{k}^{\sf{p}}\bar{\mathbf{f}}}\right),
\end{align}
respectively. Therefore, $\mathscr{P}_1$ can be reformulated as follows:
\begin{align}
\nonumber
\mathscr{P}_2: \,\,\,\, 
\minimize_{\bar{\mathbf{f}}, \mathbf{c}, C_{\sf{mc}}} \,\, 
& \sum_{j=1}^{K} \left\vert R_{{\sf target}{,j}}-\left(C_{j} + \log_{2}     \left(\frac{\bar{\mathbf{f}}^{\sf{H}}\mathbf{A}_{j}^{\sf{p}}\bar{\mathbf{f}}}{\bar{\mathbf{f}}^{\sf{H}}\mathbf{B}_{j}^{\sf{p}}\bar{\mathbf{f}}}\right)\right) \right\vert^{2} \\ 
& + \eta_{\sf{mc}} \vert R_{\sf{target, mc}}-C_{\sf{mc}} \vert ^{2}
\nonumber
\end{align}
\setcounter{equation}{31}%\vspace{-3mm}
\begin{subequations}\label{condition2}
\begin{align}
\mathrm{s.t.}\,\,\,\,\,\,
\label{PF2CST1}
& \min_{k \in \mathcal{K}} \,\, \log_{2} \left(\frac{\bar{\mathbf{f}}^{\sf{H}}\mathbf{A}_{k}^{\sf{c}}\bar{\mathbf{f}}}{\bar{\mathbf{f}}^{\sf{H}}\mathbf{B}_{k}^{\sf{c}}\bar{\mathbf{f}}}\right)  \triangleq \sum_{j=1}^{K}C_{j} + C_{\sf{mc}}, \\
\label{PF2CST2}
&C_k \geq 0, \,\, \forall k \in \mathcal{K}, \,\, C_{\sf{mc}} \geq 0,  \\
\label{PF2CST3}
& \Vert \bar{\mathbf{f}} \Vert^{2} = 1.
\end{align}
\end{subequations}
Note that the transmit power does not affect the value of rate equations as it can always be normalized. In other words, rate expressions are {scale invariant} with the transmit power. Thus, the constraint (\ref{PF2CST3}) is vanished in $\mathscr{P}_{2}$, reformulating it as
\begin{align}
\nonumber
\mathscr{P}_3: \,\,\,\, 
\minimize_{\bar{\mathbf{f}}, \mathbf{c}, C_{\sf{mc}}} \,\, 
& \sum_{j=1}^{K} \left\vert R_{{\sf target}{,j}}-\left(C_{j} + \log_{2}    \left(\frac{\bar{\mathbf{f}}^{\sf{H}}\mathbf{A}_{j}^{\sf{p}}\bar{\mathbf{f}}}{\bar{\mathbf{f}}^{\sf{H}}\mathbf{B}_{j}^{\sf{p}}\bar{\mathbf{f}}}\right)\right) \right\vert^{2} \\ 
& + \eta_{\sf{mc}} \vert R_{\sf{target, mc}}-C_{\sf{mc}} \vert ^{2}
\nonumber
\end{align}
\setcounter{equation}{32}%\vspace{-3mm}
\begin{subequations}\label{condition3}
\begin{align}
\mathrm{s.t.}\,\,\,\,\,\,
\label{PF3CST1}
& \min_{k \in \mathcal{K}} \,\, \log_{2} \left(\frac{\bar{\mathbf{f}}^{\sf{H}}\mathbf{A}_{k}^{\sf{c}}\bar{\mathbf{f}}}{\bar{\mathbf{f}}^{\sf{H}}\mathbf{B}_{k}^{\sf{c}}\bar{\mathbf{f}}}\right)  \triangleq \sum_{j=1}^{K}C_{j} + C_{\sf{mc}}, \\
\label{PF3CST2}
&C_k \geq 0, \,\, \forall k \in \mathcal{K}, \,\, C_{\sf{mc}} \geq 0.
\end{align}
\end{subequations}
%\subsection{Representation common portion to a fractional form}
Subsequently, to make the non-smooth minimum function differentiable, we employ the LogSumExp technique \cite{shen2010dual} as
\begin{align}
\label{logsumexp}
 & \!\!\!\!\! \min_{k \in \mathcal{K}} \bar{R}_{{\sf{c}}, k} \approx
 \log\left( \frac{1}{K} \sum_{i=1}^{K}
 \exp\left\{-\frac{1}{\alpha}
\log_{2}\left(\frac{\bar{\mathbf{f}}^{\sf{H}}\mathbf{A}_{i}^{\sf{c}}\bar{\mathbf{f}}}{\bar{\mathbf{f}}^{\sf{H}}\mathbf{B}_{i}^{\sf{c}}\bar{\mathbf{f}}}\right)
 \right\}\right)^{\!\! -\alpha} \!\!\!\!\!\!,
\end{align}
where the approximation becomes tight as $\alpha \rightarrow + 0$.
With the approximated equation, the problem $\mathscr{P}_3$ can be rewritten as 
\begin{align}
\nonumber
\mathscr{P}_4: \,\,\,\, 
\minimize_{\bar{\mathbf{f}}, \mathbf{c}, C_{\sf{mc}}} \,\, 
& \sum_{j=1}^{K} \left\vert R_{{\sf target}{,j}}-\left(C_{j} + \log_{2}    \left(\frac{\bar{\mathbf{f}}^{\sf{H}}\mathbf{A}_{j}^{\sf{p}}\bar{\mathbf{f}}}{\bar{\mathbf{f}}^{\sf{H}}\mathbf{B}_{j}^{\sf{p}}\bar{\mathbf{f}}}\right)\right) \right\vert^{2} \\ 
& + \eta_{\sf{mc}} \vert R_{\sf{target, mc}}-C_{\sf{mc}} \vert ^{2}
\nonumber
\end{align}
\setcounter{equation}{34}%\vspace{-3mm}
\begin{subequations}\label{condition4}
\begin{align}
\mathrm{s.t.}\,
\label{PF4CST1}
& \log\left( \frac{1}{K} \sum_{i=1}^{K}\exp
\left\{
\log_{2}\left(\frac{\bar{\mathbf{f}}^{\sf{H}}\mathbf{A}_{i}^{\sf{c}}\bar{\mathbf{f}}}{\bar{\mathbf{f}}^{\sf{H}}\mathbf{B}_{i}^{\sf{c}}\bar{\mathbf{f}}}\right)^{\!\!\! -\frac{1}{\alpha}}\right\}\right)^{\!\!\! -\alpha} 
\!\!\!\!\! = \! \sum_{j=1}^{K}C_{j} + C_{\sf{mc}}, \\
\label{PF4CST2}
& \, C_k \geq 0, \,\, \forall k \in \mathcal{K}, \,\, C_{\sf{mc}} \geq 0.
\end{align}
\end{subequations}

Nevertheless, addressing the reformed problem $\mathscr{P}_{4}$ remains a formidable task due to the non-convex constraint (\ref{PF4CST1}) and multiple constraints associated with common rates in (\ref{PF4CST2}).
%However, we observe that the problem $\mathcal{P}_{4}$ is still difficult to handle due to the non-convex constraint (\ref{PF4CST1}) and the constraint (\ref{PF4CST2}) the number of which proportionally increases associated to that of users.
{To resolve this issue, we rewrite the common rate portion of the $k$-th user $C_{k}$ into the ratio of the equation (\ref{logsumexp}) as}
\begin{align}
    & C_{k}(\bar{\mathbf{f}}, \mathbf{v}) = 
    \nonumber \\
    &  
\frac{\mathbf{v}^{\sf{H}}\mathbf{E}_{k}\mathbf{v}}{\mathbf{v}^{\sf{H}}\mathbf{v}}\log\left( \frac{1}{K} \sum_{i=1}^{K}\exp
\left\{-\frac{1}{\alpha}
\log_{2}\left(\frac{\bar{\mathbf{f}}^{\sf{H}}\mathbf{A}_{i}^{\sf{c}}\bar{\mathbf{f}}}{\bar{\mathbf{f}}^{\sf{H}}\mathbf{B}_{i}^{\sf{c}}\bar{\mathbf{f}}}\right)\right\}\right)^{ \!\! -\alpha} \! ,
\end{align}
where $\mathbf{v}\in\mathbb{C}^{(K+1)\times 1}$ and $\mathbf{E}_{k}\in\mathbb{C}^{(K+1)\times (K+1)}$ indicate $\mathbf{v}=[v_{1},\cdots,v_{K},v_{K+1}]^{\sf{T}}$ and the diagonal matrix in which the $(k, k)$-th diagonal element is set to be $1$ and otherwise $0$, {respectively}. {Then}, $C_{\sf{mc}}$ can be rewritten as
\begin{align}
    & \!\!\! C_{\sf{mc}}(\bar{\mathbf{f}}, \mathbf{v}) = 
    \nonumber \\
    & \!\!\! \frac{\mathbf{v}^{\sf{H}}\mathbf{E}_{K+1}\mathbf{v}}{\mathbf{v}^{\sf{H}}\mathbf{v}}\log 
    \left( \!\frac{1}{K} \sum_{i=1}^{K}\exp\left\{-\frac{1}{\alpha}
\log_{2}\left(\frac{\bar{\mathbf{f}}^{\sf{H}}\mathbf{A}_{i}^{\sf{c}}\bar{\mathbf{f}}}{\bar{\mathbf{f}}^{\sf{H}}\mathbf{B}_{i}^{\sf{c}}\bar{\mathbf{f}}}\right)\right\}\! \right)^{ \!\! -\alpha}
\end{align}
into the fractional form. 
It is noted that the value of $\frac{\mathbf{v}^{\sf{H}}\mathbf{E}_{k}\mathbf{v}}{\mathbf{v}^{\sf{H}}\mathbf{v}} = \frac{v_{k}^2}{\sum_{j=1}^{K+1}v_{j}^2}$, $\forall$ $k \in \{1,\cdots,K,K+1 \}$, is always between $0$ and $1$.
In addition to this, the summation for the ratio term of each common stream is $1$, i.e.,
$\sum_{j=1}^{K+1}{\frac{\mathbf{v}^{\sf{H}}\mathbf{E}_{j}\mathbf{v}}{\mathbf{v}^{\sf{H}}\mathbf{v}}} = 1.$
Thanks to these features, the constraints (\ref{PF4CST1}) and (\ref{PF4CST2}) can be vanished from $\mathscr{P}_{4}$. Therefore, $\mathscr{P}_{4}$ can be reformulated into the unconstrained problem as follows:
\begin{align}
\label{P5}
\mathscr{P}_5 : 
\minimize_{\bar{\mathbf{f}}, \mathbf{v}} 
f(\bar{\mathbf{f}}, \mathbf{v}),
\end{align}
where 
\begin{align}
    f(\bar{\mathbf{f}}, \mathbf{v}) & = \sum_{j=1}^{K} \left\vert R_{{\sf target}{,j}} \! - \! \left(\log_{2}  \left(\frac{\bar{\mathbf{f}}^{\sf{H}}\mathbf{A}_{j}^{\sf{p}}\bar{\mathbf{f}}}{\bar{\mathbf{f}}^{\sf{H}}\mathbf{B}_{j}^{\sf{p}}\bar{\mathbf{f}}}\right) + C_{j}(\bar{\mathbf{f}}, \mathbf{v}) \right)\right\vert^{2} 
\nonumber \\
&+ \eta_{\sf{mc}} \left\vert R_{\sf{target, mc}} - C_{\sf{mc}}(\bar{\mathbf{f}}, \mathbf{v})\right\vert ^{2}. 
\end{align}

\subsection{{Optimal Precoder Design via GPI-RS-NOUM Algorithm}}

{Although $\mathscr{P}_{5}$ is a unconstrained problem, obtaining optimal $\bar{\mathbf{f}}$ and $\mathbf{v}$ is still challenging due to the non-convexity of $f(\bar{\mathbf{f}}, \mathbf{v})$.}
{To resolve this, we derive the first-order KKT condition of $\mathscr{P}_{5}$ to be formulated as an NEPv \cite{cai2018eigenvector}.}
{The GPI-RS-NOUM algorithm built upon the process of conventional power iteration is introduced to identify the leading eigenvector that maximizes the eigenvalue of the presented NEPv.}
These procedures are explained step-by-step to improve readability. 

%These procedures are broken down and given an account to improve readability.

\subsubsection*{{\bf Step I) First-order KKT condition w.r.t. $\bar{\mathbf{f}}$}\rm} 

{We first derive the first-order KKT condition of the problem (\ref{P5}) with respect to $\bar{\mathbf{f}}$, i.e., $\frac{\partial{f(\bar{\mathbf{f}}, \mathbf{v})}} {\partial{\bar{\mathbf{f}}}} = \mathbf{0}$.} The following Lemma \ref{lemma1} indicates the first-order KKT condition to $\bar{\mathbf{f}}$. 
\begin{lemma}
\label{lemma1}
The first-order KKT condition of (\ref{P5}) with respect to $\bar{\mathbf{f}}$ holds when the following equation is satisfied.
\begin{align}
\label{KKT_4_f_final}
     {\mathbf{B}(\bar{\mathbf{f}}, \mathbf{v})}^{-1}  
     {\mathbf{A}(\bar{\mathbf{f}}, \mathbf{v})}\bar{\mathbf{f}} = \lambda(\bar{\mathbf{f}}, \mathbf{v})\bar{\mathbf{f}}.
\end{align}
{The matrices ${\mathbf{A}(\bar{\mathbf{f}}, \mathbf{v})}$ and ${\mathbf{B}(\bar{\mathbf{f}}, \mathbf{v})}$ $\in \mathbb{C}^{N_{\sf{t}}(K+1) \times N_{\sf{t}}(K+1)}$ in the equation (\ref{KKT_4_f_final}) are respectively expressed as {(\ref{A_kkt}) and (\ref{B_kkt})} at the top of this page with the following equations (\ref{common_part_A})$-$(\ref{lambda_den}).}
\begin{figure*}[!t]
%\vspace{-2mm}
%\noindent\rule{\textwidth}{.5pt}%\vskip3pt
\small
\begin{align}
    \label{A_kkt}
     & {\mathbf{A}(\bar{\mathbf{f}}, \mathbf{v})} =  
    \lambda_{\sf{num}}(\bar{\mathbf{f}}, \mathbf{v}) \! \times \! \Bigg\{
    \sum_{j=1}^{K}\Bigg[ R_{{\sf{target}}, j} \left(
    \frac{{\mathbf{A}_{j}^{\sf{p}}}}{\bar{\mathbf{f}}^{\sf{H}}
    {\mathbf{A}_{j}^{\sf{p}}}{\bar{\mathbf{f}}}}
    + \frac{\mathbf{v}^{\sf{H}}\mathbf{E}_{j}\mathbf{v}}{\mathbf{v}^{\sf{H}}\mathbf{v}}
    \mathbf{L}_{\sf{A}}(\bar{\mathbf{f}})\right)
    + 
    \log_{2} \left( \frac{\bar{\mathbf{f}}^{\sf{H}}\mathbf{A}_{j}^{\sf{p}}\bar{\mathbf{f}}}{\bar{\mathbf{f}}^{\sf{H}}\mathbf{B}_{j}^{\sf{p}}\bar{\mathbf{f}}} \right)
    \frac{{\mathbf{B}_{j}^{\sf{p}}}}{\bar{\mathbf{f}}^{\sf{H}}
    {\mathbf{B}_{j}^{\sf{p}}}{\bar{\mathbf{f}}}}
    +     \frac{\mathbf{v}^{\sf{H}}\mathbf{E}_{j}\mathbf{v}}{\mathbf{v}^{\sf{H}}\mathbf{v}}    
    C_{j}(\bar{\mathbf{f}}, \mathbf{v})
    \mathbf{L}_{\sf{B}}(\bar{\mathbf{f}})
    \nonumber \\
    & 
    + \frac{\mathbf{v}^{\sf{H}}\mathbf{E}_{j}\mathbf{v}}{\mathbf{v}^{\sf{H}}\mathbf{v}} 
    \log_{2} \left( \frac{\bar{\mathbf{f}}^{\sf{H}}\mathbf{A}_{j}^{\sf{p}}\bar{\mathbf{f}}}{\bar{\mathbf{f}}^{\sf{H}}\mathbf{B}_{j}^{\sf{p}}\bar{\mathbf{f}}} \right)
    \mathbf{L}_{\sf{B}}(\bar{\mathbf{f}})
    + C_{j}(\bar{\mathbf{f}}, \mathbf{v})
    \frac{{\mathbf{B}_{j}^{\sf{p}}}}{\bar{\mathbf{f}}^{\sf{H}}
    {\mathbf{B}_{j}^{\sf{p}}}{\bar{\mathbf{f}}}}
    \Bigg]
    + \eta_{\sf{mc}}  \left[ 
    R_{\sf{target, mc}} 
    %%%%
\frac{\mathbf{v}^{\sf{H}}\mathbf{E}_{K+1}\mathbf{v}}{\mathbf{v}^{\sf{H}}\mathbf{v}}
\mathbf{L}_{\sf{A}}(\bar{\mathbf{f}})
      + \frac{\mathbf{v}^{\sf{H}}\mathbf{E}_{K+1}\mathbf{v}}{\mathbf{v}^{\sf{H}}\mathbf{v}}
    C_{\sf{mc}}(\bar{\mathbf{f}}, \mathbf{v})
    \mathbf{L}_{\sf{B}}(\bar{\mathbf{f}}) \right] \Bigg\}.
\end{align}
\noindent\rule{\textwidth}{.5pt}%\vskip3pt
%\vspace{-2mm}
\end{figure*}
\begin{figure*}[!t]
%\vspace{-2mm}
%\noindent\rule{\textwidth}{.5pt}%\vskip3pt
\small
\begin{align}
    \label{B_kkt}
    & {\mathbf{B}(\bar{\mathbf{f}}, \mathbf{v})} =  
    \lambda_{\sf{den}}(\bar{\mathbf{f}}, \mathbf{v}) \! \times \! \Bigg\{
    \sum_{j=1}^{K}\Bigg[ R_{{\sf{target}}, j} \left(
    \frac{{\mathbf{B}_{j}^{\sf{p}}}}{\bar{\mathbf{f}}^{\sf{H}}
    {\mathbf{B}_{j}^{\sf{p}}}{\bar{\mathbf{f}}}}
    + \frac{\mathbf{v}^{\sf{H}}\mathbf{E}_{j}\mathbf{v}}{\mathbf{v}^{\sf{H}}\mathbf{v}}
    \mathbf{L}_{\sf{B}}(\bar{\mathbf{f}})\right)
    + 
    \log_{2} \left( \frac{\bar{\mathbf{f}}^{\sf{H}}\mathbf{A}_{j}^{\sf{p}}\bar{\mathbf{f}}}{\bar{\mathbf{f}}^{\sf{H}}\mathbf{B}_{j}^{\sf{p}}\bar{\mathbf{f}}} \right)
    \frac{{\mathbf{A}_{j}^{\sf{p}}}}{\bar{\mathbf{f}}^{\sf{H}}
    {\mathbf{A}_{j}^{\sf{p}}}{\bar{\mathbf{f}}}}
    + \frac{\mathbf{v}^{\sf{H}}\mathbf{E}_{j}\mathbf{v}}{\mathbf{v}^{\sf{H}}\mathbf{v}}    
    C_{j}(\bar{\mathbf{f}}, \mathbf{v})
    \mathbf{L}_{\sf{A}}(\bar{\mathbf{f}})
    \nonumber \\
    & 
    + \frac{\mathbf{v}^{\sf{H}}\mathbf{E}_{j}\mathbf{v}}{\mathbf{v}^{\sf{H}}\mathbf{v}} 
    \log_{2} \left( \frac{\bar{\mathbf{f}}^{\sf{H}}\mathbf{A}_{j}^{\sf{p}}\bar{\mathbf{f}}}{\bar{\mathbf{f}}^{\sf{H}}\mathbf{B}_{j}^{\sf{p}}\bar{\mathbf{f}}} \right)
    \mathbf{L}_{\sf{A}}(\bar{\mathbf{f}})
    + C_{j}(\bar{\mathbf{f}}, \mathbf{v})
    \frac{{\mathbf{A}_{j}^{\sf{p}}}}{\bar{\mathbf{f}}^{\sf{H}}
    {\mathbf{A}_{j}^{\sf{p}}}{\bar{\mathbf{f}}}}
    \Bigg]
    %%%%%%%%%%%%%%%%%%%%%%%%%%%%%%%%%%%%%%%%
     + \eta_{\sf{mc}}  \left[ 
    R_{\sf{target, mc}} 
    %%%%
\frac{\mathbf{v}^{\sf{H}}\mathbf{E}_{K+1}\mathbf{v}}{\mathbf{v}^{\sf{H}}\mathbf{v}}
\mathbf{L}_{\sf{B}}(\bar{\mathbf{f}})
      + \frac{\mathbf{v}^{\sf{H}}\mathbf{E}_{K+1}\mathbf{v}}{\mathbf{v}^{\sf{H}}\mathbf{v}}
    C_{\sf{mc}}(\bar{\mathbf{f}}, \mathbf{v})
    \mathbf{L}_{\sf{A}}(\bar{\mathbf{f}}) \right] \Bigg\} .
\end{align}
\noindent\rule{\textwidth}{.5pt}%\vskip3pt
%\vspace{-2mm}
\end{figure*}
\begin{align}
\label{common_part_A}
& \mathbf{L}_{\sf{A}}(\bar{\mathbf{f}}) \!=\! \sum_{i=1}^{K}\left({\frac{\exp\left\{-\frac{1}{\alpha} \log_{2} \left( \frac{\bar{\mathbf{f}}^{\sf{H}}\mathbf{A}_{i}^{\sf{c}}\bar{\mathbf{f}}}{\bar{\mathbf{f}}^{\sf{H}}\mathbf{B}_{i}^{\sf{c}}\bar{\mathbf{f}}} \right)
    \right\}}{\sum_{\ell=1}^{K}{\exp\left\{-\frac{1}{\alpha} \log_{2} \left( \frac{\bar{\mathbf{f}}^{\sf{H}}\mathbf{A}_{\ell}^{\sf{c}}\bar{\mathbf{f}}}{\bar{\mathbf{f}}^{\sf{H}}\mathbf{B}_{\ell}^{\sf{c}}\bar{\mathbf{f}}} \right)
    \right\}}}} \cdot
    \frac{{\mathbf{A}_{i}^{\sf{c}}}}{\bar{\mathbf{f}}^{\sf{H}}
    {\mathbf{A}_{i}^{\sf{c}}}{\bar{\mathbf{f}}}}
    \right),
%\end{align}
%\begin{align}
\\
\label{common_part_B}
& \mathbf{L}_{\sf{B}}(\bar{\mathbf{f}}) \!=\! \sum_{i=1}^{K}\left({\frac{\exp\left\{-\frac{1}{\alpha} \log_{2} \left( \frac{\bar{\mathbf{f}}^{\sf{H}}\mathbf{A}_{i}^{\sf{c}}\bar{\mathbf{f}}}{\bar{\mathbf{f}}^{\sf{H}}\mathbf{B}_{i}^{\sf{c}}\bar{\mathbf{f}}} \right)
    \right\}}{\sum_{\ell=1}^{K}{\exp\left\{-\frac{1}{\alpha} \log_{2} \left( \frac{\bar{\mathbf{f}}^{\sf{H}}\mathbf{A}_{\ell}^{\sf{c}}\bar{\mathbf{f}}}{\bar{\mathbf{f}}^{\sf{H}}\mathbf{B}_{\ell}^{\sf{c}}\bar{\mathbf{f}}} \right)
    \right\}}}} \cdot
    \frac{{\mathbf{B}_{i}^{\sf{c}}}}{\bar{\mathbf{f}}^{\sf{H}}
    {\mathbf{B}_{i}^{\sf{c}}}{\bar{\mathbf{f}}}}
    \right),
%\end{align}
%\begin{align}
\\
\label{lambda}
&\lambda(\bar{\mathbf{f}}, \mathbf{v}) = 
    \exp\left(-\eta_{\sf{mc}}\left\vert R_{{\sf{target, mc}}} - C_{{\sf{mc}}}(\bar{\mathbf{f}}, \mathbf{v}) \right\vert^{2}\right) \times 
\nonumber \\
& \prod_{j=1}^{K}\exp\left(-\left\vert R_{{\sf{target}},j} - \left(\log_{2}  \left(\frac{\bar{\mathbf{f}}^{\sf{H}}\mathbf{A}_{j}^{\sf{p}}\bar{\mathbf{f}}}{\bar{\mathbf{f}}^{\sf{H}}\mathbf{B}_{j}^{\sf{p}}\bar{\mathbf{f}}}\right) + C_{j}(\bar{\mathbf{f}}, \mathbf{v}) \right)\right\vert^{2}\right), 
%\end{align}
%\begin{align}
\\
\label{lambda_num}
&\lambda_{\sf{num}}(\bar{\mathbf{f}}, \mathbf{v}) = 
\nonumber \\
& \prod_{j=1}^{K}\exp\left(-\left\vert R_{{\sf{target}},j} - \left(\log_{2} \left(\frac{\bar{\mathbf{f}}^{\sf{H}}\mathbf{A}_{j}^{\sf{p}}\bar{\mathbf{f}}}{\bar{\mathbf{f}}^{\sf{H}}\mathbf{B}_{j}^{\sf{p}}\bar{\mathbf{f}}}\right) + C_{j}(\bar{\mathbf{f}}, \mathbf{v}) \right)\right\vert^{2}\right),
\end{align}
and
\begin{align}
    \label{lambda_den}
\lambda_{\sf{den}}(\bar{\mathbf{f}}, \mathbf{v}) = 
    \exp\left(\eta_{\sf{mc}}\left\vert R_{{\sf{target, mc}}} - C_{{\sf{mc}}}(\bar{\mathbf{f}}, \mathbf{v}) \right\vert^{2}\right),
\end{align}
where $\mathbf{L}_{\sf{A}}(\bar{\mathbf{f}})$ and $\mathbf{L}_{\sf{B}}(\bar{\mathbf{f}})$ $\in \mathbb{C}^{N_{\sf{t}}(K+1) \times N_{\sf{t}}(K+1)}$.
\end{lemma}
\begin{proof}
Please refer to Appendix A.
\end{proof}

\subsubsection*{{\bf Step II) First-order KKT condition w.r.t. ${\mathbf{v}}$}\rm}

{Subsequently, we derive the first-order KKT condition of (\ref{P5}) with respect to ${\mathbf{v}}$, i.e., $    \frac{\partial{f(\bar{\mathbf{f}}, \mathbf{v})}} {\partial{{\mathbf{v}}}} = \mathbf{0}$.} The following Lemma \ref{lemma2} shows the first-order KKT condition to $\mathbf{v}$.
\begin{lemma}
\label{lemma2}
The first-order KKT condition of (\ref{P5}) with respect to $\mathbf{v}$ holds when the following equation is satisfied. 
%The following Lemma \ref{lemma2} indicates the first-order KKT condition with respect to $\mathbf{v}$.
\begin{align}
\label{KKT_4_v_final}
     {\mathbf{E}(\bar{\mathbf{f}}, \mathbf{v})}^{-1}  
     {\mathbf{D}(\bar{\mathbf{f}}, \mathbf{v})}{\mathbf{v}} = \lambda(\bar{\mathbf{f}}, \mathbf{v}){\mathbf{v}},
\end{align}
{where ${\mathbf{D}(\bar{\mathbf{f}}, \mathbf{v})}$ and ${\mathbf{E}(\bar{\mathbf{f}}, \mathbf{v})}$ $\in \mathbb{R}^{(K+1) \times (K+1)}$ are respectively presented as (\ref{D_kkt}) and (\ref{E_kkt}) at the top of the next page.}
\begin{figure*}[!t]
%\vspace{-2mm}
%\noindent\rule{\textwidth}{.5pt}%\vskip3pt
\small
\begin{align}
    \label{D_kkt}
    {\mathbf{D}(\bar{\mathbf{f}}, \mathbf{v})} =  & \lambda_{\sf{num}}(\bar{\mathbf{f}}, \mathbf{v}) \! \times \!
    \Bigg\{ \sum\limits_{j=1}^{K} \Bigg[ R_{{\sf{target}}, j}
\frac{\mathbf{E}_j}{\mathbf{v}^{\sf{H}}\mathbf{v}}
%%%%%%%%%%%%%%%%%%%%%%%%%%%%%%%%%%%%%%%%%%%%%%%
 + \log\left( \frac{1}{K} \sum_{i=1}^{K}\exp\left(-\frac{1}{\alpha}
\log_{2}\left(\frac{\bar{\mathbf{f}}^{\sf{H}}\mathbf{A}_{i}^{\sf{c}}\bar{\mathbf{f}}}{\bar{\mathbf{f}}^{\sf{H}}\mathbf{B}_{i}^{\sf{c}}\bar{\mathbf{f}}}\right)
 \right)\right)^{\!\! -\alpha} \!\!\!\!\!\! \cdot \frac{\mathbf{v}^{\sf{H}}\mathbf{E}_{j}\mathbf{v}}{\mathbf{v}^{\sf{H}}\mathbf{v}}
 \cdot
 \frac{\mathbf{v}^{\sf{H}}\mathbf{E}_{j}\mathbf{v}}{(\mathbf{v}^{\sf{H}}\mathbf{v})^{2}}\mathbf{I}
%%%%%%%%%%%%%%%%%%%%%%%%%%%%%%%%%%%%%%%%%%%%%%%
+  \log_{2} \left( \frac{\bar{\mathbf{f}}^{\sf{H}}\mathbf{A}_{j}^{\sf{p}}\bar{\mathbf{f}}}{\bar{\mathbf{f}}^{\sf{H}}\mathbf{B}_{j}^{\sf{p}}\bar{\mathbf{f}}} \right)
\frac{\mathbf{v}^{\sf{H}}\mathbf{E}_{j}\mathbf{v}}{(\mathbf{v}^{\sf{H}}\mathbf{v})^{2}}\mathbf{I} \Bigg]
 %%%%%%%%%%%%%%%%%%%%%%%%%%%%%%%%%%%%%%%%%%%%%%%
\nonumber \\
 %%%%%%%%%%%%%%%%%%%%%%%%%%%%%%%%%%%%%%%%%%%%%%%
& + \eta_{\sf{mc}}
\Bigg[ R_{\sf{target, mc}}
\frac{\mathbf{E}_{K+1}}{\mathbf{v}^{\sf{H}}\mathbf{v}}
%%%%%%%%%%%%%%%%%%%%%%%%%%%%%%%%%%%%%%%%%%%%%%%
%%%%%%%%%%%%%%%%%%%%%%%%%%%%%%%%%%%%%%%%%%%%%%%
+ \log\left( \frac{1}{K} \sum_{i=1}^{K}\exp\left(-\frac{1}{\alpha}
\log_{2}\left(\frac{\bar{\mathbf{f}}^{\sf{H}}\mathbf{A}_{i}^{\sf{c}}\bar{\mathbf{f}}}{\bar{\mathbf{f}}^{\sf{H}}\mathbf{B}_{i}^{\sf{c}}\bar{\mathbf{f}}}\right)
 \right)\right)^{\!\! -\alpha} \!\!\!\!\!\! \cdot \frac{\mathbf{v}^{\sf{H}}\mathbf{E}_{K+1}\mathbf{v}}{\mathbf{v}^{\sf{H}}\mathbf{v}} \cdot
 \frac{\mathbf{v}^{\sf{H}}\mathbf{E}_{K+1}\mathbf{v}}{(\mathbf{v}^{\sf{H}}\mathbf{v})^{2}}\mathbf{I}
\Bigg] \Bigg\}.
\end{align}
\noindent\rule{\textwidth}{.5pt}%\vskip3pt
%\vspace{-2mm}
\end{figure*}
\begin{figure*}[!t]
%\vspace{-2mm}
%\noindent\rule{\textwidth}{.5pt}%\vskip3pt
\small
\begin{align}
    \label{E_kkt}
    {\mathbf{E}(\bar{\mathbf{f}}, \mathbf{v})} =  & \lambda_{\sf{den}}(\bar{\mathbf{f}}, \mathbf{v}) \! \times \!
    \Bigg\{ \sum\limits_{j=1}^{K} \Bigg[ R_{{\sf{target}}, j}
\frac{\mathbf{v}^{\sf{H}}\mathbf{E}_{j}\mathbf{v}}{(\mathbf{v}^{\sf{H}}\mathbf{v})^{2}}\mathbf{I}
%%%%%%%%%%%%%%%%%%%%%%%%%%%%%%%%%%%%%%%%%%%%%%%
 + \log\left( \frac{1}{K} \sum_{i=1}^{K}\exp\left(-\frac{1}{\alpha}
\log_{2}\left(\frac{\bar{\mathbf{f}}^{\sf{H}}\mathbf{A}_{i}^{\sf{c}}\bar{\mathbf{f}}}{\bar{\mathbf{f}}^{\sf{H}}\mathbf{B}_{i}^{\sf{c}}\bar{\mathbf{f}}}\right)
 \right)\right)^{\!\! -\alpha} \!\!\!\!\!\! \cdot
 \frac{\mathbf{v}^{\sf{H}}\mathbf{E}_{j}\mathbf{v}}{\mathbf{v}^{\sf{H}}\mathbf{v}}
 \cdot
\frac{\mathbf{E}_j}{\mathbf{v}^{\sf{H}}\mathbf{v}}
%%%%%%%%%%%%%%%%%%%%%%%%%%%%%%%%%%%%%%%%%%%%%%%
+  \log_{2} \left( \frac{\bar{\mathbf{f}}^{\sf{H}}\mathbf{A}_{j}^{\sf{p}}\bar{\mathbf{f}}}{\bar{\mathbf{f}}^{\sf{H}}\mathbf{B}_{j}^{\sf{p}}\bar{\mathbf{f}}} \right)
 \frac{\mathbf{E}_j}{\mathbf{v}^{\sf{H}}\mathbf{v}} \Bigg]
 %%%%%%%%%%%%%%%%%%%%%%%%%%%%%%%%%%%%%%%%%%%%%%%
\nonumber \\
 %%%%%%%%%%%%%%%%%%%%%%%%%%%%%%%%%%%%%%%%%%%%%%%
& + \eta_{\sf{mc}} 
\Bigg[ R_{\sf{target, mc}}
\frac{\mathbf{v}^{\sf{H}}\mathbf{E}_{K+1}\mathbf{v}}{(\mathbf{v}^{\sf{H}}\mathbf{v})^{2}}\mathbf{I}
%%%%%%%%%%%%%%%%%%%%%%%%%%%%%%%%%%%%%%%%%%%%%%%
%%%%%%%%%%%%%%%%%%%%%%%%%%%%%%%%%%%%%%%%%%%%%%%
+ \log\left( \frac{1}{K} \sum_{i=1}^{K}\exp\left(-\frac{1}{\alpha}
\log_{2}\left(\frac{\bar{\mathbf{f}}^{\sf{H}}\mathbf{A}_{i}^{\sf{c}}\bar{\mathbf{f}}}{\bar{\mathbf{f}}^{\sf{H}}\mathbf{B}_{i}^{\sf{c}}\bar{\mathbf{f}}}\right)
 \right)\right)^{\!\! -\alpha} \!\!\!\!\!\! \cdot \frac{\mathbf{v}^{\sf{H}}\mathbf{E}_{K+1}\mathbf{v}}{\mathbf{v}^{\sf{H}}\mathbf{v}}
\cdot
\frac{\mathbf{E}_{K+1}}{\mathbf{v}^{\sf{H}}\mathbf{v}}
\Bigg] \Bigg\}.
\end{align}
\noindent\rule{\textwidth}{.5pt}%\vskip3pt
%\vspace{-2mm}
\end{figure*}
\end{lemma}
\begin{proof}
    Please refer to Appendix B.
\end{proof}

\subsubsection*{{\bf Step III) First-order KKT condition of $\mathscr{P}_{5}$}\rm}

{By stacking (\ref{KKT_4_f_final}) and (\ref{KKT_4_v_final}) into a block-diagonal matrix form, the first-order KKT condition of (\ref{P5}) is expressed into the NEPv as}
\begin{align}
\label{KKT_sol}
&\begin{bmatrix}
{\mathbf{B}(\bar{\mathbf{f}}, \mathbf{v})}^{-1} & \mathbf{0} \\
\mathbf{0} & {\mathbf{E}(\bar{\mathbf{f}}, \mathbf{v})}^{-1}
\end{bmatrix}
\!\!\!
\begin{bmatrix}
{\mathbf{A}(\bar{\mathbf{f}}, \mathbf{v})} & \mathbf{0} \\
\mathbf{0} & {\mathbf{D}(\bar{\mathbf{f}}, \mathbf{v})}
\end{bmatrix}
\!\!\!
\begin{bmatrix}
\bar{\mathbf{f}}
\\
\mathbf{v}
\end{bmatrix}
\!\! = \lambda(\bar{\mathbf{f}}, \mathbf{v}) \!\! \begin{bmatrix}
\bar{\mathbf{f}}
\\
\mathbf{v}
\end{bmatrix}
\nonumber \\
& \Leftrightarrow \mathbf{U}(\bar{\mathbf{f}}, \mathbf{v})^{-1}
\mathbf{W}(\bar{\mathbf{f}}, \mathbf{v}) 
\begin{bmatrix}
\bar{\mathbf{f}}
\\
\mathbf{v}
\end{bmatrix}
= \lambda(\bar{\mathbf{f}}, \mathbf{v}) \begin{bmatrix}
\bar{\mathbf{f}}
\\
\mathbf{v}
\end{bmatrix},
\end{align}
{where $[\bar{\mathbf{f}}^{\sf{T}}, \mathbf{v}^{\sf{T}}]^{\sf{T}}$, $\mathbf{U}(\bar{\mathbf{f}}, \mathbf{v})^{-1} \mathbf{W}(\bar{\mathbf{f}}, \mathbf{v})$, and $\lambda(\bar{\mathbf{f}}, \mathbf{v})$ are eigenvector, {eigenvector-dependent} matrix, and eigenvalue, respectively. 
Herein, $\mathbf{U}(\bar{\mathbf{f}}, \mathbf{v})$ and $\mathbf{W}(\bar{\mathbf{f}}, \mathbf{v})$ are defined as follows:}
\begin{align}
\label{U,W_kkt}
&\begin{bmatrix}
{\mathbf{B}(\bar{\mathbf{f}}, \mathbf{v})} & \mathbf{0} \\
\mathbf{0} & {\mathbf{E}(\bar{\mathbf{f}}, \mathbf{v})}
\end{bmatrix}
\triangleq \mathbf{U}(\bar{\mathbf{f}}, \mathbf{v}) \in \mathbb{C}^{N_{\sf{t}}(K+2) \times N_{\sf{t}}(K+2)},
\nonumber \\
&\begin{bmatrix}
{\mathbf{A}(\bar{\mathbf{f}}, \mathbf{v})} & \mathbf{0} \\
\mathbf{0} & {\mathbf{D}(\bar{\mathbf{f}}, \mathbf{v})}
\end{bmatrix}
\! \triangleq \mathbf{W}(\bar{\mathbf{f}}, \mathbf{v}) \in \mathbb{C}^{N_{\sf{t}}(K+2) \times N_{\sf{t}}(K+2)}.
\end{align}
{We note that the eigenvalue $\lambda(\bar{\mathbf{f}}, \mathbf{v})$ of (\ref{KKT_sol}) is given by $\exp{\left(-f(\bar{\mathbf{f}}, \mathbf{v})\right)}$ in which $f(\bar{\mathbf{f}}, \mathbf{v})$ is the objective function for $\mathscr{P}_{5}$.
Thus, identifying the leading eigenvector that maximizes the eigenvalue is equivalent to finding the best local optimums $\bar{\mathbf{f}}^{[\star]}$ and $\mathbf{v}^{[\star]}$ that minimize the disparity between the demands and offered rates among the first-order optimal points.}

\subsubsection*{{\bf Step IV) GPI-RS-NOUM algorithm}\rm}

{To obtain the best local optima $\bar{\mathbf{f}}^{[\star]}$ and $\mathbf{v}^{[\star]}$ in a computational efficient manner, we propose GPI-RS-NOUM algorithm built upon the method in \cite{choi2019joint, park2022rate}.} We update $\bar{\mathbf{f}}^{[t]}$ and $\mathbf{v}^{[t]}$ at the $t$-th iteration by
\begin{align}
\label{GPI}
     &
     [(\bar{\mathbf{f}}^{[t]})^{\sf{T}}, (\mathbf{v}^{[t]})^{\sf{T}}]^{\sf{T}}
     \leftarrow \\
     &
     \frac{\mathbf{U}(\bar{\mathbf{f}}^{[t-1]}, \mathbf{v}^{[t-1]})^{-1}\mathbf{W}(\bar{\mathbf{f}}^{[t-1]}, \mathbf{v}^{[t-1]})[(\bar{\mathbf{f}}^{[t-1]})^{\sf{T}}, (\mathbf{v}^{[t-1]})^{\sf{T}}]^{\sf{T}}}
     {\Vert \mathbf{U}(\bar{\mathbf{f}}^{[t-1]}, \mathbf{v}^{[t-1]})^{-1}\mathbf{W}(\bar{\mathbf{f}}^{[t-1]}, \mathbf{v}^{[t-1]})[(\bar{\mathbf{f}}^{[t-1]})^{\sf{T}}, (\mathbf{v}^{[t-1]})^{\sf{T}}]^{\sf{T}} \Vert}
     \nonumber
\end{align}
until the stopping criterions $\Vert \bar{\mathbf{f}}^{[t]}-\bar{\mathbf{f}}^{[t-1]} \Vert < \epsilon$ and $\Vert {\mathbf{v}}^{[t]}-{\mathbf{v}}^{[t-1]} \Vert < \epsilon$ are met for the sufficiently small tolerance level $\epsilon$. 
With the obtained solutions $\bar{\mathbf{f}}^{[\star]}$ and $\mathbf{v}^{[\star]}$, instantaneous unicast and multicast rates are calculated.
The detailed procedure of the proposed algorithm is summarized in \textbf{Algorithm \ref{Algorithm 1}}. 
The convergence speed of the proposed algorithm is determined by a scale of $\alpha$ \cite{park2022rate}. 
As the scale of $\alpha$ diminishes, the {difference} between the LogSumExp and true minimum function values narrows; however, the convergence decelerates due to the degradation of the smoothness characteristic inherent to the LogSumExp function.
%Although the gap between the values of the LogSumExp and the minimum function becomes tighter as the scale of $\alpha$ is smaller, convergence becomes slow since the smoothness property of the LogSumExp function deteriorates.
%On the other hand, with larger values of $\alpha$, the algorithm achieves rapid convergence, albeit at the cost of diminished performance.
On the other hand, with larger values of $\alpha$, the algorithm achieves rapid convergence at the expense of compromised performance.
%When the value of $\alpha$ is large, on the other hand, the performance is degraded although the algorithm converges fast. 
To strike a balance, we initialize the value of $\alpha$ to a modest level and escalate it if the algorithm does not converge within the maximum iteration number, $t^{\sf{max}}$.   

% Algorithm 이 걸려서 Remark까지 전체적으로 파란색은 못한다
\begin{remark} 
{{\rm\textbf{(Principle of GPI-RS-NOUM algorithm)}:}}
    \rm{ {The principle of the GPI-RS-NOUM algorithm is elucidated using the conventional power iteration method.
In this method, the leading eigenvector of a matrix ${\mathbf{M}} \in \mathbb{C}^{N \times N}$ is obtained by iteratively computing ${\mathbf{q}}^{[t]} = \frac{{\mathbf{M}}^{t} {\mathbf{q}}^{[0]}}{\Vert {\mathbf{M}}^{t} {\mathbf{q}}^{[0]} \Vert }$. 
The reasoning behind the convergence of the conventional power iteration method is as follows. 
Let ${\mathbf{x}}^{[\star]}$ be the leading eigenvector and $\mathbf{x}_i$, $ \forall i \in \{2, \cdots, N\}$, be the non-leading eigenvectors, with the corresponding eigenvalues satisfying $\vert \lambda^{[\star]} \vert > \vert \lambda_2 \vert \geq \cdots \geq \vert \lambda_N \vert$.
Since the eigenvectors form a basis, any arbitrary unit vector ${\mathbf{q}}^{[0]}$ can be expressed as ${\mathbf{q}}^{[0]} = \alpha_1 {\mathbf{x}}^{[\star]} + \sum_{i = 2}^{N} \alpha_i {\mathbf{x}}_i$, where $\alpha_i$ are the corresponding coefficients. Also, as ${\mathbf{M}}({\mathbf{x}}^{[\star]} + \sum_{i = 2}^{N} {\mathbf{x}}_i) = {\lambda}^{[\star]}{\mathbf{x}}^{[\star]} + \sum_{i = 2}^{N} {\lambda}_i {\mathbf{x}}_i $, the following holds.
\begin{align}
\label{rev_1}
    \mathbf{M}^{t}\mathbf{q}^{[0]}  & = \alpha_1 (\lambda^{[\star]})^{t}\mathbf{x}^{[\star]} + 
    \sum_{i =2}^N \alpha_i \lambda_i^{t} \mathbf{x}_i
    \nonumber \\
     & = \alpha_1(\lambda^{[\star]})^{t} \left(\mathbf{x}^{[\star]} + \underbrace{\sum_{i=2}^N \frac{\alpha_i}{\alpha_1}\left(\frac{\lambda_i}{\lambda^{[\star]}}\right)^{t}\mathbf{x}_i}_{(d)}\right),
 \end{align}
where $(d)$ vanishes as $t \rightarrow \infty$. As such, ${\mathbf{q}}^{[t]}$ converges to the leading eigenvector  ${\mathbf{x}}^{[\star]}$ by iteratively projecting ${\mathbf{q}}^{[t]}$ onto ${\mathbf{M}}$.}

{Our problem, cast as an NEPv \cite{cai2018eigenvector}, extends the conventional eigenvalue problem by taking into account the matrix $\mathbf{U}(\bar{\mathbf{f}}, \mathbf{v})^{-1} \mathbf{W}(\bar{\mathbf{f}}, \mathbf{v})$, which depends on the eigenvector $[\bar{\mathbf{f}}^{\sf{T}}, \mathbf{v}^{\sf{T}}]^{\sf{T}}$.
For convenience, we substitute $\mathbf{U}(\bar{\mathbf{f}}, \mathbf{v})^{-1} \mathbf{W}(\bar{\mathbf{f}}, \mathbf{v})$ with ${\mathbf{M}}({\mathbf{x}}) \in \mathbb{C}^{N \times N}$ and $[\bar{\mathbf{f}}^{\sf{T}},\mathbf{v}^{\sf T}]^{\sf{T}}$ with $\mathbf{x} \in \mathbb{C}^{N \times 1}$ in the following explanations, without loss of generality.
To identify the leading eigenvector ${\mathbf{x}}^{[\star]}$ of the eigenvector-dependent matrix ${\mathbf{M}}({\bf{x}})$ that fulfills ${\mathbf{M}}({\mathbf{x}}^{[\star]}) {\mathbf{x}}^{[\star]} = \lambda^{[\star]} {\mathbf{x}}^{[\star]}$, where  $\lambda^{[\star]}$ is the maximum eigenvalue, we update
$\mathbf{x}^{[t]}$ for the $t$-th iteration as follows:
\begin{align}
\label{rev_1.5}
  \mathbf{x}^{[t]} \leftarrow \frac{\mathbf{M}(\mathbf{x}^{[t-1]}){\mathbf{x}^{[t-1]}}}{\Vert \mathbf{M}(\mathbf{x}^{[t-1]}){\mathbf{x}^{[t-1]}} \Vert}.
\end{align}
The iterative principle of (\ref{rev_1.5}) is similar to the conventional power iteration method, but a different matrix is used for the update at each iteration.
To demonstrate the rationale for achieving the leading eigenvector ${\mathbf{x}}^{[\star]}$ through (\ref{rev_1.5}), we use the Taylor expansion of ${\mathbf{M}}({\mathbf{x}}){\mathbf{x}}$ around ${\mathbf{x}}^{[\star]}$ to derive
\begin{align}
\label{rev_2}
      ({\mathbf{M}}({\mathbf{x}}){\mathbf{x}})^{\sf{H}}\mathbf{x} & = ({\mathbf{M}}({\mathbf{x}}^{[\star]}){\mathbf{x}}^{[\star]})^{\sf{H}}\mathbf{x} 
      \\
      & + ({\mathbf{x}}-{\mathbf{x}}^{[\star]})^{\sf{H}} \nabla_{\mathbf{x}}[ {\mathbf{M}}({\mathbf{x}}^{[\star]}){\mathbf{x}}^{[\star]}]\mathbf{x} + o(\| {\mathbf{x}} \! - \! {\mathbf{x}}^{[\star]} \|)
      \nonumber 
\end{align}
with an arbitrary unit vector $\mathbf{x}$.
Then, with a set of basis $\{{\mathbf{x}}^{[\star]}, {\mathbf{u}}_2, \cdots, {\mathbf{u}}_N \}$, the equation (\ref{rev_2}) can be represented as
\begin{align}
\label{rev_2.5}
     \!\!\!\!\! ({\mathbf{M}}({\mathbf{x}}){\mathbf{x}})^{\sf{H}}{\mathbf{x}} &= \alpha_{1}({\mathbf{M}}({\mathbf{x}}){\mathbf{x}})^{\sf{H}} 
      \mathbf{x}^{[\star]}+{\sum_{i=2}^{N}{\alpha_{i}}({\mathbf{M}}({\mathbf{x}}){\mathbf{x}})^{\sf{H}} \mathbf{u}_{i}}.
\end{align}
Herein, since ${\mathbf{M}}({\mathbf{x}}^{[\star]}) {\mathbf{x}}^{[\star]} = \lambda^{[\star]} {\mathbf{x}}^{[\star]}$ holds, we have
\begin{align}
\label{rev_3}
           \left[ ({\mathbf{M}}({\mathbf{x}}){\mathbf{x}})^{\sf{H}}\mathbf{x}^{[\star]} \right]^2  = \left[ \lambda^{[\star]} + o(\| {\mathbf{x}} - {\mathbf{x}}^{[\star]} \|) \right]^2.
\end{align}
Moreover, assuming that the vectors ${\mathbf{x}}^{[\star]}, {\mathbf{u}}_2, \cdots, {\mathbf{u}}_N$ are orthonormal to each other, we can yield
\begin{align}
\label{rev_4}
&\sum_{i=2}^{N}  \left[ ({\mathbf{M}}({\mathbf{x}}){\mathbf{x}})^{\sf{H}}\mathbf{u}_{i} \right]^2 
\nonumber \\
& \! \le \sum_{i=2}^{N} \left[ \lambda_{i}^2(\mathbf{x}^{\sf{H}}\mathbf{u}_{i})^2 \! + \!  2\lambda_{i} ( \mathbf{x}^{\sf{H}}\mathbf{u}_{i} ) o(\| {\mathbf{x}} \! - \! {\mathbf{x}}^{[\star]} \|) 
 \! + \!  o(\| {\mathbf{x}} \! - \! {\mathbf{x}}^{[\star]} \|)^2 \right]
\nonumber \\
& \! \le \left[ \lambda_2 \Vert \mathbf{x} - \mathbf{x}^{[{\star}]} \Vert + o(\| {\mathbf{x}} - {\mathbf{x}}^{[\star]} \|) \right]^2
\end{align}
because 
\begin{align}
\label{rev_5}
    \sum_{i = 2}^{N} ({\mathbf{x}}^{\sf H} {\mathbf{u}}_i)^2 & = 1 - ({\mathbf{x}}^{\sf H} {\mathbf{x}}^{[\star]})^2 
    \nonumber \\
     & \le 2(1 - {\mathbf{x}}^{\sf H} {\mathbf{x}}^{[\star]}) 
     \le \|{\mathbf{x}} - {\mathbf{x}}^{[\star]}\|^2.
\end{align}
Given that $\vert \lambda^{[\star]} \vert> \vert \lambda_2 \vert \ge \vert \lambda_i \vert, \forall i \neq 2$, the components associated with the non-leading eigenvectors ${\mathbf{x}}_2, \cdots, {\mathbf{x}}_N$ vanish through iterative projection of ${\mathbf{x}}$ onto ${\mathbf{M}}({\mathbf{x}})$. Consequently, only the component associated with the leading eigenvector} 
%---------------------Cutting Line--------------------------%
{${\mathbf{x}}^{[\star]}$ persists. Thus, the GPI-RS-NOUM algorithm achieves the leading eigenvector ${\mathbf{x}}^{[\star]}$, as also described in \cite{park2022rate}.
\footnote{{
The self-consistent field (SCF) iteration algorithm has been demonstrated to converge to the leading eigenvectors of NEPv given certain mild conditions \cite{cai2018eigenvector}.
The SCF algorithm iteratively updates \({\mathbf{x}}^{[t]}\) at the \(t\)-th iteration by performing eigenvector decomposition on \({\mathbf{M}}({\mathbf{x}}^{[t-1]})\); thus, it can be interpreted as a generalized eigenvector decomposition. 
Considering that the proposed algorithm is a generalized version of the conventional power iteration, which tracks the outcomes of eigenvector decomposition in conventional eigenvalue problems, we can infer that the proposed algorithm also converges to the leading eigenvector found by the SCF algorithm.
To support this, we validate the convergence of the proposed algorithm through simulations in Section IV.}}} 
}
\end{remark}

\begin{remark} {\rm\textbf{(Algorithm complexity)}:}
\rm{The primary computational complexity of the GPI-RS-NOUM algorithm lies in the inverse calculation of $\mathbf{B}(\bar{\mathbf{f}}, \mathbf{v})$ and $\mathbf{E}(\bar{\mathbf{f}}, \mathbf{v})$.
%Since the matrix $\mathbf{B}(\bar{\mathbf{f}}, \mathbf{v})$ is linearly combined with block diagonal matrices, the inverse of $\mathbf{B}(\bar{\mathbf{f}}, \mathbf{v})$ can be implemented by calculating the inverse of each submatrix. 
Given that matrix $\mathbf{B}(\bar{\mathbf{f}}, \mathbf{v})$ is a linear combination of block diagonal matrices, the inversion of $\mathbf{B}(\bar{\mathbf{f}}, \mathbf{v})$ can be achieved by computing the inverses of its constituent submatrices. Thus, the computational complexity of $\mathbf{B}(\bar{\mathbf{f}}, \mathbf{v})^{-1}$ is in order of $\mathcal{O}( \frac{1}{3} N_{\sf{t}}^3 K  )$. 
%The computational complexity of $\mathbf{E}(\bar{\mathbf{f}}, \mathbf{v})^{-1}$ is in order of $\mathcal{O}(K+1)$ since it is linearly combined with diagonal matrices.
The computational complexity of $\mathbf{E}(\bar{\mathbf{f}}, \mathbf{v})^{-1}$ is in order of $\mathcal{O}(K+1)$ since it is formed {from} a linear combination of diagonal matrices.
Therefore, the total computational complexity of the proposed GPI-RS-NOUM algorithm in a big-O sense is $\mathcal{O}(\frac{1}{3} N_{\sf{t}}^3 K)$.}
\end{remark}

\begin{remark} {\rm\textbf{(Key differences from prior GPI-RS-based algorithms \cite{park2022rate, parks2023rate, kim2023distributed, maxminGPI})}:}
    \rm{{In spite of the proven superiority of generalized power iteration for rate-splitting (GPI-RS) algorithm over conventional methods, such as the weighted minimum mean square error (WMMSE) algorithm,
    prior work has not rigorously determined an optimal common rate portion  \cite{park2022rate, parks2023rate, kim2023distributed, maxminGPI}.}   
    %Specifically, the previous works \cite{park2022rate, parks2023rate, kim2023distributed} have investigated optimizing the precoding vector $\bar{\mathbf{f}}$ alone. Even though the authors in \cite{maxminGPI} have determined the common portion allocation, 
    {In NOUM transmission, however, optimizing the common rate portions in a rigorous way is essential since the multicast message is incorporated in the common stream.}
    %\tcr{Rigorously allocating the common rate portion to users is essential to ensure the accommodation of each user's individual QoS requirements.}
    {In stark contrast to the previous studies,} we jointly optimize precoding vectors and common rate portions by expressing the common rate portion as the ratio of the minimum common rate. {Notably, while jointly optimizing the common rate portion, our proposed algorithm retains the same computational complexity as previous GPI-RS-based algorithms that update the precoding vector alone, in terms of big-O notation.}}
\end{remark}

%\textit{\textbf{Remark}:} 
%To summarize, in our proposed algorithm, the optimization variables $\bar{\mathbf{f}}$, $\mathbf{c}$, $C_{\sf{mc}}$, $\mathbf{a}$, $\mathbf{b}$, and $d$ are updated iteratively until a stopping criterion is fulfilled.
%It is worth noting that once $\bar{\mathbf{f}}$ is determined in each iteration, $\mathbf{c}$, $C_{\sf{mc}}$, $\mathbf{a}$, $\mathbf{b}$, and $d$ becomes uniquely determined owing to the Lemma \ref{lemma3}. 
%Thus, $\mathbf{c}$, $C_{\sf{mc}}$, $\mathbf{a}$, $\mathbf{b}$, and $d$ can be considered as a function with respect to $\bar{\mathbf{f}}$. Consequently, we can conclude that the proposed algorithm converges in a computationally efficient manner as the previous works \cite{park2022rate, parks2023rate} that only consider finding the optimal $\bar{\mathbf{f}}$.

\begin{algorithm}[!t]
\caption{GPI-RS-NOUM Algorithm}\label{Algorithm 1}
\begin{algorithmic}[1]
%\State \textbf{Initialize}: $\bar{\mathbf{f}}^{[0]}$ = MRT precoder,
%\\ \quad \quad \quad \,\,\,\,\, $v^{[0]}_{k} = \frac{R_{{\sf{target}}, k}}{\sum_{j=1}^{K}R_{{\sf{target}}, j} + R_{\sf{target, mc}}}$, $\forall k$,
%\\ \quad \quad \quad \,\,\,\,\, $v^{[0]}_{\sf{mc}} = \frac{R_{{\sf{target, mc}}}}{\sum_{j=1}^{K}R_{{\sf{target}}, j} + R_{\sf{target, mc}}}$, $t=1$.
\State \textbf{Initialize}: $\bar{\mathbf{f}}^{[0]}$ and $\mathbf{v}^{[0]}$
\State Set the iteration count $t=1$.
\Repeat
      \State $t \leftarrow t+1$ 
      %\State Construct $\mathbf{A}(\bar{\mathbf{f}}^{[t-1]}, \mathbf{v}^{[t-1]})$ and $\mathbf{B}(\bar{\mathbf{f}}^{[t-1]}, \mathbf{v}^{[t-1]})$ 
      %based on (\ref{A_kkt}) and (\ref{B_kkt}).
      %\State Construct $\mathbf{D}(\bar{\mathbf{f}}^{[t-1]}, \mathbf{v}^{[t-1]})$ and $\mathbf{E}(\bar{\mathbf{f}}^{[t-1]}, \mathbf{v}^{[t-1]})$ based on (\ref{D_kkt}) and (\ref{E_kkt}).
      \State Construct $\mathbf{U}(\bar{\mathbf{f}}^{[t-1]}, \mathbf{v}^{[t-1]})$ and $\mathbf{W}(\bar{\mathbf{f}}^{[t-1]}, \mathbf{v}^{[t-1]})$ based on (\ref{U,W_kkt}).
      \State Update $\bar{\mathbf{f}}^{[t]}$ and $\mathbf{v}^{[t]}$ based on (\ref{GPI}).
\Until{$\Vert \bar{\mathbf{f}}^{[t]}-\bar{\mathbf{f}}^{[t-1]} \Vert < \epsilon$ and $\Vert {\mathbf{v}}^{[t]}-{\mathbf{v}}^{[t-1]} \Vert < \epsilon$}
\State \textbf{Return}: Calculate instantaneous unicast rate $R_{k}^{\sf{ins}}$, $\forall k \in \mathcal{K}$ and multicast rate $C_{\sf{mc}}^{\sf{ins}}$ with $\bar{\mathbf{f}}^{[\star]}$ and $\mathbf{v}^{[\star]}$.
\end{algorithmic}
\end{algorithm}

\section{Performance Evaluation}

In this section, we evaluate the performance of the proposed framework (denoted as ``{\sf {GPI-RS-NOUM}}''). 
Unless explicitly stated otherwise, the simulation parameters are configured as the following.
%The LEO satellite parameters are configured with an altitude of $600$ km, a radius 
%of coverage area spanning $120$ km, and a power budget of $50$ W.
%The LEO satellite, operating at an altitude of $600$ km, covers a service area with a radius of $120$ km with a total transmit power budget of $50$ W.
{The LEO satellite, located at an altitude of $600$ km, covers a service area spanning a radius of $120$ km with a total transmit power budget of $50$ W.}
%Within the coverage area, $8$ users are served by the LEO satellite equipped with $6 \times 6$ UPA antennas using Ka-band as operating bandwidth.
{ Within the coverage footprint, the satellite equipped with $36$ antennas, i.e., $N_{\sf{t}}^x = N_{\sf{t}}^y = 6$, serves $8$ users in the Ka-band operating bandwidth.}
%The traffic demands of unicast and multicast messages are considered as $\mathbf{r}_{\sf target, uc} = [0.5, 0.5, 1, 1, 1.5, 2, 2.5, 2.5]^{\sf{T}}$ bps/Hz and $R_{\sf target, mc} = 1$ bps/Hz.  
The carrier frequency, bandwidth, Rician K-factor, transmit antenna gain, user antenna gain, system noise temperature, and variance of noise are set to be $f_{\sf{c}} = 20$ GHz, $B = 10$ MHz, $\kappa = 12$ dB, $G_{\sf{Tx}} = 6$ dBi, $G_{\sf{Rx}} = 25$ dBi, $T_{\sf{sys}} = 150$ K, and $\sigma_n^2 = 1$, respectively.
$\alpha$, $\epsilon$, and $t^{\sf{max}}$ are set as $\alpha = 10^{-2}$, $\epsilon = 10^{-4}$, and $t^{\sf{max}} = \num{1000}$.
{ Besides, to fairly satisfy the traffic demands for both unicast and multicast messages, we set the regularization parameter $\eta_{\sf{mc}}$ to the ratio of the average unicast traffic demand to the multicast traffic demand.}
%The rest of the simulation parameters are provided in Table \ref{Table1}. 
%Numerical results, averaged over 1000 randomly generated channel realizations, reflect the varying locations of users in each realization. 
Numerical results are obtained from the average of $\num{1000}$ channel realizations, with the location of users being randomly determined at each realization.
The proposed scheme is compared with the following baseline methods, under both perfect and imperfect CSIT conditions.
\begin{itemize}
%\begin{align}
%\nonumber
%\mathcal{P}: \,\, \minimize_{\mathbf{f}_{\sf{mc}}, \mathbf{f}_{1}, \cdots, %\mathbf{f}_{K}} \,\, 
%&\sum_{j=1}^{K} \vert R_{{\sf target}{,j}}- R_{j} \vert ^{2} + \eta_{\sf{mc}} %\vert R_{\sf{target, mc}}-R_{\sf{mc}} \vert ^{2}
%\nonumber
%\end{align}
%\setcounter{equation}{53}\vspace{-3mm}
%\begin{subequations}
%\begin{align}
%\mathrm{s.t.}\,\,
%& R_{{\sf mc}} \triangleq \min_{k \in \mathcal{K}} R_{{\sf mc},k}, \,\, \Vert %\mathbf{f}_{\sf{mc}} \Vert^{2} + \sum\limits_{j=1}^{K} \Vert \mathbf{f}_{j} %\Vert^{2}\leq 1,
%\end{align}
%\end{subequations}

\item { \textbf{RM-OUM}}:
% GPI-Based Rate-Matching in Orthogonal Unicast and Multicast Transmission 
This method designs the rate-matching precoder based on the OUM transmission. To be specific, unicast and multicast transmissions are performed by dividing the given resource in half, such as TDM or frequency division multiplexing (FDM). Since the given resource is divided for unicast and multicast transmission, the multicast message does not suffer from interference signals, while unicast messages suffer from interference signals induced by other unicast messages. 
\item { \textbf{LDM-RM-NOUM} \cite{kim2008superposition}}:
%GPI-Based no Rate-Splitting for Rate-Matching in Non-Orthogonal Unicast and Multicast Transmission 
The rate-matching precoder for NOUM transmission is designed based on the principle of LDM. The multicast stream is first decoded by treating unicast streams as noise and is removed from the received signal through SIC. Then, dedicated unicast streams are decoded considering other unicast streams as noise \cite{kim2008superposition}.
%To implement, we set $\mathbf{v}$ as $\mathbf{v}=[0,\cdots,0,1]^{\sf{T}}$ from the proposed scheme.
The optimization problem can be solved via setting $\mathbf{v}$ as $\mathbf{v}=[0,\cdots,0,1]^{\sf{T}}$ from the proposed scheme.

%To this end, multicast and unicast precoding vectors are optimized similarly to Section \uppercase\expandafter{\romannumeral3}. %The key different feature with {\sf GPI-RS-NOUM} is that $\mathbf{f}_{\sf{c}}$ is designed for the multicast stream alone. 
\item {\textbf{SCA-RM-NOUM} \cite{cui2023energy}}: {This approach extends the method in \cite{cui2023energy} by taking into account the multicast traffic demands.} In \cite{cui2023energy}, the RSMA precoder is designed to minimize unmet/unused rates of unicast messages by utilizing the MMSE precoder as the normalized private precoding vectors. The remaining RSMA parameters, such as the common precoder and allocated power for the normalized private precoding vectors, are optimized via the successive convex approximation (SCA) method. Given that the original work \cite{cui2023energy} has not addressed multicast transmission, we adjust the objective function to minimize unmet/unused rates for both unicast and multicast messages within the RSMA precoder design.
%Note that, the multicast transmission has not been considered in the original work. 
\item { \textbf{WMMSE-MMF-NOUM} \cite{li2023cooperative}}: %This method indicates the RSMA-based NOUM transmission which maximizes the minimum unicast rate among multiple users while satisfying the traffic demand of multicast rate through the QoS constraint. 
This method refers to an RSMA-based NOUM transmission, aimed at maximizing the minimum unicast rate across multiple users while meeting multicast traffic demands under the QoS constraint.
To optimize the RSMA precoder with this purpose, the WMMSE-based alternative optimization (AO) algorithm has been utilized in \cite{li2023cooperative}. Note that since the ideal perfect CSIT case has only been explored in 
\cite{li2023cooperative}, we employ the sample average approximation (SAA) technique to incorporate the effects of imperfect CSIT. To do so, we randomly generate $\num{1000}$ channel samples $\mathbf{h}_{k} = g_{k} \mathbf{a}_{k}$ in which $\mathbf{a}_{k}$ is given and $g_{k}$ is randomly produced such that ${\sf{Re}}\{g_{k}\}, {\sf{Im}}\{g_{k}\} \sim \mathcal{N}\left(\sqrt{\frac{\kappa_k\gamma_k}{2(\kappa_k+1)}},\frac{\gamma_k}{2(\kappa_k+1)}\right)$.
%WMMSE-SAA-Based Rate-Splitting for Rate-Matching in Non-Orthogonal Unicast and Multicast Transmission 
\end{itemize}
In what follows, the simulation results are demonstrated in which {\sf{UC}} and {\sf{MC}} denote unicast and multicast, respectively.

\subsubsection*{{\bf Achievable rate comparison}\rm}

%\begin{figure}[!t]
%\centering
%    \subfigure[Perfect CSIT]{\includegraphics[width=.493\columnwidth]{Figures/result_3.pdf}\label{result_3}}
%    \hfil
%    \subfigure[Imperfect CSIT]{\includegraphics[width=.493\columnwidth]{Figures/result_4.pdf}\label{result_4}}
%    \hfil
%    \vspace{-4mm}
%\caption{
%Common and private rate portion of each message for {\sf{GPI-RS-NOUM}} when the unicast and multicast traffic demands are  $\mathbf{r}_{\sf target, uc} = [0.5, 0.5, 1, 1, 1.5, 2, 2.5, 2.5]^{\sf{T}}$ bps/Hz and $R_{\sf target, mc} = 1$ bps/Hz.}
%\vspace{-2mm}
%\end{figure}

\begin{figure}[!h]
\centering
    \subfigure[Perfect CSIT]{\includegraphics[width=.7\columnwidth]{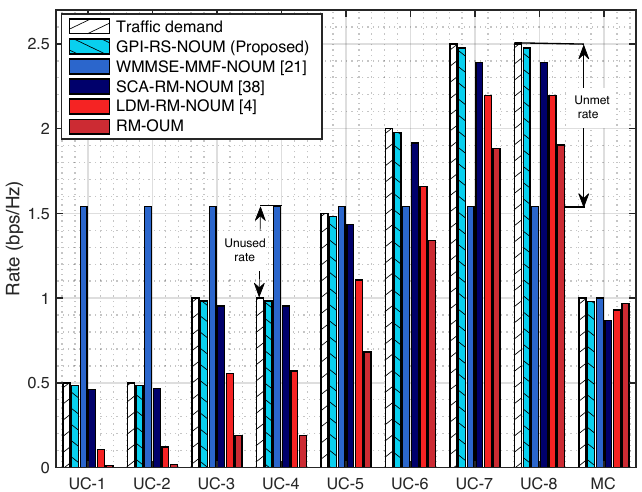}\label{result_1}}
    %\vspace{-2mm}
    \\
    \subfigure[Imperfect CSIT]{\includegraphics[width=.7\columnwidth]{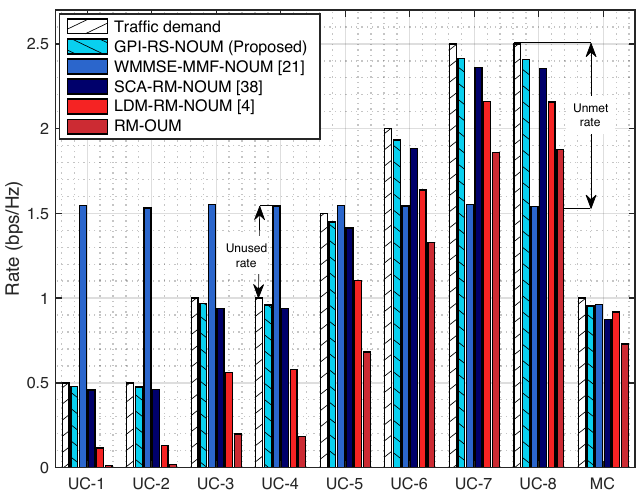}\label{result_2}}
\caption{
Achievable rate comparison for each message. The unicast and multicast traffic demands are set to be $\mathbf{r}_{\sf target, uc} = [0.5, 0.5, 1, 1, 1.5, 2, 2.5, 2.5]^{\sf{T}}$ bps/Hz and $R_{\sf target, mc} = 1$ bps/Hz.}
\vspace{-2mm}
\end{figure}

We compare the achievable rate of each scheme under both perfect and imperfect CSIT cases through Fig. \ref{result_1} and  Fig. \ref{result_2}.
The leftmost white bar represents the traffic demand of each message, where the traffic demands of unicast and multicast messages are set to be $\mathbf{r}_{\sf target, uc} = [0.5, 0.5, 1, 1, 1.5, 2, 2.5, 2.5]^{\sf{T}}$ bps/Hz and $R_{\sf target, mc} = 1$ bps/Hz, respectively.
As shown in the figures, the proposed {\sf {GPI-RS-NOUM}} framework achieves superior performance over benchmark schemes in terms of traffic demand satisfaction, demonstrating substantial improvements regardless of CSI conditions.
%As demonstrated in the figures, the proposed {\sf {GPI-RS-NOUM}} framework achieves significant performance gain over benchmark schemes in terms of traffic demand matching under both perfect and imperfect CSIT scenarios. 
In particular, we observe that {\sf{WMMSE-MMF-NOUM}} leads to pronounced unused or unmet unicast rates in both perfect and imperfect CSIT cases due to its oversight of non-uniform unicast demand distribution.

{The considerable performance gap between the proposed scheme and {\sf {LDM-RM-NOUM}}, illustrated in Fig. \ref{result_1} and Fig. \ref{result_2}, stems from the pivotal role played by the common stream.}
% MC는 UC에 간섭으로 안들어감 -> 단 파워 측면에서 손해보는건 있다.(근데 파워로 손해보는거는 RSMA에서도 동일하게 손해를 보고 있다)
{To elaborate more in detail, we provide the rate portion comparison between {\sf {GPI-RS-NOUM}} and {\sf {LDM-RM-NOUM}} in Fig. \ref{result_3}. The blue, yellow, and orange bars in Fig. \ref{result_3} denote the rate portion of common, multicast, and private streams, respectively.}
%{To delve into the details behind the results, we provide the rate portion comparison between {\sf {GPI-RS-NOUM}} and {\sf {LDM-RM-NOUM}} in Fig. \ref{result_3}. The blue, yellow, and orange bars in Fig. \ref{result_3} denote the rate portion of common, multicast, and private streams for each message, respectively.}
{Since we aim to minimize the sum of {disparity} between traffic demands and actual offered rates, in {\sf {LDM-RM-NOUM}}, the LEO satellite prioritizes fulfilling the traffic demands of users with higher unicast traffic demands.}
% RSMA가 RM에 대해 강력한 MA인 이유
% 높은 애들 끼리의 간섭으로 인해서 높은 애들도 맞추질 못한다
{During this process, significant inter-user interference is induced on users with lower unicast traffic demands, exacerbating the rate-matching performance.} {Moreover, owing to inter-user interference among the users with high unicast traffic demands, it turns out that those are also not completely met.}
{Additionally, interference from unicast messages to the multicast message leads to the disparity between the traffic demands and actual offered rates.}
% UC는 MC에 간섭으로 들어간다 -> 단 MC가 UC에 간섭으로 안들어가니까, 파워만 충분하다면 MC는 잘 맞출거다 
The proposed {\sf {GPI-RS-NOUM}} scheme utilizes a common stream decodable by all users not only for serving the multicast message but also for partially fulfilling unicast traffic demands, as depicted in Fig. \ref{result_3}. In other words, the proposed scheme allows the LEO satellite to flexibly adjust transmit power for common and private streams according to the traffic demands; the proposed {\sf {GPI-RS-NOUM}} scheme, in turn, effectively reduces inter-user interference and interference from unicast messages to the multicast message, enabling it to uniformly satisfy the traffic demands of both unicast and multicast messages.
{Further, it is shown that the proposed scheme outperforms {\sf{SCA-RM-NOUM}} in terms of satisfying traffic demands for all messages.
This is because the proposed scheme expresses the first-order KKT condition of the reformulated problem into NEPv, achieving the best local optimal point via the GPI-based method. On the other hand, {\sf{SCA-RM-NOUM}} combines SCA and MMSE methods, designing the RSMA precoder that is not optimal.}
Meanwhile, {\sf {RM-OUM}} exhibits the worst performance due to inefficient resource utilization.

\begin{figure}[t]
    \centering
    \includegraphics[width=0.9\linewidth]{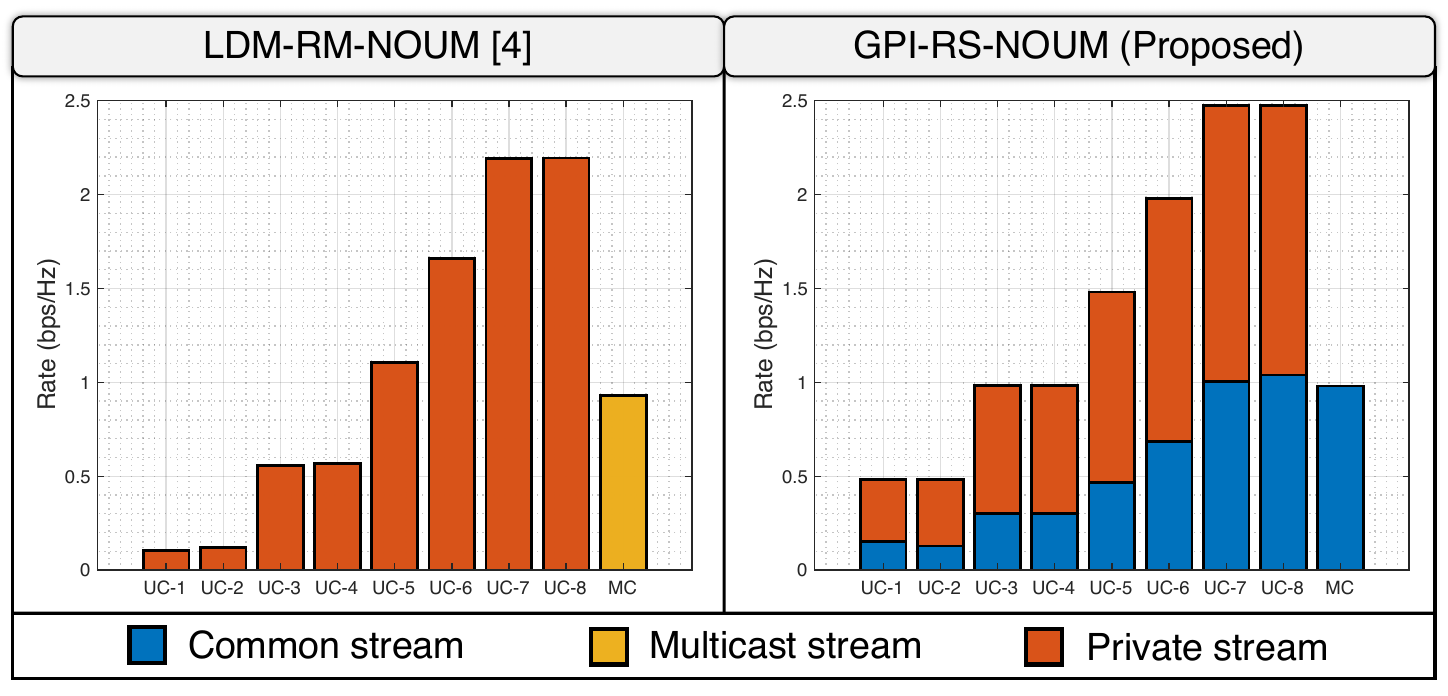}
    \caption{Rate portion comparison between {\sf {GPI-RS-NOUM}} and {\sf {LDM-RM-NOUM}} for each message under perfect CSIT. The unicast and multicast traffic demands are set to be $\mathbf{r}_{\sf target, uc} = [0.5, 0.5, 1, 1, 1.5, 2, 2.5, 2.5]^{\sf{T}}$ and $R_{\sf target, mc} = 1$ bps/Hz.}
    \label{result_3}
  %  \vspace{-4mm}
\end{figure}

\begin{figure}[!t]
\centering
    \subfigure[Perfect CSIT]{\includegraphics[width=.7\columnwidth]{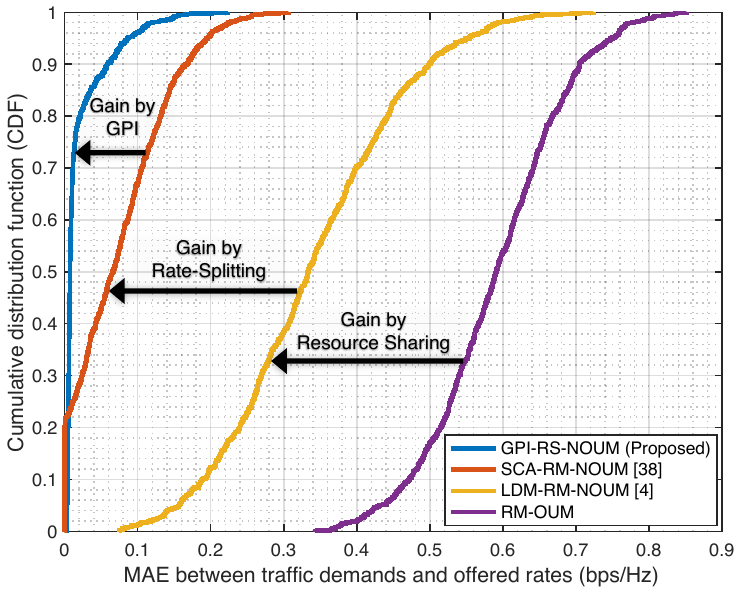}\label{result_4}}
   % \vspace{-2mm}
    \\
    \subfigure[Imperfect CSIT]{\includegraphics[width=.7\columnwidth]{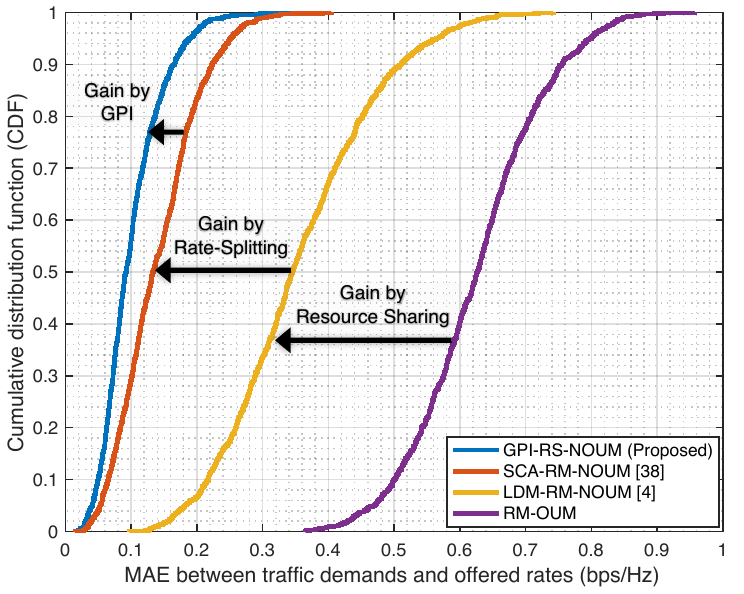}\label{result_5}}
\caption{
Comparison of CDF per MAE between the traffic demands and offered rates. The unicast and multicast traffic demands are set as $\mathbf{r}_{\sf target, uc} = [0.5, 0.5, 1, 1, 1.5, 2, 2.5, 2.5]^{\sf{T}}$ bps/Hz and $R_{\sf target, mc} = 1$ bps/Hz.}
\vspace{-2mm}
\end{figure}

\subsubsection*{{\bf Rate-matching error comparison}\rm}
The rate-matching error is compared in Fig. \ref{result_4} and Fig. \ref{result_5} under $\mathbf{r}_{\sf target, uc} = [0.5, 0.5, 1, 1, 1.5, 2, 2.5, 2.5]^{\sf{T}}$ bps/Hz and $R_{\sf target, mc} = 1$ bps/Hz. 
The cumulative distribution function (CDF) is demonstrated under perfect and imperfect CSIT in Fig. \ref{result_4} and Fig. \ref{result_5}, respectively.
The x-axis presents the mean absolute error (MAE) between the traffic demands and offered rates, i.e., 
\begin{align}
\frac{1}{K+1} \bigg( \sum_{j=1}^{K}{\left\vert R_{{\sf{target}}, j} - R_{j}^{\sf{ins}}\right\vert} + {\left\vert R_{\sf{target, mc}} - C_{\sf{mc}}^{\sf{ins}}\right\vert} \bigg),    
\end{align}
where lower levels indicate better performance.
Overall, as depicted in Fig. \ref{result_4} and Fig. \ref{result_5}, the proposed scheme exhibits the lowest MAE levels compared to the benchmark schemes regardless of the CSIT condition. 
%In essence, the proposed framework ensures stable traffic demand satisfaction of both unicast and multicast messages by effectively minimizing unmet or unused rates.
%To elaborate numerically, compared to {\sf{SCA-RM-NOUM}}, the 95 percentile MAE value of the proposed scheme decreases by about $55.6$ $\%$ and $24.7$ $\%$ under perfect and imperfect CSIT conditions, respectively. For comparison with {\sf{LDM-RM-NOUM}}, the 95 percentile MAE value of the proposed scheme decreases by about $84.3$ $\%$ and $66.6$ $\%$ under perfect and imperfect CSIT conditions, respectively.
%Comparing with {\sf{RM-OUM}}, the 95 percentile MAE value of the proposed scheme decreases by about $88.4$ $\%$ and $76.4$ $\%$ under perfect and imperfect CSIT conditions, respectively.
To elaborate further, compared to {\sf{SCA-RM-NOUM}}, our proposed scheme achieves a $95$-percentile MAE reduction of approximately $55.6$ and $24.7$ $\%$ under perfect and imperfect CSIT conditions, respectively. {Compared to {\sf{LDM-RM-NOUM}}, our scheme demonstrates a reduction of about $84.3$ and $66.6$ $\%$ under perfect and imperfect CSIT conditions, respectively.} Additionally, compared to {\sf{RM-OUM}}, our scheme exhibits reductions of approximately $88.4$ and $76.4$ $\%$ under perfect and imperfect CSIT conditions, respectively. 
{Consequently, the superior performance of our proposed scheme demonstrates its capability to ensure stable traffic demand satisfaction for both unicast and multicast messages, regardless of the CSIT conditions and user locations.}

\subsubsection*{{\bf Rate-matching error per multicast traffic demand}\rm}

We depict the average value of MAE (AMAE) for $\num{1000}$ channel realizations per different multicast traffic demands through Fig. \ref{result_6}. Herein, 
we assume that the LEO satellite is equipped with $36$ transmit antennas, arranged in a $6\times6$ UPA configuration, and the unicast traffic demands are fixed to $\mathbf{r}_{\sf target, uc} = [0.5, 0.5, 1, 1, 1.5, 2, 2.5, 2.5]^{\sf{T}}$ bps/Hz. 
In Fig. \ref{result_6}, the lower y-axis value indicates better performance. The solid line and dashed line represent the perfect and imperfect CSIT scenarios, respectively. Under both perfect and imperfect CSIT cases, the proposed {\sf{GPI-RS-NOUM}} framework outperforms benchmark schemes in terms of AMAE regardless of the multicast traffic demand. 
In particular, even when a perfect CSIT is not available, the proposed scheme shows a better AMAE level compared to {\sf{LDM-RM-NOUM}} and {\sf{RM-OUM}} with a perfect CSIT, thanks to the common stream and efficient resource reuse.
Moreover, when the multicast traffic demand exceeds $1.75$ bps/Hz, the proposed scheme with imperfect CSIT demonstrates superior AMAE performance compared to the {\sf{SCA-RM-NOUM}} with perfect CSIT.

\begin{figure}[!ht]
    \centering
{\includegraphics[width=.7\columnwidth]{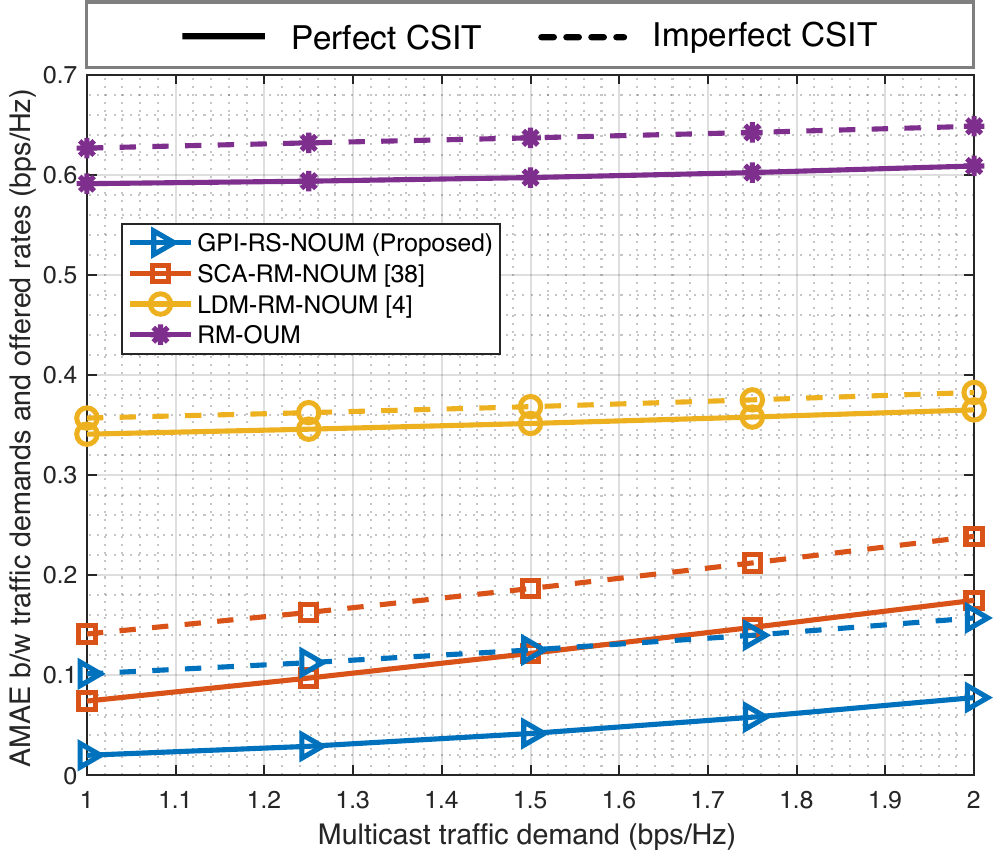}}\vspace{-7mm}
    \caption{Comparison of AMAE per multicast traffic demand. The unicast traffic demands are set as $\mathbf{r}_{\sf target, uc} = [0.5, 0.5, 1, 1, 1.5, 2, 2.5, 2.5]^{\sf{T}}$ bps/Hz.}
    \label{result_6}
    \vspace{-2mm}
\end{figure}
\begin{figure}[!ht]
    \centering
{\includegraphics[width=.7\columnwidth]{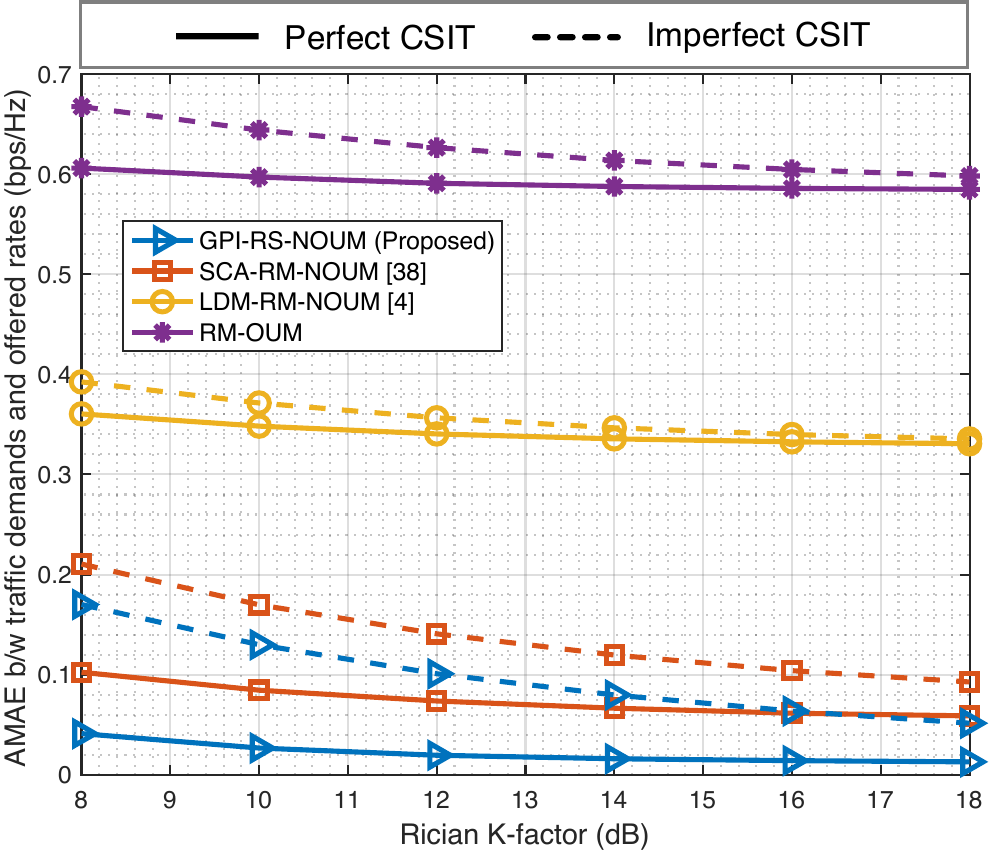}}\vspace{-7mm}
    \caption{Comparison of AMAE per Rician K-factor. The unicast and multicast traffic demands are set as $\mathbf{r}_{\sf target, uc} = [0.5, 0.5, 1, 1, 1.5, 2, 2.5, 2.5]^{\sf{T}}$ bps/Hz and $R_{\sf target, mc} = 1$ bps/Hz.}
    \label{result_7}
    \vspace{-7mm}
\end{figure}

\begin{figure}[!t]
\centering
    \subfigure[Perfect CSIT with $N_{\sf{t}} = 16$ {(i.e., $N_{\sf{t}}^x=N_{\sf{t}}^y=4$)}]{\includegraphics[width=.7\columnwidth]{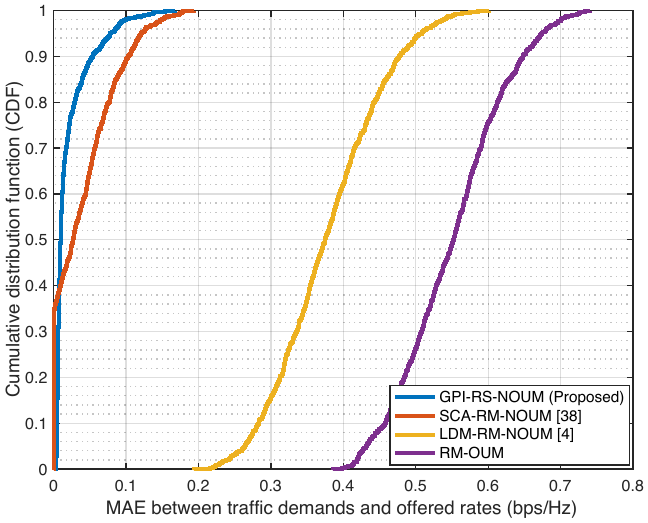}\label{result_8}}
     \vspace{-2mm}
    \\
    \subfigure[Imperfect CSIT with $N_{\sf{t}} = 16$ {(i.e., $N_{\sf{t}}^x=N_{\sf{t}}^y=4$)}]{\includegraphics[width=.7\columnwidth]{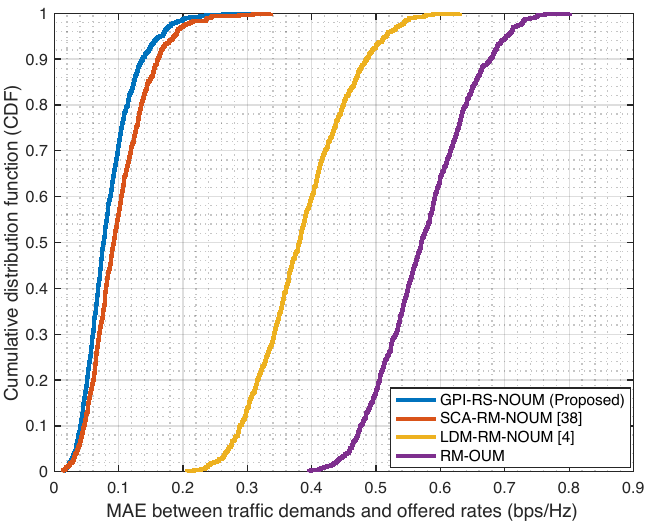}\label{result_9}}
\caption{
Comparison of CDF per MAE when the number of transmit antenna is set as $N_{\sf{t}} = 16$. The unicast and multicast traffic demands are $\mathbf{r}_{\sf target, uc} = [0.5, 0.5, 0.5, 1, 1, 1.5, 2, 2]^{\sf{T}}$ bps/Hz and $R_{\sf target, mc} = 0.5$ bps/Hz.}
\vspace{-2mm}
\end{figure}
\begin{figure}[!t]
\centering
    \subfigure[Perfect CSIT with $N_{\sf{t}} = 64$ {(i.e., $N_{\sf{t}}^x=N_{\sf{t}}^y=8$)}]{\includegraphics[width=.7\columnwidth]{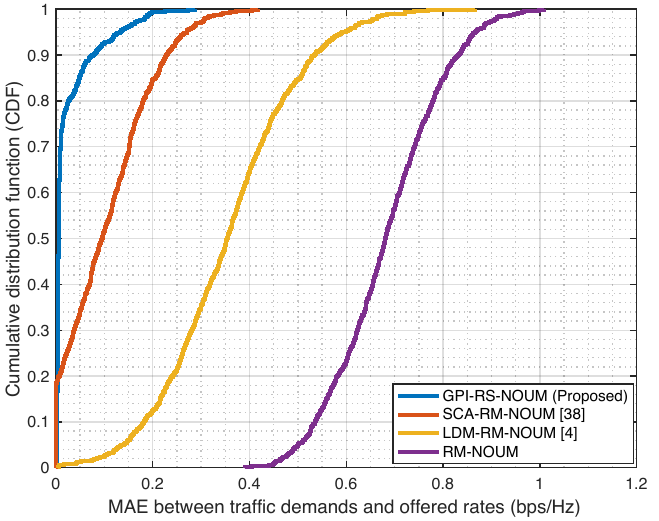}\label{result_10}}
    \vspace{-2mm}
    \\
    \subfigure[Imperfect CSIT with $N_{\sf{t}} = 64$ {(i.e., $N_{\sf{t}}^x=N_{\sf{t}}^y=8$)}]{\includegraphics[width=.7\columnwidth]{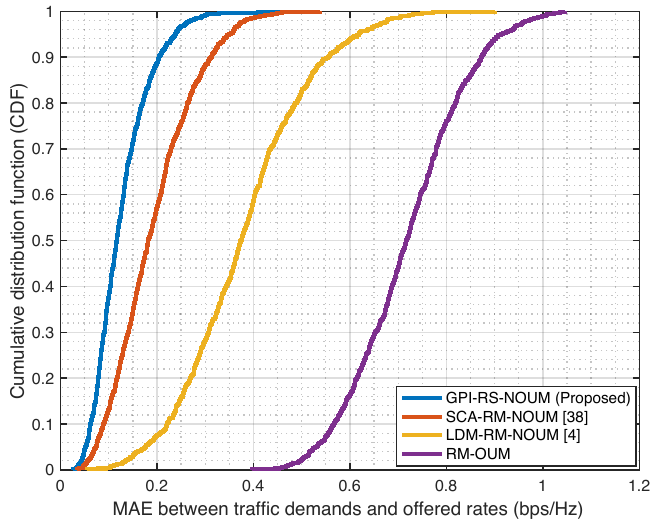}\label{result_11}}
\caption{
Comparison of CDF per MAE when the number of transmit antenna is set as $N_{\sf{t}} = 64$. The unicast and multicast traffic demands are $\mathbf{r}_{\sf target, uc} = [1, 1, 1.5, 1.5, 2, 2, 3, 3]^{\sf{T}}$ bps/Hz and $R_{\sf target, mc} = 1$ bps/Hz.}
\vspace{-2mm}
\end{figure}

\subsubsection*{{\bf Rate-matching error per Rician K-factor}\rm}

To shed light on the rate-matching error performance in comprehensive LEO SATCOM environments, we illustrate the AMAE per different Rician K-factor in Fig. \ref{result_7}. 
In Fig. \ref{result_7}, the lower value on the y-axis denotes superior performance, and the solid and dashed lines, respectively, illustrate the results under perfect and imperfect CSIT scenarios.
The number of transmit antenna $N_{\sf{t}}$ is set to be $N_{\sf{t}} = 36$ (i.e., $N_{\sf{t}}^x = N_{\sf{t}}^y = 6$), and the unicast and multicast traffic demands are respectively configured as $\mathbf{r}_{\sf target, uc} = [0.5, 0.5, 1, 1, 1.5, 2, 2.5, 2.5]^{\sf{T}}$ bps/Hz and $R_{\sf target, mc} = 1$ bps/Hz. 
The results clearly show the superior performance of the proposed {\sf{GPI-RS-NOUM}} framework in terms of AMAE, irrespective of the Rician K-factor, for both perfect and imperfect CSIT scenarios. 
This implies that the proposed scheme can reliably fulfill both unicast and multicast traffic demands, despite variations in the Rician K-factor due to elevation angle and scattering conditions of the target service area.
Furthermore, it can be observed that as the Rician K-factor increases, the performance gap between perfect and imperfect CSIT cases diminishes.
%This is because, as the value of the Rician K-factor becomes large, the randomness of the channel decreases due to the increased strength of the LOS path compared to the non-LOS paths, which matches our intuition on the LOS-dominated channel. 
{This is because, as the Rician K-factor increases, the channel relies more on the geometrical relation of satellite-to-users due to the increased strength of the LOS path compared to the non-LOS paths,  tightening our upper bound for ergodic rate expressions.}

% K-factor가 커지면 확률적으로도 mean이 커지고 variance가 0에 가까워져, channel의 randomness가 줄어든다. -> 그러면 이게 수식적으로 어떻게 LOS를 반영하냐?: sqrt(1/2 *Large Scale Power)*(1+j) = sqrt(Large Scale Power)*(cos(45) + jsin(45)) = sqrt(Large Scale Power)*(exp{45}) <- user측에서의 Doppler compensation만 이처럼 하면 된다.
% 즉 Rate expression에도 randomness가 줄어든다: 수식적으로 보려면 h 를 g*a 로 분해한다음에 log 밖에 expection씌운거 하고 <K-factor가 커질수록 이 식에서 g^2가 gamma에 가깝게 간다> g에 대해 expection 취한 식 비교해보면 된다. 
% 또한, instantaneous 한 rate하고도 같아진다 <비슷한 수치의 평균 ~= 그 중 하나의 수치 >
% Ergodic =  Upper bound = Instantaneous rate / upper bound becomes tight

\subsubsection*{{\bf Rate-matching error in different numbers of transmit antennas and various traffic demand distributions}\rm}

To consider small- and large-scale antenna configurations, we conduct a numerical analysis under the scenarios with $N_{\sf{t}} = 16$ (i.e., UPA with $4\times4$ antennas) and $N_{\sf{t}} = 64$ (i.e., UPA with $8\times8$ antennas).
Through Fig. \ref{result_8} and Fig. \ref{result_9}, we first depict the CDF for MAE under perfect and imperfect CSIT cases when the number of transmit antennas $N_{\sf{t}}$ is set as $N_{\sf{t}} = 16$.
{In such a setup, the unicast and multicast traffic demands are set to be $\mathbf{r}_{\sf{target, uc}} = [0.5, 0.5, 0.5, 1, 1, 1.5, 2, 2]^{\sf{T}}$ bps/Hz and $R_{\sf{target, mc}} = 0.5$ bps/Hz, respectively.}
As can be observed from the figures, the proposed {\sf{GPI-RS-NOUM}} framework shows superior MAE performance compared to the benchmarks for both perfect and imperfect CSIT conditions.

Subsequently, the CDF for MAE with $N_{\sf{t}} = 64$ under perfect and imperfect CSIT are illustrated in Fig. \ref{result_10} and Fig. \ref{result_11}, respectively.
The traffic demands of unicast and multicast messages are set to be $\mathbf{r}_{\sf{target, uc}} = [1, 1, 1.5, 1.5, 2, 2, 3, 3]^{\sf{T}}$ bps/Hz and $R_{\sf{target, mc}} = 1$ bps/Hz, respectively.
From the simulations, we observe that the proposed scheme achieves a reduction of around $51$ and $32$ $\%$ in the $95$-percentile of MAE compared to {\sf{SCA-RM-NOUM}} under perfect and imperfect CSIT, respectively.
Compared to {\sf{LDM-RM-NOUM}}, the proposed scheme exhibits a decrease of approximately $77.7$ and $61.6$ $\%$ under perfect and imperfect CSIT conditions.
In addition, the proposed scheme shows a decrease of around $84.7$ and $74$ $\%$ compared to {\sf{RM-OUM}} under perfect and imperfect CSIT conditions.
{These imply that the proposed {\sf{GPI-RS-NOUM}} scheme stably satisfies the non-uniform unicast traffic demands and effectively provides the required multicast service, irrespective of the number of transmit antennas.}

% 이걸 tcb로 하면 subsubsection이 이상해진다
{\subsubsection*{{\bf Convergence analysis}\rm}
Fig. \ref{result_convergence_f} and Fig. \ref{result_convergence_v} illustrate the convergence behavior of the proposed {\sf{GPI-RS-NOUM}} algorithm for random channel realizations under imperfect CSIT, where the user locations are randomly determined for each realization.
For $N_{\sf{t}}^{x}=N_{\sf{t}}^{y}=8$, the unicast and multicast traffic demands are set to $\mathbf{r}_{\sf target, uc} = [1, 1, 1.5, 1.5, 2, 2, 3, 3]^{\sf{T}}$ bps/Hz and $R_{\sf target, mc} = 1$ bps/Hz, respectively. For $N_{\sf{t}}^{x}=N_{\sf{t}}^{y}=6$, the traffic demands are set to $\mathbf{r}_{\sf target, uc} = [0.5, 0.5, 1, 1, 1.5, 2, 2.5, 2.5]^{\sf{T}}$ bps/Hz and $R_{\sf target, mc} = 1$ bps/Hz. For $N_{\sf{t}}^{x}=N_{\sf{t}}^{y}=4$, they are $\mathbf{r}_{\sf target, uc} = [0.5, 0.5, 0.5, 1, 1, 1.5, 2, 2]^{\sf{T}}$ bps/Hz and $R_{\sf target, mc} = 0.5$ bps/Hz.
It can be observed that the precoding vector $\bar{\mathbf{f}}$ converges after several iterations, while  ${\mathbf{v}}$ converges even more rapidly. These results demonstrate that the convergence of the proposed {\sf{GPI-RS-NOUM}} algorithm is well assured in various LEO SATCOM environments.}

\begin{figure}[!t]
\centering
    \subfigure[{Residual value of $\bar{\mathbf{f}}$}]{\includegraphics[width=.7\columnwidth]{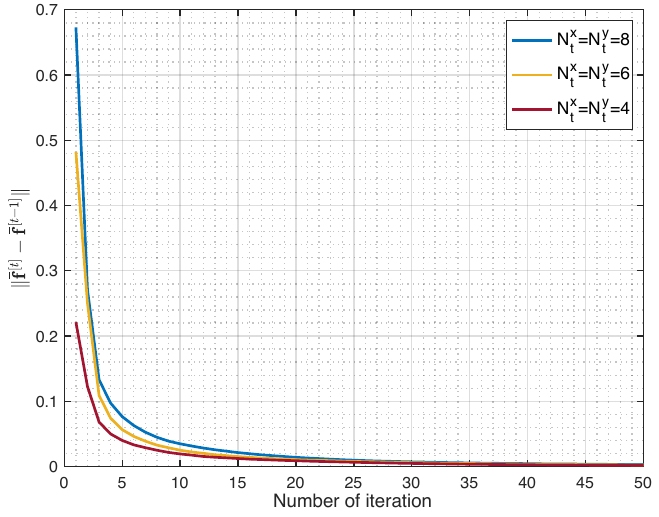}\label{result_convergence_f}}
    \vspace{-2mm}
    \\
    \subfigure[{ Residual value of ${\mathbf{v}}$}]{\includegraphics[width=.7\columnwidth]{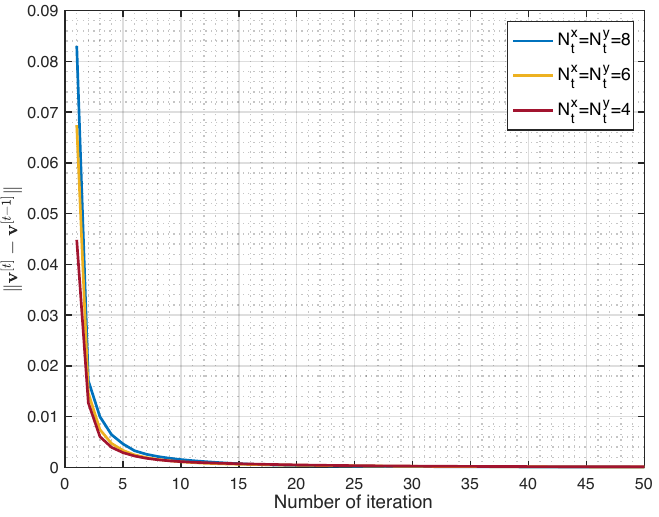}\label{result_convergence_v}}
\caption{
{ Convergence behavior of the proposed GPI-RS-NOUM algorithm under imperfect CSIT.}}
\vspace{-2mm}
\end{figure}

%\begin{table}[!t]\renewcommand{\arraystretch}{1} % Default value: 1
%\vspace{-4mm}
%\caption{{Simulation Parameters}}
%\label{Table1}
%\centering
%\begin{tabular}{|c|c|c|}
%\hlineB{3}
%\textbf{Abbreviation} & \textbf{Definition} & \textbf{Value} \\
%\hlineB{3}
%\hhline{|---|}
%{$f_{\sf c}$} & {Frequency band (carrier frequency)} & {K ($20$ GHz)}\\
%\hline
%{$B$} & {Bandwidth} & {$10$ MHz}\\
%\hline
%{$\kappa_{k}$} & {Rician factor} & {$12$ dB}\\
%\hline
%{$G_{\sf Tx}$} & {Transmit antenna gain} & {$6$ dBi}\\
%\hline
%{$G_{\sf R}$} & {User antenna gain} & {$25$ dBi}\\
%\hline
%{$T_{\sf sys}$} & {System noise temperature} & %{$150$ K}\\
%\hlineB{3}
%\end{tabular}
%\end{table}

%\vspace{-2mm}

\section{Conclusion and Future Directions}
This paper has explored an RSMA-based RM framework for NOUM transmission in LEO SATCOM in the presence of imperfect CSIT.
{We have formulated} the problem of minimizing the difference between the traffic demands and actual offered rates for both unicast and multicast messages. To address the challenges posed by the non-smoothness and non-convexity of the formulated problem, we {have employed} the LogSumExp technique to approximate the minimum common rate. Subsequently, we {have reformulated} the common rate portion as the ratio of the approximated function, transforming the problem into an unconstrained form. 
We {have demonstrated} that the first-order KKT condition of the reformulated problem is cast as an NEPv. {To efficiently compute its principal eigenvector,} the GPI-RS-NOUM algorithm {has been} introduced.
%With this reformulation, the GPI-RS-NOUM algorithm is introduced to efficiently identify the best local optimal solution. 
{Extensive numerical analysis in diverse LEO SATCOM environments has verified} the superior performance of the proposed scheme in terms of traffic demand satisfaction compared to the benchmarks.
%with variation in traffic demand distributions, user locations, scattering conditions, and the number of transmit antennas. 

{For future directions, it is promising to consider multiple receive antennas at the receiver to further improve spectral efficiency under severe path attenuation of LEO SATCOM while minimizing the increased cost and energy consumption of mobile devices. This can be accomplished by introducing analog hardware, such as controllable phase shifters, or by optimizing the selection of antennas connected to the single RF chain based on the receive antenna selection technique \cite{9050842, toka2022outage}.
Additionally, adapting machine learning approaches, such as deep learning and deep reinforcement learning, in RSMA-based rate matching precoder design is promising. Further reducing the computational complexity to achieve more practical and efficient real-world implementations, considering the constraints on onboard resources and computation in LEO satellites, is also of interest.}

%\vspace{-2mm}

\section{Appendix A: Proof of Lemma \ref{lemma1}}
{The first-order KKT condition of the problem  (\ref{P5}) to $\bar{\mathbf{f}}$ is expanded as (\ref{APPENDIX_A_3}) at the top of the next page since we have}
\begin{align}
    \label{APPENDIX_A_1}
    \frac{\partial}{\partial{\bar{\mathbf{f}}}}
    \left\{\log_{2}    \left(\frac{\bar{\mathbf{f}}^{\sf{H}}\mathbf{A}^{\sf{p}}_{k}\bar{\mathbf{f}}}{\bar{\mathbf{f}}^{\sf{H}}\mathbf{B}^{\sf{p}}_{k}\bar{\mathbf{f}}}\right)\right\}
    =
    \frac{2}{\ln2} \left( 
    \frac{\mathbf{A}^{\sf{p}}_{k}}{\bar{\mathbf{f}}^{\sf{H}}\mathbf{A}^{\sf{p}}_{k}\bar{\mathbf{f}}}
    - \frac{\mathbf{B}^{\sf{p}}_{k}}{\bar{\mathbf{f}}^{\sf{H}}\mathbf{B}^{\sf{p}}_{k}\bar{\mathbf{f}}}
    \right) \bar{\mathbf{f}}
\end{align}
and
\begin{align}
    \label{APPENDIX_A_2}
    & \frac{\partial}{\partial{\bar{\mathbf{f}}}}
    \left\{\log\left( \frac{1}{K} \sum_{i=1}^{K}\exp\left(-\frac{1}{\alpha}
\log_{2}\left(\frac{\bar{\mathbf{f}}^{\sf{H}}\mathbf{A}_{i}^{\sf{c}}\bar{\mathbf{f}}}{\bar{\mathbf{f}}^{\sf{H}}\mathbf{B}_{i}^{\sf{c}}\bar{\mathbf{f}}}\right)\right)\right)^{-\alpha}\right\} \! = \! \frac{2}{\ln{2}} \times
    \nonumber \\
    &
    \sum_{i=1}^{K}\left({\frac{\exp\left(-\frac{1}{\alpha} \log_{2} \left( \frac{\bar{\mathbf{f}}^{\sf{H}}\mathbf{A}_{i}^{\sf{c}}\bar{\mathbf{f}}}{\bar{\mathbf{f}}^{\sf{H}}\mathbf{B}_{i}^{\sf{c}}\bar{\mathbf{f}}} \right)
    \right)}{\sum_{\ell=1}^{K}{\exp\left(-\frac{1}{\alpha} \log_{2} \left( \frac{\bar{\mathbf{f}}^{\sf{H}}\mathbf{A}_{\ell}^{\sf{c}}\bar{\mathbf{f}}}{\bar{\mathbf{f}}^{\sf{H}}\mathbf{B}_{\ell}^{\sf{c}}\bar{\mathbf{f}}} \right)
    \right)}}} 
    \left(
    \frac{{\mathbf{A}_{i}^{\sf{c}}}}{\bar{\mathbf{f}}^{\sf{H}}
    {\mathbf{A}_{i}^{\sf{c}}}{\bar{\mathbf{f}}}} -
    \frac{{\mathbf{B}_{i}^{\sf{c}}}}{\bar{\mathbf{f}}^{\sf{H}}
    {\mathbf{B}_{i}^{\sf{c}}}{\bar{\mathbf{f}}}}
    \right)
    \right)
    \nonumber \\
    & = \frac{2}{\ln{2}} \left( \mathbf{L}_{\sf{A}}(\bar{\mathbf{f}}) - \mathbf{L}_{\sf{B}}(\bar{\mathbf{f}}) \right) {\bar{\mathbf{f}}}.
\end{align} 

Subsequently, defining $\mathbf{A}(\bar{\mathbf{f}}, \mathbf{v})$, $\mathbf{B}(\bar{\mathbf{f}}, \mathbf{v})$, $\lambda(\bar{\mathbf{f}}, \mathbf{v})$, $\lambda_{\sf{num}}(\bar{\mathbf{f}}, \mathbf{v})$, and 
$\lambda_{\sf{den}}(\bar{\mathbf{f}}, \mathbf{v})$ as (\ref{A_kkt}), (\ref{B_kkt}), (\ref{lambda}), (\ref{lambda_num}), and (\ref{lambda_den}), respectively, the equation (\ref{APPENDIX_A_3}) is rearranged as
\begin{figure*}[!t]
%\vspace{-2.5mm}
%\noindent\rule{\textwidth}{.5pt}%\vskip3pt
\small
\begin{align}
    \label{APPENDIX_A_3}
    \frac{\partial{f(\bar{\mathbf{f}}, \mathbf{v})}} {\partial{\bar{\mathbf{f}}}} 
     & = \sum_{j=1}^{K} 
    \Bigg[
    -R_{{\sf{target}}, j} 
    \left\{ 
    \frac{\mathbf{A}^{\sf{p}}_{j}\bar{\mathbf{f}}}{\bar{\mathbf{f}}^{\sf{H}}\mathbf{A}^{\sf{p}}_{j}\bar{\mathbf{f}}}
    - \frac{\mathbf{B}^{\sf{p}}_{j}\bar{\mathbf{f}}}{\bar{\mathbf{f}}^{\sf{H}}\mathbf{B}^{\sf{p}}_{j}\bar{\mathbf{f}}}
    \right\}
    %%%%%%%%%%%%%%%%%%%%%%%%%%%%%%%%%%%%%%%%%%%%%%%%%%
    - R_{{\sf{target}}, j}  \frac{\mathbf{v}^{\sf{H}}\mathbf{E}_{j}\mathbf{v}}{\mathbf{v}^{\sf{H}}\mathbf{v}} 
    \left\{
    \mathbf{L}_{\sf{A}}(\bar{\mathbf{f}}) - \mathbf{L}_{\sf{B}}(\bar{\mathbf{f}})
    \right\}
    + 
     \log_{2}    \left(\frac{\bar{\mathbf{f}}^{\sf{H}}\mathbf{A}_{j}^{\sf{p}}\bar{\mathbf{f}}}{\bar{\mathbf{f}}^{\sf{H}}\mathbf{B}_{j}^{\sf{p}}\bar{\mathbf{f}}}\right)
    \left\{ 
    \frac{\mathbf{A}^{\sf{p}}_{j}\bar{\mathbf{f}}}{\bar{\mathbf{f}}^{\sf{H}}\mathbf{A}^{\sf{p}}_{j}\bar{\mathbf{f}}}
    - \frac{\mathbf{B}^{\sf{p}}_{j}\bar{\mathbf{f}}}{\bar{\mathbf{f}}^{\sf{H}}\mathbf{B}^{\sf{p}}_{j}\bar{\mathbf{f}}}
    \right\}
    %%%%%%%%%%%%%%%%%%%%%%%%%%%%%%%%%%%%%%%%%%%%%%%%%%
    \nonumber \\
    & + 
    \left( \frac{\mathbf{v}^{\sf{H}}\mathbf{E}_{j}\mathbf{v}}{\mathbf{v}^{\sf{H}}\mathbf{v}} 
    \right)^2
    \log\left( \frac{1}{K} \sum_{i=1}^{K}\exp\left(-\frac{1}{\alpha}
\log_{2}\left(\frac{\bar{\mathbf{f}}^{\sf{H}}\mathbf{A}_{i}^{\sf{c}}\bar{\mathbf{f}}}{\bar{\mathbf{f}}^{\sf{H}}\mathbf{B}_{i}^{\sf{c}}\bar{\mathbf{f}}}\right)\right)\right)^{-\alpha} 
\!\!\!\!\!\!
\left\{\mathbf{L}_{\sf{A}}(\bar{\mathbf{f}}) - \mathbf{L}_{\sf{B}}(\bar{\mathbf{f}})\right\}
 %%%%%%%%%%%%%%%%%%%%%%%%%%%%%%%%%%%%%%%%%%%%%%%%%
+ \frac{\mathbf{v}^{\sf{H}}\mathbf{E}_{j}\mathbf{v}}{\mathbf{v}^{\sf{H}}\mathbf{v}}
     \log_{2}    \left(\frac{\bar{\mathbf{f}}^{\sf{H}}\mathbf{A}_{j}^{\sf{p}}\bar{\mathbf{f}}}{\bar{\mathbf{f}}^{\sf{H}}\mathbf{B}_{j}^{\sf{p}}\bar{\mathbf{f}}}\right)
     \left\{\mathbf{L}_{\sf{A}}(\bar{\mathbf{f}}) - \mathbf{L}_{\sf{B}}(\bar{\mathbf{f}})\right\}   %%%%%%%%%%%%%%%%%%%%%%%%%%%%%%%%%%%%%%%%%%%%%%%%%
\nonumber \\
    & 
     + \frac{\mathbf{v}^{\sf{H}}\mathbf{E}_{j}\mathbf{v}}{\mathbf{v}^{\sf{H}}\mathbf{v}}
     \log\left( \frac{1}{K} \sum_{i=1}^{K}\exp\left(-\frac{1}{\alpha} 
\log_{2}\left(\frac{\bar{\mathbf{f}}^{\sf{H}}\mathbf{A}_{i}^{\sf{c}}\bar{\mathbf{f}}}{\bar{\mathbf{f}}^{\sf{H}}\mathbf{B}_{i}^{\sf{c}}\bar{\mathbf{f}}}\right)\right)\right)^{-\alpha}
\!\!\!\!\! \cdot 
\left\{ \frac{\mathbf{A}^{\sf{p}}_{j}\bar{\mathbf{f}}}{\bar{\mathbf{f}}^{\sf{H}}\mathbf{A}^{\sf{p}}_{j}\bar{\mathbf{f}}}
    - \frac{\mathbf{B}^{\sf{p}}_{j}\bar{\mathbf{f}}}{\bar{\mathbf{f}}^{\sf{H}}\mathbf{B}^{\sf{p}}_{j}\bar{\mathbf{f}}}
    \right\}
    \Bigg]
    %%%%%%%%%%%%%%%%%%%%%%%%%%%%%%%%%%%%%%%%%%%%%%%%%%%
    + \eta_{\sf{mc}} \Bigg[ 
    -R_{{\sf{target, mc}}}  \frac{\mathbf{v}^{\sf{H}}\mathbf{E}_{K+1}\mathbf{v}}{\mathbf{v}^{\sf{H}}\mathbf{v}} 
    \left\{\mathbf{L}_{\sf{A}}(\bar{\mathbf{f}}) - \mathbf{L}_{\sf{B}}(\bar{\mathbf{f}})\right\}
    %%%%%%%%%%%%%%%%%%%%%%%%%%%%%%%%%%%%%%%%%%%%%%%%
    \nonumber \\
    & + 
    \left( \frac{\mathbf{v}^{\sf{H}}\mathbf{E}_{K+1}\mathbf{v}}{\mathbf{v}^{\sf{H}}\mathbf{v}} 
    \right)^2
    \log\left( \frac{1}{K} \sum_{i=1}^{K}\exp\left(-\frac{1}{\alpha}
\log_{2}\left(\frac{\bar{\mathbf{f}}^{\sf{H}}\mathbf{A}_{i}^{\sf{c}}\bar{\mathbf{f}}}{\bar{\mathbf{f}}^{\sf{H}}\mathbf{B}_{i}^{\sf{c}}\bar{\mathbf{f}}}\right)\right)\right)^{-\alpha} 
\!\!\!\!\!\!  
\left\{\mathbf{L}_{\sf{A}}(\bar{\mathbf{f}}) - \mathbf{L}_{\sf{B}}(\bar{\mathbf{f}})\right\} 
    \Bigg] = \mathbf{0}
\end{align}
\noindent\rule{\textwidth}{.5pt}%\vskip3pt
%\vspace{-2mm}
\end{figure*}
\begin{align}
\label{APPENDIX_A_5}
\mathbf{A}(\bar{\mathbf{f}}, \mathbf{v}) \bar{\mathbf{f}} = \lambda(\bar{\mathbf{f}}, \mathbf{v}) \mathbf{B}(\bar{\mathbf{f}}, \mathbf{v}) \bar{\mathbf{f}}.
\end{align}
We note that since $\mathbf{A}(\bar{\mathbf{f}}, \mathbf{v})$ and $\mathbf{B}(\bar{\mathbf{f}}, \mathbf{v})$ are expressed as the summation of positive definite (PD) matrices and positive semi-definite (PSD) matrices, both $\mathbf{A}(\bar{\mathbf{f}}, \mathbf{v})$ and $\mathbf{B}(\bar{\mathbf{f}}, \mathbf{v})$ are PD matrices; thus, $\mathbf{A}(\bar{\mathbf{f}}, \mathbf{v})$ and $\mathbf{B}(\bar{\mathbf{f}}, \mathbf{v})$ are invertible matrices.
Thanks to this feature, (\ref{APPENDIX_A_5}) can be rewritten as
\begin{align}
\label{APPENDIX_A_6}
     {\mathbf{B}(\bar{\mathbf{f}}, \mathbf{v})}^{-1}  
     {\mathbf{A}(\bar{\mathbf{f}}, \mathbf{v})}\bar{\mathbf{f}} = \lambda(\bar{\mathbf{f}}, \mathbf{v})\bar{\mathbf{f}},
\end{align}
and this completes the proof.

%\vspace{-2mm}

\section{Appendix B: Proof of Lemma \ref{lemma2}}
{The first-order KKT condition of the problem  (\ref{P5}) to ${\mathbf{v}}$ is expanded as (\ref{APPENDIX_B_2}) at the top of this page since we have}
\begin{align}
    \label{APPENDIX_B_1}
    \frac{\partial}{\partial{{\mathbf{v}}}}
    \left\{\frac{\mathbf{v}^{\sf{H}} \mathbf{E}_{k} \mathbf{v}}{\mathbf{v}^{\sf{H}} \mathbf{v}}\right\} 
    & =
    \frac{2\mathbf{E}_{k}\mathbf{v}}{\mathbf{v}^{\sf{H}} \mathbf{v}}
    -  \frac{2\mathbf{v}^{\sf{H}}\mathbf{E}_{k}\mathbf{v}}{(\mathbf{v}^{\sf{H}}\mathbf{v})^{2}} \mathbf{v}
    \nonumber \\
    & =
    2 \left( \frac{\mathbf{E}_{k}}{\mathbf{v}^{\sf{H}} \mathbf{v}}
    -\frac{\mathbf{v}^{\sf{H}}\mathbf{E}_{k}\mathbf{v}}{(\mathbf{v}^{\sf{H}}\mathbf{v})^{2}} \mathbf{I} \right) \mathbf{v}.
\end{align}

Subsequently, defining $\mathbf{D}(\bar{\mathbf{f}}, \mathbf{v})$, $\mathbf{E}(\bar{\mathbf{f}}, \mathbf{v})$, $\lambda(\bar{\mathbf{f}}, \mathbf{v})$, $\lambda_{\sf{num}}(\bar{\mathbf{f}}, \mathbf{v})$, and 
$\lambda_{\sf{den}}(\bar{\mathbf{f}}, \mathbf{v})$ as (\ref{D_kkt}), (\ref{E_kkt}), (\ref{lambda}), (\ref{lambda_num}), and (\ref{lambda_den}), respectively, the equation (\ref{APPENDIX_B_2}) is rearranged as
\begin{figure*}[!t]
%\vspace{-2.5mm}
\small
%\noindent\rule{\textwidth}{.5pt}%\vskip3pt
\begin{align}
    \label{APPENDIX_B_2}
    \frac{\partial{f(\bar{\mathbf{f}}, \mathbf{v})}} {\partial{{\mathbf{v}}}} 
     & =
\sum_{j=1}^{K} 
        \Bigg[ 
        -R_{{\sf{target}}, j} 
\left\{ 
 \frac{\mathbf{E}_{j}\mathbf{v}}{\mathbf{v}^{\sf{H}} \mathbf{v}}
    -
\frac{\mathbf{v}^{\sf{H}}\mathbf{E}_{j}\mathbf{v}}{(\mathbf{v}^{\sf{H}}\mathbf{v})^{2}} \mathbf{v}
\right\}
%%%%%%%%%%%%%%%%%%%%%%%%%%%%%%%%%%%%%%%%%%%%%%%%%%%%%%%%%%%%%
+ \log\left( \frac{1}{K} \sum_{i=1}^{K}\exp\left(-\frac{1}{\alpha}
\log_{2}\left(\frac{\bar{\mathbf{f}}^{\sf{H}}\mathbf{A}_{i}^{\sf{c}}\bar{\mathbf{f}}}{\bar{\mathbf{f}}^{\sf{H}}\mathbf{B}_{i}^{\sf{c}}\bar{\mathbf{f}}}\right)\right)\right)^{-\alpha}  
\!\!\!\!\!\! \cdot
\frac{\mathbf{v}^{\sf{H}}
\mathbf{E}_{j}\mathbf{v}}{\mathbf{v}^{\sf{H}}\mathbf{v}} 
\left\{ 
 \frac{\mathbf{E}_{j}\mathbf{v}}{\mathbf{v}^{\sf{H}} \mathbf{v}}
    -
 \frac{\mathbf{v}^{\sf{H}}\mathbf{E}_{j}\mathbf{v}}{(\mathbf{v}^{\sf{H}}\mathbf{v})^{2}} \mathbf{v}
\right\}
\nonumber \\
& + \log_{2}    \left(\frac{\bar{\mathbf{f}}^{\sf{H}}\mathbf{A}_{j}^{\sf{p}}\bar{\mathbf{f}}}{\bar{\mathbf{f}}^{\sf{H}}\mathbf{B}_{j}^{\sf{p}}\bar{\mathbf{f}}}\right) \left\{ 
 \frac{\mathbf{E}_{j}\mathbf{v}}{\mathbf{v}^{\sf{H}} \mathbf{v}}
    -
\frac{\mathbf{v}^{\sf{H}}\mathbf{E}_{j}\mathbf{v}}{(\mathbf{v}^{\sf{H}}\mathbf{v})^{2}} \mathbf{v}
\right\} \Bigg]
%%%%%%%%%%%%%%%%%%%%%%%%%%%%%%%%%%%%%%%%%%%%%%%%%%%%%%%%%%%%%
 + \eta_{\sf{mc}} \Bigg[ -R_{\sf{target, mc}} \left\{ 
 \frac{\mathbf{E}_{K+1}\mathbf{v}}{\mathbf{v}^{\sf{H}} \mathbf{v}}
    -
\frac{\mathbf{v}^{\sf{H}}\mathbf{E}_{K+1}\mathbf{v}}{(\mathbf{v}^{\sf{H}}\mathbf{v})^{2}} \mathbf{v}
\right\}
%%%%%%%%%%%%%%%%%%%%%%%%%%%%%%%%%%%%%%%%%%%%%%%%%%%%%%%%%%%%%
\nonumber \\
& + \log\left( \frac{1}{K}  \sum_{i=1}^{K}\exp\left(-\frac{1}{\alpha}
\log_{2}\left(\frac{\bar{\mathbf{f}}^{\sf{H}}\mathbf{A}_{i}^{\sf{c}}\bar{\mathbf{f}}}{\bar{\mathbf{f}}^{\sf{H}}\mathbf{B}_{i}^{\sf{c}}\bar{\mathbf{f}}}\right)\right)\right)^{-\alpha} 
\!\!\!\!\!\! \cdot
\frac{\mathbf{v}^{\sf{H}}\mathbf{E}_{K+1}\mathbf{v}}{\mathbf{v}^{\sf{H}}\mathbf{v}} 
\left\{ 
 \frac{\mathbf{E}_{K+1}\mathbf{v}}{\mathbf{v}^{\sf{H}} \mathbf{v}}
    -
\frac{\mathbf{v}^{\sf{H}}\mathbf{E}_{K+1}\mathbf{v}}{(\mathbf{v}^{\sf{H}}\mathbf{v})^{2}} \mathbf{v}
\right\} \Bigg] = \mathbf{0}
\end{align}
\noindent\rule{\textwidth}{.5pt}%\vskip3pt
%\vspace{-2mm}
\end{figure*}
\begin{align}
\label{APPENDIX_B_3}
\mathbf{D}(\bar{\mathbf{f}}, \mathbf{v}) {\mathbf{v}} = \lambda(\bar{\mathbf{f}}, \mathbf{v}) \mathbf{E}(\bar{\mathbf{f}}, \mathbf{v}) {\mathbf{v}}.
\end{align}
We note that since $\mathbf{D}(\bar{\mathbf{f}}, \mathbf{v})$ and $\mathbf{E}(\bar{\mathbf{f}}, \mathbf{v})$ are expressed as the summation of PD matrices and PSD matrices, both $\mathbf{D}(\bar{\mathbf{f}}, \mathbf{v})$ and $\mathbf{E}(\bar{\mathbf{f}}, \mathbf{v})$ are PD matrices; thus, $\mathbf{D}(\bar{\mathbf{f}}, \mathbf{v})$ and $\mathbf{E}(\bar{\mathbf{f}}, \mathbf{v})$ are invertible matrices.
Thanks to this feature, (\ref{APPENDIX_A_5}) can be rewritten as
\begin{align}
\label{APPENDIX_B_4}
     {\mathbf{E}(\bar{\mathbf{f}}, \mathbf{v})}^{-1}  
     {\mathbf{D}(\bar{\mathbf{f}}, \mathbf{v})}{\mathbf{v}} = \lambda(\bar{\mathbf{f}}, \mathbf{v}){\mathbf{v}},
\end{align}
and this completes the proof.

\bibliographystyle{IEEEtran}

\bibliography{jhseong_reff}
\end{document}